\documentclass{IEEEtran}
\usepackage[cmex10]{amsmath}
\usepackage{amsfonts,amssymb,amsbsy,empheq}
\usepackage{graphicx,comment}

\newtheorem{theorem}{Theorem}

\newtheorem{example}{Example}
\newtheorem{lemma}[theorem]{Lemma}

\newtheorem{remark}{Remark}
\newenvironment{proof}[1][Proof]{{\em #1:} }{\ \rule{0.5em}{0.5em}}

\newcommand{\sinc}{{\rm sinc}}

\newcommand{\E}[1]{{\rm E}\left[{#1}\right]}

\DeclareMathOperator{\sech}{sech}

\begin{document}

\author{
%\IEEEauthorblockN{Authors \\ \today}
%\IEEEauthorblockN{Gerhard Kramer \\ \today}
%\IEEEauthorblockA{
%Electrical Engineering Department \\
%abc@tum.de
%}
%\and
\IEEEauthorblockN{Gerhard Kramer}
%\IEEEauthorblockA{
%Institute for Communications Engineering \\
%Department of Electrical and Computer Engineering \\
%Technical University of Munich \\
%Arcisstra{\ss}e 21, 80333 M\"unchen, Germany
%gerhard.kramer@tum.de
%}
\thanks{Date of current version \today.

G. Kramer was supported in part by an Alexander von Humboldt
Professorship endowed by the German Federal Ministry of Education and
Research, and in part by the DFG under Grant KR 3517/8-1.
This paper was presented in part at the 2017
Conference on Lasers and Electro-Optics Pacific Rim.

G. Kramer is with the Institute for Communications Engineering,
Department of Electrical and Computer Engineering,
Technical University of Munich, 80333 Munich, Germany (e-mail:
gerhard.kramer@tum.de).

\vspace{1cm}
}
}

%=============================================================
\title{Autocorrelation Function for Dispersion-Free
Fiber Channels with Distributed Amplification}

\maketitle

\begin{abstract}
Optical fiber signals with high power exhibit spectral broadening that
seems to limit capacity.
To study spectral broadening, the autocorrelation function
of the output signal given the input signal is derived for a simplified fiber model
that has zero dispersion, distributed optical amplification (OA), and
idealized spatial noise processes.
The autocorrelation function is used to upper bound the output power
of bandlimited or time-resolution limited receivers, and thereby to
bound spectral broadening and the capacity of receivers with thermal noise.
The output power scales at most as the square-root of the launch
power, and thus capacity scales at most as one-half the logarithm
of the launch power. The propagating signal bandwidth
scales at least as the square-root of the launch power.
However, in practice the OA bandwidth should exceed the
signal bandwidth to compensate attenuation. Hence, there is a
launch power threshold beyond which the fiber model loses practical relevance. 
Nevertheless, for the mathematical model an upper bound on capacity
is developed when the OA bandwidth scales as the square-root of the
launch power, in which case capacity scales
at most as the inverse fourth root of the launch power.
\end{abstract}


%=============================================================
%=============================================================
\section{Introduction}
\label{sec:intro}
Optical fiber is a medium that exhibits frequency-dependent dispersion, nonlinearity,
and noise, where one source of noise is distributed optical amplification (OA)~\cite{ekwfg-JLT10}.
One obstacle to understand the capacity of optical fiber is that combining nonlinearity and
OA causes \emph{spectral broadening} that seems difficult to characterize. To make progress,
we study a simplified fiber model that hopefully retains the essential features of spectral
broadening. In particular, we neglect dispersion and the frequency dependence of
nonlinearity, and we consider distributed noise processes with idealized statistics.

There are two existing approaches to analyze dispersion-free fiber with OA.
The first is by Mecozzi~\cite{Mecozzi-94} who derived the \emph{per-sample} statistics
of the channel, including the channel conditional probability distribution.
Turitsyn et al.~\cite{Turitsyn-Derevyanko-Yurkevich-Turitsyn-PRL03}
and Yousefi and Kschischang~\cite{Yousefi-Kschischang-IT11}
rederive this distribution with other methods. They further argue that, for
large launch power $P$, the \emph{per-sample} capacity is the
same as the capacity of an additive white Gaussian noise (AWGN) channel
with intensity modulation and a direct detection receiver, i.e.,
capacity grows as $\frac{1}{2} \log P$ for large $P$.
Refined results appear in~\cite{Mecozzi-04,Terekhov-etal-A16,Fahs-etal-A17}.

A second approach considers the entire received \emph{waveform}.
Tang~\cite{Tang-JLT01a,Tang-JLT01b} studied the auto- and crosscorrelation
functions of the channel input and output signals when the input signals are
Gaussian and stationary, and in particular when the input signals are sinc pulses
with complex and circularly symmetric Gaussian modulation. The autocorrelation
function defines the signal power spectral density (PSD) that lets one study
spectral broadening. Tang used the PSD to evaluate Pinsker's  capacity lower
bound~\cite{Pinsker-64} for wavelength division multiplexing (WDM) and
per-channel receivers without cooperation.

%=============================================================
\subsection{Limitations of the Per-Sample Model}
\label{subsec:per-sample-vs-multi-sample}
The per-sample model is attractive because one has closed-form expressions
for the statistics. Furthermore, one might suspect that the per-sample capacity
predicts what is ultimately possible with high-speed receivers.
However, we argue that  the model has several limitations and pitfalls.

First, the per-sample statistics do not capture spectral broadening,
and this tempts one to consider only the \emph{launch} signal bandwidth rather than
the \emph{propagating}
signal bandwidth\footnote{There are many reasonable definitions for bandwidth. We
use a common one, namely the length of the frequency range centered at the
carrier frequency that contains a specified fraction of the signal power.}.
The propagating signal bandwidth $W$ grows with the launch power
$P$ and a practical requirement is that the OA bandwidth $B$ exceed $W$ to
compensate attenuation, i.e., one requires $B \ge W$.
However, we show that there is a $P$ beyond which $B$ does not
exceed $W$ and the model loses practical relevance.\footnote{The short
article~\cite{Wei-Plant-A06} also argues that the model of~\cite{Mecozzi-94}
may be impractical for large $P$. The arguments are based on
empirical observations concerning spectral broadening and signal-noise mixing.}
The growth of $W$ is due to signal-noise mixing that cannot be controlled by
waveform design. %, other than by reducing $P$.

Second, a per-sample receiver has infinite bandwidth while practical receivers
are bandlimited. In other words, a per-sample analysis takes limits in a particular order:\ first
the receiver bandwidth is made infinite and then $P$ is made large. However, for a
given system the receiver bandwidth is fixed, and changing the order of limits
(first $P$ is made large) can change the capacity scaling. 

Third, the per-sample model ignores correlations in the received waveform,
and this can lead to suboptimal receivers. In fact, we show that a \emph{three-sample}
receiver achieves unbounded capacity for \emph{any} $P$ for the model studied
in~\cite{Mecozzi-94,Turitsyn-Derevyanko-Yurkevich-Turitsyn-PRL03,Yousefi-Kschischang-IT11}.
The per-sample rate $\frac{1}{2} \log P$ thus underestimates capacity.\footnote{
The potential for capacity increase was noted in~\cite[Sec.~VIII]{Yousefi-Kschischang-IT11}
but without recognizing the extent of the effect, i.e., that the noise model is unreasonable.
Hence, the main conclusions in~\cite[Sec.~VIII]{Yousefi-Kschischang-IT11}
should be treated with caution, namely that the capacity of dispersion-free
fiber grows as $\frac{1}{2}\log P$, and that a potential peak of
spectral efficiency curves is due to deterministic effects only, and not due
to signal-noise mixing.}
This issue will also appear for the nonlinear Schr\"odinger equation
(NLSE) with dispersion, nonlinearity, and distributed noise.

The result may be understood as follows:\ the noise in the model of~\cite{Mecozzi-94}
has limited bandwidth while a per-sample receiver has infinite bandwidth.
Thus, by sending signal energy in the noise-free spectrum one
achieves large rate, cf.~\cite[Thm.~5]{Wyner-BSTJ66}. The reader may expect
that an obvious fix is to add white (thermal or electronic) noise to the channel
or receiver models. However, the per-sample capacity is then zero. This conundrum
shows that reasonable and precise noise models, device models, and spectral
analyses are needed when analyzing capacity, e.g., see~\cite[Sec.~IX.A-B]{ekwfg-JLT10}.

Based on these observations, we conclude that one should study the waveform
model, and not only the per-sample model. More precisely, we study
\emph{filter-and-sample} models where the receiver projects its input waveform
onto orthogonal functions, e.g., time-shifted sinc pulses or time-shifted rectangular pulses. 
The motivation for considering these two sets of pulses is to include the engineering
constraints of finite bandwidth and/or finite time resolution. We further model the
projections as being corrupted by thermal noise. We then proceed to study
\emph{two-sample} statistics to compute autocorrelation functions, PSDs, and receiver
power levels. Finally, we study OA bandwidth that grows with the propagating signal
bandwidth to better understand spectral broadening.

We remark that, to permit analysis, we make several idealizations in
addition to neglecting dispersion and the frequency dependence of nonlinearity.
For example, we idealize the spatial noise statistics at two different time instances
to be jointly Wiener. The resulting model is subtly different than the one studied
in~\cite{Mecozzi-94,Turitsyn-Derevyanko-Yurkevich-Turitsyn-PRL03,Yousefi-Kschischang-IT11,Mecozzi-04,Terekhov-etal-A16,Fahs-etal-A17}, and we discuss these differences in
Section~\ref{subsec:dispersion-free-model} and Appendix~\ref{app:raman}.
The model lets us show that spectral efficiency decreases rapidly
with increasing $P$ for \emph{any} launch signal and for large $P$.
A similar result was shown for WDM in optically-routed networks in~\cite{ekwfg-JLT10}.

%=============================================================
\subsection{Organization}
\label{subsec:organization}
This paper is organized as follows. Section~\ref{sec:prelim} describes notation,
second order statistics, AWGN channels and their capacities, and certain
hyperbolic functions. Section~\ref{sec:fiber-models} describes the fiber and OA noise models
under study. Section~\ref{sec:rx-models} reviews several receiver models, including
per-sample models, filter-and-sample models, bandlimited receivers, and time-resolution
limited receivers. Section~\ref{sec:auto} states our main result:\ the autocorrelation
function for a dispersion-free fiber model with distributed OA and idealized noise statistics.
Section~\ref{sec:rectangular} studies the autocorrelation function for rectangular pulses.
Section~\ref{sec:P-E-scaling} develops upper bounds on the output power and
energy of the receivers, as well as lower bounds on the propagating signal bandwidth.
Section~\ref{sec:capacity} uses the power bounds to 
develop capacity upper bounds. Section~\ref{sec:conclusions} concludes the paper.
The appendices provide supporting material, including a review of
theory from \cite{Mecozzi-94}.

%=============================================================
%=============================================================
\section{Preliminaries}
\label{sec:prelim}
%=============================================================
\subsection{Basic Notation}
\label{subsec:notation}
This section describes basic notation that we use for signals and random variables.
For convenience, further selected notation is listed in Table~\ref{table:notation} at the end
of the document.

We study signals $u(z,t)$ where $z$ is a spatial variable and $t$ is a time variable.
The position $z=0$ is where the information-bearing signal $u(0,\cdot)$ is launched.
To shorten notation, we often write $u_z(t)=u(z,t)$, and even drop
$z$ if the position is clear from the context. For example, we often write
$u(t)$ for $u(z,t)=u_z(t)$. We also often drop the time indices for convenience,
e.g., we write $u_0=u_0(t)$ and $u_0'=u_0(t')$.

We write random variables with uppercase letters and realizations
of random variables with the corresponding lowercase letters. For example,
we follow~\cite{Mecozzi-94} and study the statistics of the random variables
$U(z,t)$ for different $t$ when conditioned on the event $U_0(\cdot)=u_0(\cdot)$.
The expectation of $X$ is denoted by $\E{X}$, and the conditional expectation
based on the event $Y=y$ is denoted by $\E{X | Y=y}$.
%The mutual information of two random variables $X$ and $Y$ is written as $I(X;Y)$.
%The differential entropy of $Y$ without and with conditioning on $X$ is denoted by
%$h(Y)$ and $h(Y|X)$, respectively.

The notation $y^*$ refers to the complex conjugate of $y$.
$\Re(y)$ and $\Im(y)$ are the respective real and imaginary parts of $y$.
The function $1(\cdot)$ is the indicator function that takes on the value 1 if its
argument is true, and is otherwise 0. The function $\delta(\cdot)$ is the
Dirac-$\delta$ operator, and we write $\sinc(y)=\sin(\pi y)/(\pi y)$ with $\sinc(0) = 1$.
The functions $I_0(\cdot)$ and $I_1(\cdot)$ are the modified Bessel functions of the
first kind of orders 0 and 1, respectively. We write $Q \lesssim P^x$ if
$\lim_{P\rightarrow\infty} \left[\log Q/\log P\right] \le x$, and similarly for $Q \gtrsim P^x$.

%=============================================================
%=============================================================
\subsection{Autocorrelation Functions and Power Spectral Densities}
\label{subsec:autocorr-psd}
We study the conditional and average autocorrelation functions
\begin{align}
  & A_z(t,t') = \E{ \left. U_z(t) \, U_z(t')^* \right| U_0(\cdot)=u_0(\cdot) } \label{eq:autocorr-cond} \\
  & \bar{A}_z(t,t') = \E{A_z(t,t')} = \E{ U_z(t) \, U_z(t')^*}
\end{align}
where the bar above $A_z(t,t')$ specifies that we have taken the expectation with
respect to the launch signal $U_0(\cdot)$.
As described above, we often drop the subscript $z$ for convenience, e.g., we write
$A(t,t')$ for $A_z(t,t')$. A basic property of the autocorrelation function is
$A(t,t')=A(t',t)^*$.

The PSD is defined as
\begin{align}
\bar{\mathcal{P}}(f) & = \lim_{T\rightarrow\infty} \bar{\mathcal{P}}(f,T)
\label{eq:PSD-def}
\end{align}
assuming the limit exists, where
\begin{align}
\bar{\mathcal{P}}(f,T) & = \frac{1}{T} \, {\rm E} \left[ 
\left| \int_{-T/2}^{T/2} U(t) e^{-j 2 \pi f t} \, dt \right|^2 \right] \nonumber \\
& = \frac{1}{T} \, \int_{-T/2}^{T/2} \int_{-T/2}^{T/2} \bar{A}(t,t') e^{- j 2 \pi f (t-t')} \, dt' \, dt .
\label{eq:PSD-T-def}
\end{align}

%=============================================================
%=============================================================
\subsection{Pulse Amplitude Modulation}
\label{subsec:cyclostationary}
We will sometimes consider pulse-amplitude modulation (PAM) with time period $T_s$
for which the launch signals are
\begin{align}
   u_0(t) = \sum_{k} x_k \, g(t-kT_s) \label{eq:PAM}
\end{align}
where the $x_k$ are complex-valued modulation symbols and
$g(\cdot)$ is a pulse shape with unit energy.
If the $x_k$ are realizations of a stationary discrete-time process, then
the signals \eqref{eq:PAM} are cyclostationary~\cite[p.~70]{Proakis-Salehi-5}.
That is, for all integers $\ell$, we have
$$\bar{A}(t-\ell T_s,t'-\ell T_s) = \bar{A}(t,t').$$
We may thus focus on the time-averaged autocorrelation function
\begin{align}
\bar{A}(\tau) & = \frac{1}{T_s} \int_{0}^{T_s} \bar{A}(t,t-\tau) dt
\label{eq:autocorr-stationary}
\end{align}
and we have
\begin{align}
\bar{\mathcal{P}}(f) & =  \int_{-\infty}^{\infty} \bar{A}(\tau) e^{- j 2 \pi f \tau} \, d\tau .
\end{align}

%=============================================================
%=============================================================
\subsection{Additive White Gaussian Noise Channels}
\label{subsec:AWGN}
The classic way of dealing with noise for linear channels is to use the AWGN model
\begin{align}
   u_r(t) = u_0(t) + n_r(t)
   \label{eq:AWGN-noise-model}
\end{align}
where $n_r(\cdot)$ is a realization of the complex, circularly symmetric, white, Gaussian process
$N_r(\cdot)$ with a one-sided PSD of $N_0$ Watts/Hertz across all frequencies.
We will consider {\em thermal} noise with $N_0=k_B T_e$ where
$k_B \approx 1.381 \times 10^{-23}$ Joules/Kelvin is Boltzmann's constant, and where
$T_e$ is the temperature in Kelvin.

The model \eqref{eq:AWGN-noise-model} is artificial because $n_r(\cdot)$ has
infinite bandwidth and infinite power.\footnote{The per-sample capacity of
the AWGN channel is therefore zero.} Of course, noise encountered in practice
has finite bandwidth and power, and the idea is that the noise
PSD is flat for frequencies much larger than those of the processing capabilities
of the transmitter or receiver. An optimal receiver projects its input
signal onto the linear subspace spanned by the transmit signals, see
Sec.~\ref{subsec:receiver-noise}.

Consider next the bandlimited AWGN channel
\begin{align}
   u_r(t) = \left( u_0(t)  + n_r(t) \right)* W \sinc(W t)
   \label{eq:AWGN-noise-model-bandlimited}
\end{align}
where $*$ denotes convolution.
One can convert this channel into a discrete-time channel by
sampling $u_r(\cdot)$ at the Nyquist rate $W$ Hz. The capacity
under the average power constraint 
\begin{align}
   \bar{P}_T = \frac{1}{T} \int_{-T/2}^{T/2} \E{\left| U_0(t) \right|^2} \, dt \le P
   \label{eq:P-constraint}
\end{align}
for large $T$ is achieved by using PAM, sinc pulses, and Gaussian modulation,
and is given by \cite[Sec.~25]{Shannon48}
\begin{align}
   C(W) = W \log_2\left( 1 + \frac{P}{W N_0} \right) \text{ bits/s}.
   \label{eq:CW}
\end{align}
The value $C$ increases with $W$, and we have
\begin{align}
   \lim_{W \rightarrow \infty} C(W) = \frac{P}{N_0} \log_2(e) \text{ bits/s}.
   \label{eq:CWlim}
\end{align}
In other words, capacity scales logarithmically with the signal-to-noise
ratio (SNR) $P/(WN_0)$ \emph{with} a bandwidth limitation,
and linearly with $P/N_0$ \emph{without} a bandwidth limitation.

The spectral efficiency is defined as
\begin{align}
   \eta(W) = \frac{C(W)}{W} = \log_2\left( 1 + \frac{P}{W N_0} \right) \text{ bits/s/Hz}
   \label{eq:etaW}
\end{align}
and we have $\eta(W) \rightarrow 0$ in the limit of large $W$. However,
one usually studies $P=E/T_s$ where $E$ is the average energy of PAM
with sinc pulses that are offset by $T_s=1/W$ seconds. We thus have
\begin{align}
   \eta(W) = \frac{C(W)}{W} = \log_2\left( 1 + \frac{E}{N_0} \right) \text{ bits/s/Hz}
   \label{eq:eta}
\end{align}
which is independent of $W$. Note that this approach has a transmit power
$P=E W$ that grows with $W$.

\begin{remark}
The constraint \eqref{eq:P-constraint} permits {\em peaky} or {\em flash} signals
with arbitrarily large amplitudes if $T\rightarrow\infty$.
In practice, however, the input amplitude is limited, i.e.,
we require $|u_0(t)|\le A_{\rm max}$ for all $t$ and for some positive $A_{\rm max}$.
The capacity under an input amplitude constraint was studied in~\cite[Sec.~26]{Shannon48},
for example.
\end{remark}
\begin{remark}
Suppose $U_0(\cdot)$ has the PSD $\bar{\mathcal{P}}_0(\cdot)$ so that
the power at the output of the channel \eqref{eq:AWGN-noise-model-bandlimited} is
\begin{align}
   \bar{P}_r(W) & = WN_0 + \int_{-W/2}^{W/2} \bar{\mathcal{P}}_0(f) \, df .
   \label{eq:Pr-def}
\end{align}
Suppose further that, instead of the \emph{launch} constraint \eqref{eq:P-constraint},
the \emph{receiver} signal $U_r(\cdot)$ must satisfy $\bar{P}_r(W)  \le P+WN_0$.
The capacity of the channel \eqref{eq:AWGN-noise-model-bandlimited}
is then again (see \cite[Sec.~29]{Shannon48} and \cite{Gastpar-IT07})
\begin{align}
  C(W) = W \log_2\left( 1 + \frac{P}{W N_0}  \right) \text{ bits/s} .
  \label{eq:general-capacity-bound}
\end{align}
\end{remark}

%=============================================================
%=============================================================
\subsection{Hyperbolic Functions}
\label{subsec:hyperbolic-functions}
We use the following functions with complex arguments:
\begin{align}
& S(c) = \sech\left( \sqrt{2c} \, z\right) \label{eq:S-function} \\
& T(c) = \left. \tanh\left( \sqrt{2c} \, z\right)\right/\sqrt{2c} \label{eq:T-function}
\end{align}
where we suppress the dependence on $z$ on the left-hand side (LHS)
of \eqref{eq:S-function} and \eqref{eq:T-function} for notational simplicity.
As further simplification, we write $S_R(c)=\Re(S(c))$, $S_I(c)=\Im(S(c))$,
$T_R(c)=\Re(T(c))$, and $T_I(c)=\Im(T(c))$.

Consider $c=-j x/z^2$ where $x$ is real and non-negative.
The following bounds are valid numerically, see Fig.~\ref{fig:Splot} and~\ref{fig:Tplot}:
\begin{align}
   \begin{array}{ll}
   |S(c)| \le 1, & |T(c)| \le z \\
   -0.136 \le S_R(c) \le 1, & 0\le T_R(c) \le z  \\
   -0.028 \le S_I(c) \le x, & 0 \le T_I(c) \le (2/3) x z .
   \end{array} \label{eq:STbounds1}
\end{align}
For small $x$, we have
\begin{align}
\begin{array}{ll}
   |S(c)| \ge d, & |T(c)| \ge d\cdot z \\
   S_R(c) \ge d, & T_R(c) \ge d\cdot z \\
   S_I(c) \ge d\cdot x, & T_I(c) \ge d\cdot (2/3) x z
\end{array} \label{eq:STbounds2}
\end{align}
for a constant $d$ that approaches 1 as $x$ approaches 0.
We further have the following bounds, see Figs.~\ref{fig:Splot} and~\ref{fig:Tplot}:
\begin{align}
   & |S(c)| \le \sqrt{5} \, e^{-\sqrt{x}} \label{eq:Smagbound} \\
   & T_I(c) \ge  \frac{z}{3} \min\left( x, \frac{1}{\sqrt{x}} \right) \label{eq:Tbound3} \\
   & \frac{S_I(c)^2}{T_I(c)} \le  \frac{3x}{2z}. \label{eq:SbyTbound}
\end{align}

%%%%%%%%%%%%%%%%%%%%%%%%%%%%%%%%%%%%%%%%
%\begin{figure*}[t!]
%\begin{minipage}[t]{.48\textwidth}
\begin{figure}[t!]
  \centerline{\includegraphics[scale=0.48]{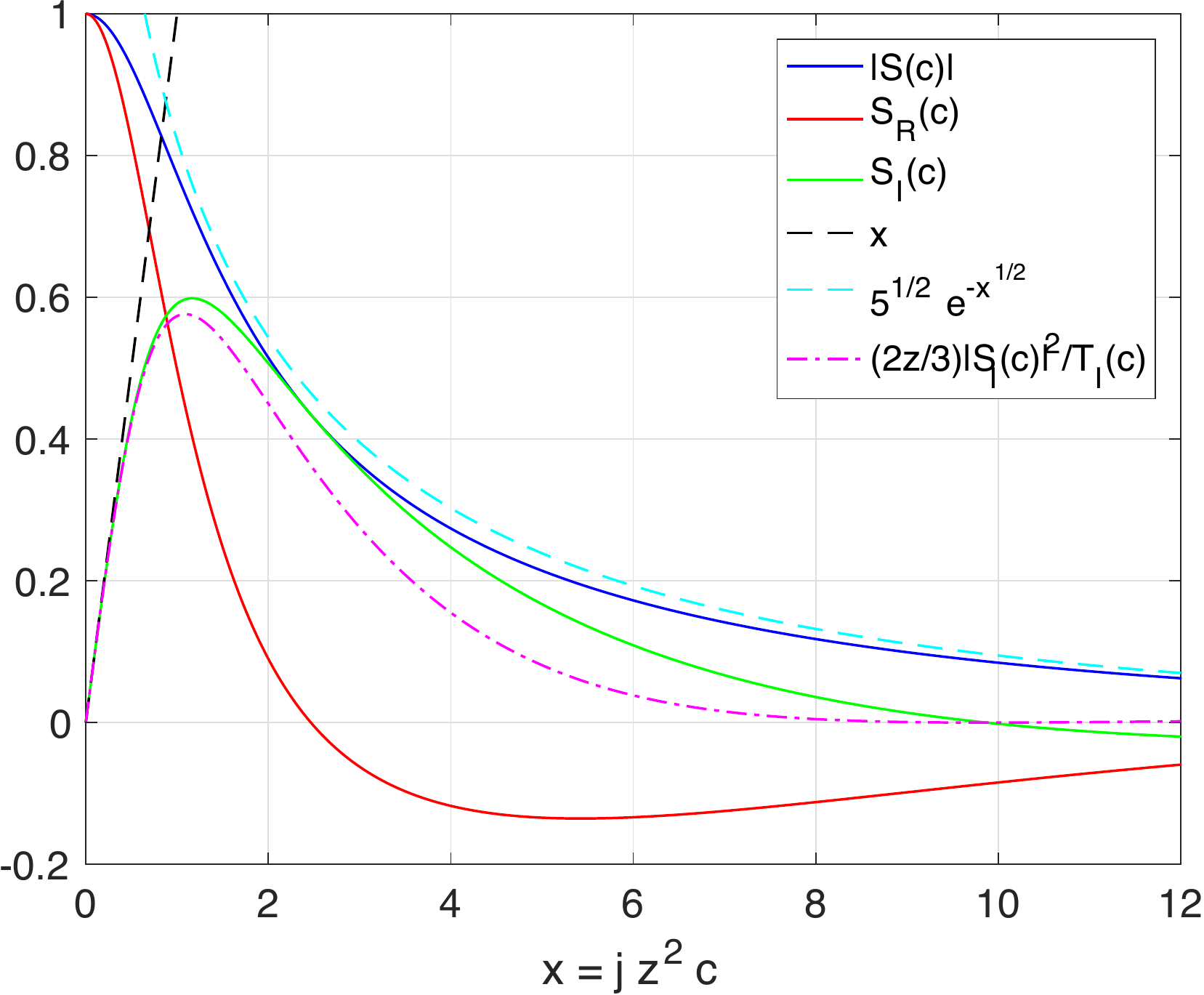}}
  \caption{Plot of $|S(c)|$, $S_R(c)$, $S_I(c)$, and related bounds for $c=-jx/z^2$.}
  \label{fig:Splot}
\end{figure}
%\end{minipage}
%
%\hfill
%
%\begin{minipage}[t]{.48\textwidth}
\begin{figure}[t!]
  \centerline{\includegraphics[scale=0.48]{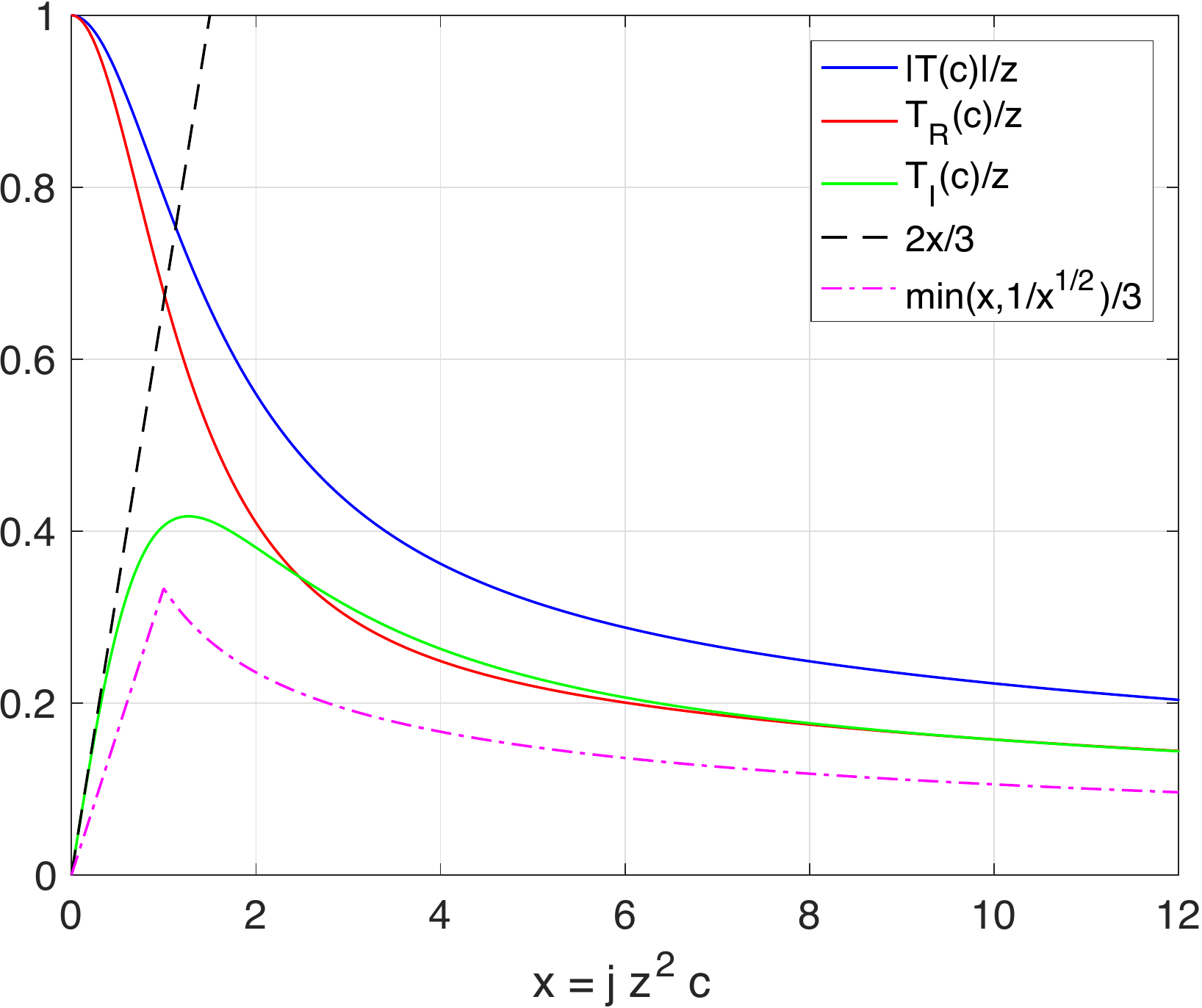}}
  \caption{Plot of $|T(c)|/z$, $T_R(c)/z$, $T_I(c)/z$, and related bounds for $c=-jx/z^2$.}
  \label{fig:Tplot}
\end{figure}
%\end{minipage}
%\end{figure*}
%%%%%%%%%%%%%%%%%%%%%%%%%%%%%%%%%%%%%%%%

%=============================================================
%=============================================================
\section{Fiber and OA Noise Models}
\label{sec:fiber-models}
%
%=============================================================
%=============================================================
\subsection{Nonlinear Schr\"odinger Equation}
\label{subsec:NLSE}
Consider the slowly varying component $a(z,T)$ of a single-mode, linearly-polarized,
electromagnetic wave $a(z,T) e^{j 2 \pi f_0 T}$ along an optical fiber, where $z$ is
the position, $T$ is time, and $f_0$ is the carrier frequency. The propagation of
$a(z,T)$ in the signaling regime of interest is governed by the stochastic NLSE
(see~\cite[Eq. (2.3.27)]{Agrawal-03})
\begin{align}
   \frac{\partial a}{\partial z}
   + \frac{\alpha}{2} a
   + \beta_1 \frac{\partial a}{\partial T}
   + j \frac{\beta_2}{2} \frac{\partial^2 a}{\partial T^2}
   - j \gamma |a|^2 a - n = 0
   \label{eq:NLSE}
\end{align}
where $\alpha$ is the loss coefficient,
$\beta_1$ is the group velocity,
$\beta_2$ is the group velocity dispersion parameter, and
$\gamma$ is the Kerr coefficient.
All of these parameters are frequency dependent in general, but we neglect
this dependence for simplicity. The variables $n(z,T)$ are realizations of noise
random variables $N(z,T)$ whose characteristics we  discuss below in
Sec.~\ref{subsec:OA-noise-model}.
\begin{remark}
The wave propagation is sometimes defined using the complex conjugate of $a(z,T)$.
For example, this approach is used in~\cite{Mecozzi-94}.
\end{remark}
\begin{remark}
For signals with very large power and bandwidth, one should include more effects
such as self-steepening and intra-pulse Raman scattering,
see~\cite[Eq. (2.3.33) and (2.3.39)]{Agrawal-03}. We do not consider these effects here.
\end{remark}

The NLSE is usually expressed using the retarded time $t=T-\beta_1 z$
and the amplified signal $u(z,t)=e^{\alpha z/2} a(z,t)$. Inserting these modifications
into \eqref{eq:NLSE}, we have the simplified equation
\begin{align}
   \frac{\partial u}{\partial z}
   = - j \frac{\beta_2}{2} \frac{\partial^2 u}{\partial t^2}
      + j \gamma |u|^2 u \, e^{-\alpha z} + n \, e^{\alpha z/2}
   \label{eq:NLSE2a}
\end{align}
where we abuse notation and write $n(z,t)$ for $n(z,t + \beta_1 z)$.
A commonly studied version of \eqref{eq:NLSE2a} has $\alpha=0$ so that
\begin{align}
   \frac{\partial u}{\partial z}
   = - j \frac{\beta_2}{2} \frac{\partial^2 u}{\partial t^2}
      + j \gamma |u|^2 u + n.
   \label{eq:NLSE2}
\end{align}
The model \eqref{eq:NLSE2} has many interesting features. For example,
if $\beta_2<0$ and there is no noise, i.e., $n(t)=0$ for all $t$,
then the fiber supports {\em bright solitons},
see ~\cite[Ch.~5]{Agrawal-03}. However, the general model seems to have no
closed-form solution and, to gain insight, one often studies channels without
nonlinearity ($\gamma=0)$ or without dispersion ($\beta_2=0$) or without noise.

%=============================================================
%=============================================================
\subsection{OA Noise Model}
\label{subsec:OA-noise-model}
The signal $n(z,\cdot)$ in \eqref{eq:NLSE}-\eqref{eq:NLSE2} represents OA noise.
Two common choices for OA are erbium-doped fiber amplifiers (EDFAs) at specified
positions along the fiber and distributed Raman amplification~\cite[Sec.~IX.B]{ekwfg-JLT10}.
We consider Raman amplification, and we observe that the noise statistics are
rather complex, see Appendix~\ref{app:raman}. We consider two noise models;
the first model is described in this section and is used in Sec.~\ref{subsec:linear-model} below.
The second model is developed in Sec.~\ref{subsec:dispersion-free-model}
and is used in the remaining sections to analyze spectral broadening.

The usual approach is to model the OA noise $n(z,\cdot)$ as a bandlimited Gaussian process with
the same bandwidth $B$ as the OA bandwidth. The noise is assumed
independent across positions $z$ so that we have the spatiotemporal
autocorrelation function\footnote{The transformation $t=T-\beta_1 z$
does not change this equation.}
\begin{align}
   \E{N(z,t) N(z',t')^*} = K \, \delta(z-z') \, \sinc(B (t-t'))
   \label{eq:st-autocorr}
\end{align}
where $K=N_A B$  is the noise power distance density (PDD) in W/m, and
where $N_A$ is the OA noise power spectral-distance density (PSDD) in W/Hz/m.

More precisely, the accumulated noise at time $t$ is modeled as a {\em spatial}
Wiener process $\sqrt{K} \, W(\cdot,t)$ with $W(z,t)=(W_R(z,t) + j W_I(z,t))/\sqrt{2}$
where $W_R(\cdot,t)$ and $W_I(\cdot,t)$  are independent standard real Wiener processes.
A useful model is that $W(\cdot,t)$ is the limit for large $\ell$ of the processes
\begin{align}
   W_{\ell}(z,t) = \sum_{i=1}^{\lfloor \ell z \rfloor} \frac{1}{\sqrt{\ell}} N_{i}(t)
   \label{eq:Wiener}
\end{align}
where the time processes in the sequence $\{N_{i}(\cdot)\}_{i=1}^{\infty}$ are
independent and identically distributed (i.i.d.), complex, circularly-symmetric,
bandlimited, Gaussian random processes with mean 0 and autocorrelation
function
\begin{align}
   \E{N_i(t) N_i(t')^*} = \sinc(B (t-t')).
   \label{eq:time-corr}
\end{align}
Thus, $W_{\ell}(z,t)$ and $W_{\ell}(z,t')$ are correlated in general, and from
\eqref{eq:Wiener}-\eqref{eq:time-corr} we have the correlation coefficient
\begin{align}
   \rho_{\ell}(z,t,t') & = \frac{\E{W_\ell(z,t) W_\ell(z,t')^*}}{\sqrt{ \E{|W_\ell(z,t)|^2} \E{|W_\ell(z,t')|^2} } } \nonumber \\
   & = \sinc(B (t-t')). \label{eq:Wcorrelation}
\end{align}
Observe that \eqref{eq:Wcorrelation} is independent of $\ell$, $z$, and the
absolute time $t$, so we use the notation $\rho(t-t')$ instead of $\rho_{\ell}(z,t,t')$.

%=============================================================
%=============================================================
\subsection{Lossless and Linear Fiber Model}
\label{subsec:linear-model}
Consider \eqref{eq:NLSE2} but with $\gamma=0$. Taking Fourier transforms,
the propagation equation is
\begin{align}
   \frac{\partial \tilde{u}}{\partial z}
   = j \frac{\beta_2}{2} (2 \pi f)^2 \tilde{u} + \tilde{n}
   \label{eq:NLSE-linear}
\end{align}
where $\tilde{u}(z,\cdot)$ and $\tilde{n}(z,\cdot)$ are the respective Fourier transforms of
$u(z,\cdot)$ and $n(z,\cdot)$. The solution of \eqref{eq:NLSE-linear} is
\begin{align}
   \tilde{u}(z,f) = \tilde{u}(0,f) \, e^{j \frac{\beta_2}{2} (2 \pi f)^2 z}
   + \sqrt{K} \tilde{w}(z,f)
   \label{eq:linear-model}
\end{align}
where $\tilde{w}(z,\cdot)$ is the Fourier transform of $w(z,\cdot)$.
Using the noise model described in Sec.~\ref{subsec:OA-noise-model}, the process $\tilde{w}(\cdot,f)$
is a realization of a spatial Wiener process if $|f|<B/2$, and is zero otherwise.%\footnote{Alternatively,
%$\tilde{w}(\cdot,B/2)$ and $\tilde{w}(\cdot,-B/2)$ may be chosen as spatial Wiener
%processes with half the power of the $\tilde{w}(\cdot,f)$ with $|f|<B/2$.}
The channel filter is therefore an all-pass filter with phase shifts proportional to $f^2$;
the frequency-dependence of the phase is
called {\em chromatic dispersion}. Furthermore, the channel is noise-free
outside the band $|f|<B/2$. This property is problematic when considering
information theoretic limits of communication, see Sec.~\ref{subsec:large-capacity} below.

%=============================================================
%=============================================================
\subsection{Lossless and Dispersion-Free Fiber Models}
\label{subsec:dispersion-free-model}
Consider \eqref{eq:NLSE2} but without dispersion. We have
\begin{align}
   \frac{\partial u}{\partial z} = j \gamma |u|^2 u + n.
   \label{eq:NLSE-dispersion-free}
\end{align}
The formal solution of \eqref{eq:NLSE-dispersion-free} is (see \cite[eq. (11)]{Mecozzi-94})
\begin{align}
   u & = \left[ u_0 + \sqrt{K} \, \hat{w} \right]
   \exp\left[ j \gamma \int_0^z \left| u(z',t) \right|^2 dz' \right] 
   \label{eq:dispersion-free-model-formal}
\end{align}
where the accumulated noise is
\begin{align}
   & \hat{w}(z,t) = \int_{0}^{z} n(z',t) \exp\left[ -j \gamma \int_0^{z'} \left| u(z'',t) \right|^2 dz'' \right] dz'.
   \label{eq:dispersion-free-model-noise}
\end{align}
Suppose that $n(z,t)$ is Gaussian with uniform phase and is spatially white, as described
in Sec.~\ref{subsec:OA-noise-model}. This means that
$\hat{w}(\cdot,t)$ is a realization of a spatial Wiener process that is independent of $u_0(t)$,
see \cite[eq. (13)]{Mecozzi-94}. One may thus write the sampled output
\eqref{eq:dispersion-free-model-formal} explicitly as
\begin{align}
   u(z,t) & = \left[ u_0(t) + \sqrt{K} \, w(z,t) \right] \nonumber \\
   & \quad \exp\left[ j \gamma \int_0^z \left| u_0(t) + \sqrt{K} w(z',t) \right|^2 dz' \right]
   \label{eq:dispersion-free-model}
\end{align}
where $w(\cdot,t)$ is a spatial Wiener noise process.

However, the autocorrelation $A_z(t,t')$ involves the statistics of \emph{two}
samples $u(z,t)$ and $u(z,t')$, and the processes $\hat{w}(\cdot,t)$ and
$\hat{w}(\cdot,t')$ that define these respective samples might not be
\emph{jointly} Wiener.\footnote{At times $t$ and $t'$ for which
$\rho(t-t')=0$, the accumulated noise processes $\hat{w}(\cdot,t)$ and $\hat{w}(\cdot,t')$ are
statistically independent and hence jointly Wiener.}
The reason is that the exponential in \eqref{eq:dispersion-free-model-noise}
decorrelates the noise, i.e., the temporal noise process
$\hat{w}(z,\cdot)$ also experiences spectral broadening.
As an additional complication, the noise model of Sec.~\ref{subsec:OA-noise-model}
may not be accurate for Raman scattering because the coupling of the pump and
propagating signals also decorrelates the noise. We address Raman scattering
in Appendix~\ref{app:raman}.

To circumvent the difficulties, and to permit analysis, we study the model
\eqref{eq:dispersion-free-model} where $w(\cdot,t)$ and $w(\cdot,t')$ are
jointly Wiener processes for any $t$ and $t'$.  This model seems reasonable
for small $\gamma$, $K$, $z$, and $P$, but accumulated noise with
exponential terms such as in \eqref{eq:dispersion-free-model-noise} deserve more study.

As a final remark, if we expand the quadratic term of the exponential in
\eqref{eq:dispersion-free-model}, then we obtain three parts:\ a self-phase
modulation term $|u(0,t)|^2$, a signal-noise mixed term
$2 \sqrt{K} \Re\{ u(0,t) w(z',t)^* \}$, and a noise term $K |w(z',t)|^2$.
If $|u(0,t)|$ is large, then the signal-noise mixed term will cause large and random
phase variations that result in uncontrolled spectral broadening. Understanding
this effect seems key to understanding the {\em nonlinear Shannon limit} of
optical fiber, i.e., the limitation of the capacity~\cite{ekwfg-JLT10}.

%=============================================================
%=============================================================
\section{Receiver Models}
\label{sec:rx-models}

This section reviews two classes of receiver models. The first class is the per-sample
models studied in~\cite{Mecozzi-94,Turitsyn-Derevyanko-Yurkevich-Turitsyn-PRL03,Yousefi-Kschischang-IT11,Mecozzi-04,Terekhov-etal-A16,Fahs-etal-A17}. The second class is the filter-and-sample
models that are commonly used in information theory~\cite[Ch.~8.1]{Gallager68}. For the latter models,
we will study receivers that we consider to be of engineering relevance, namely bandlimited
receivers and time-resolution limited receivers, both with thermal noise.

%=============================================================
%=============================================================
\subsection{Per-Sample Model and Unbounded Capacity}
\label{subsec:large-capacity}
The capacity of the channel \eqref{eq:linear-model} is well-understood. For example,
if the launch signal $\tilde{u}(0,\cdot)$ can have energy outside the noise band,
then the capacity is unbounded for any launch power~\cite[Thm.~5]{Wyner-BSTJ66}.
The purpose of this section is to show that the nonlinear model
\eqref{eq:dispersion-free-model} can also have unbounded capacity, which
suggests that the per-sample model gives limited insight for
realistic receivers, see Sec.~\ref{subsec:per-sample-vs-multi-sample}.

Consider $T \le 1/B \le T_s$ and the launch signal
\begin{align}
   u(0,t) = x \, g(t-T/2)  - 2x \, g(t -3T/2)
\end{align}
where $x$ is an information symbol and $g(\cdot)$ is a rectangular pulse
of unit norm in the interval $[0,T_s)$. We claim that, for small $T$,
the launch bandwidth is described mainly by $T_s$ and not $T$. We have
\begin{align}
   |\tilde{u}(0,f)| = |x| \cdot \left| \tilde{g}(f) \right|  \cdot \sqrt{ 5 - 4 \cos( 2 \pi f T)}
\end{align}
where $\tilde{g}(\cdot)$ is the Fourier transform of $g(\cdot)$.
Choosing small $T$  thus does not increase the launch bandwidth, e.g.,
if one uses a measure such as the band having 99\% of the power.

We now choose $T \ll 1/B$ so that the noise variables at $t=0,\,T,\,2T$
are approximately the same, i.e., we have
$$w(z,0) \approx w(z,T) \approx w(z,2T).$$
We can make the approximations as accurate as desired by choosing 
sufficiently small $T$. Let $w(z,0)=w_R + j w_I$ and suppose $x$ is real.
We choose a 3-sample receiver that outputs
\begin{align}
   & y_0 = \left| u_z(0) \right|^2 = K \left( w_R^2 + w_I^2 \right) \\
   & y_1 = \left| u(z,T) \right|^2 \approx \left( x/\sqrt{T_s} + \sqrt{K} w_R \right)^2 + K w_I^2 \\
   & y_2 = \left| u(z,2T) \right|^2 \approx \left( -x/\sqrt{T_s} + \sqrt{K} w_R \right)^2 + K w_I^2
\end{align}
and compute
\begin{align}
   & x^2 \approx T_s\left( \frac{y_1 + y_2}{2} - y_0 \right)
\end{align}
to any desired accuracy by choosing sufficiently small $T$. This means that, for 
the launch energy $x^2(1+4T/T_s)\approx x^2$, we can achieve {\em any}
rate by choosing sufficiently small $T$. In other words, \emph{the capacity is
unbounded for any launch power.}

\begin{remark}
The reader might consider this example unsatisfactory because it requires
rapid signaling. In fact, the example does not work with sinc pulses, yet it
seems artificial to limit attention to such pulses. The main purpose of the
example is to show that bandlimited noise has pitfalls when dealing with
capacity, and that care is needed in treating bandwidth~\cite{Slepian76}.
\end{remark}
\begin{remark}
The above observations apply also to the full NLSE model of Sec.~\ref{subsec:NLSE}.
For example, at small launch power the channel is basically an AWGN channel
with bandlimited noise.
\end{remark}

%=============================================================
%=============================================================
\subsection{Filter-and-Sample Model}
\label{subsec:receiver-noise}
A natural approach to circumvent the infinite-capacity problem of bandlimited noise is
to add AWGN to the nonlinear model \eqref{eq:dispersion-free-model}.
In other words, the new model has {\em two} noise processes: the bandlimited
distributed OA noise and the receiver AWGN.
More precisely, consider a receiver that operates on a noisy signal
\begin{align}
   u_r(t) = u(z,t) + n_r(t)
   \label{eq:electronic-noise-model}
\end{align}
where $n_r(\cdot)$ is the same as in \eqref{eq:AWGN-noise-model}.

Now consider the set $L^2[0,T]$ of continuous and finite-energy signals in the
time\footnote{One may also consider the set of continuous and finite-energy
signals in a frequency interval.}
interval $[0,T]$.  This set has a complete orthonormal basis
$\left\{\phi_{m}(\cdot)\right\}_{m=1}^{\infty}$ and one usually has a
receiver that puts out a finite number of projection values
\begin{align}
   Y_m & = \int_0^T U_r(z,t) \phi_{m}(t)^* dt \nonumber \\
   & = Z_m + \int_0^T U(z,t) \phi_{m}(t)^* dt
   \label{eq:projection}
\end{align}
for $m=1,2,\ldots,M$ where
\begin{align}
   Z_m & = \int_0^T N_r(t) \phi_{m}(t)^* dt .
   \label{eq:Zprojection}
\end{align}
In other words, $Z_1,Z_2,\ldots,Z_M$ is a string of statistically independent,
complex, circularly-symmetric, Gaussian random variables with variance $N_0$.
The set $\{Y_m\}_{m=1}^{M}$ of measurements forms a set of sufficient statistics
if every possible signal $u(z,\cdot)$ lies in the subspace spanned by the signals
$\left\{\phi_{m}(\cdot)\right\}_{m=1}^{M}$. Otherwise, one must let $M \rightarrow \infty$
in general. We refer to the above model as the \emph{filter-and-sample} model to
distinguish it from the \emph{per-sample} model.

\begin{remark}
An alternative to introducing receiver thermal noise is to assume that $u(z,\cdot)$
has spectral components inside the OA bandwidth only. However, this approach
prevents considering finite-time pulses such as rectangular pulses. Furthermore,
spectral broadening prevents $u(z,\cdot)$, $z>0$, from remaining strictly bandlimited
even if $u(0,\cdot)$ is strictly bandlimited.
\end{remark}
\begin{remark} \label{rmk:infinite-bandwidth-noise}
A second alternative is to study OA with infinite bandwidth but with
a finite PDD of $K$ W/m. Another way to think of this is that the
noise PSDD is $N_A=K/B$ W/Hz/m and one considers the limit of increasing $B$.
This is effectively what was done in~\cite[Sec.~IV]{Mecozzi-94}, and we
develop results for this model in Appendix~\ref{app:infiniteB}. The model
is artificial but it has two useful features:\ the analysis greatly simplifies
and the model gives insight into systems where the optical noise process
has much larger bandwidth than the signals propagating along the fiber.
Related studies on models with white phase noise can be found
in~\cite{Goebel-etal-IT11,Barletta-Kramer-CROWNCOM14,Barletta-Kramer-ISIT14}. 
\end{remark}
\begin{remark} \label{rmk:nonlinear-capacity}
We show in Appendix~\ref{app:nonlinear-capacity} that the nonlinearity can
increase capacity. The idea is to use the nonlinearity to 
convert amplitude-shift keying (ASK) to orthogonal frequency-shift
keying (FSK). We remark that orthogonal FSK achieves capacity
for large bandwidth $W$~\cite[p.~207]{Proakis-Salehi-5} and that capacity grows linearly
in launch power for large $W$, see~\eqref{eq:CWlim}.
\end{remark}

%=============================================================
%=============================================================
\subsection{Bandlimited Receiver}
\label{subsec:bandlimited-receiver-defn}
We will consider two receivers that are related. The first
is bandlimited to $W$ Hz, i.e., the receiver collects energy in the frequency band
$f\in [-W/2,W/2]$ only. The average receiver power after filtering is
(see~\eqref{eq:Pr-def})
\begin{align}
   \bar{P}_r(W) & = \int_{-W/2}^{W/2} \bar{\mathcal{P}}(f) \, df.
   \label{eq:Pr-def1}
\end{align}
Note that we have not included the noise $N_r(\cdot)$ in $\bar{P}_r(W)$;
this noise will contribute an additional $WN_0$ Watts.

For convenience, we define (see~\eqref{eq:PSD-T-def})
\begin{align}
   & \bar{P}_r(W,T) = \int_{-W/2}^{W/2} \bar{\mathcal{P}}(f,T) \, df . %\nonumber \\
   %& \quad = \frac{1}{T} \int_{-T/2}^{T/2} \int_{-T/2}^{T/2} \bar{A}(t,t') \, \sinc\left( W(t-t') \right) \, dt' \, dt .
   \label{eq:PrT-def}
\end{align}
To later help us bound $\bar{P}_r(W,T)$, we upper
bound the PSD of unit height in the frequency interval $[-W/2,W/2]$ by
\begin{align}
   \tilde{b}(f) = \left\{ \begin{array}{ll}
   2\left(1 - \frac{|f|}{W} \right), &  |f| \le W \\
   0, & \text{else}.
   \end{array} \right.
   \label{eq:Bfreq}
\end{align}
The motivation for this step is to ensure that the absolute value of the
corresponding time signal
\begin{align}
   b(t) = 2W \, \sinc\left( W t \right)^2
   \label{eq:Btime}
\end{align}
integrates to a finite value for $t\ge0$, namely the value
\begin{align}
   \int_0^\infty |b(t)| \, dt = 1.
   \label{eq:Btime-bound}
\end{align}

Inserting \eqref{eq:Bfreq} into \eqref{eq:PrT-def}, we have
\begin{align}
   \bar{P}_r(W,T)
   &  \le \int_{-W}^{W} \bar{\mathcal{P}}(f,T) \, \tilde{b}(f) \, df \nonumber \\
%   & = \int_{-\infty}^{\infty} \left[\frac{1}{T} \, \int_{-T/2}^{T/2} \int_{-T/2}^{T/2} \bar{A}(t,t')
%         e^{- j 2 \pi f (t-t')} \, dt' \, dt  \right] \nonumber \\
%   & \qquad \qquad \cdot R(f) \, df \nonumber \\
   & = \frac{1}{T} \int_{-T/2}^{T/2} \int_{-T/2}^{T/2} \bar{A}(t,t') \, b(t-t') \, dt' \, dt .
   \label{eq:PrT-bound}
\end{align}
For cyclostationary signals, we obtain
\begin{align}   
   \bar{P}_r(W) \le \int_{-\infty}^{\infty} \bar{A}(\tau) \, b(\tau) \, d\tau .
   \label{eq:Pr-bound}
\end{align}

%=============================================================
%=============================================================
\subsection{Time-Resolution Limited Receiver}
\label{subsec:timelimited-receiver}
The second receiver is time-resolution limited to $T_r$ seconds where
$T_r\le T_s$. More precisely, we consider a normalized integrate-and-dump filter
over $T_r$ seconds. The energy output by the receiver is
\begin{align}
 E_m(T_r) & = {\rm E} \left[ \left. \left| \int_{mT_r}^{(m+1)T_r} \frac{1}{\sqrt{T_r}} \, U(t) \, dt \right|^2 \right|  U_0(\cdot)=u_0(\cdot) \right] \nonumber \\
 & = \frac{1}{T_r} \int_{mT_r}^{(m+1)T_r} \int_{mT_r}^{(m+1)T_r} A(t,t') \, dt' \,  dt .
 \label{eq:Er-def}
\end{align}
Thus, the average energy is
\begin{align}
 \bar{E}_m(T_r) = \frac{1}{T_r} \int_{mT_r}^{(m+1)T_r} \int_{mT_r} ^{(m+1)T_r} \bar{A}(t,t') \, dt' \, dt .
 \label{eq:Erbar-def}
\end{align}
The value $\bar{E}_m(T_r)$ is closely related to the right-hand side
(RHS) of \eqref{eq:PrT-bound}.

\begin{remark}
One can build a receiver with time-resolution $T_r/2$ seconds with two
receivers with time resolution $T_r$ seconds by offsetting their integration
times by $T_r/2$ seconds. Thus, it might make more sense to
state that our receivers have limited time {\em precision}.
Similarly, one can build a bandwidth
$2W$ receiver with two bandwidth $W$ receivers whose center frequencies
are offset by $W$ Hz.  Of course, these approaches increase
complexity and cost, and in practice one is limited by the available receivers.
\end{remark}

%=============================================================
%=============================================================
\section{Autocorrelation Function}
\label{sec:auto}
The following theorem is our main analytic result. Consider
\begin{align}
   c = - j \gamma (K/2) \sqrt{1-\rho^2}
   \label{eq:c-def}
\end{align}
so that in Sec.~\ref{subsec:hyperbolic-functions} we have $c=-jx/z^2$ where
\begin{align}
   x=\gamma (K/2) z^2 \sqrt{1-\rho^2}.
   \label{eq:x-def}
\end{align}

\begin{theorem} \label{thm:autocorr}
The conditional autocorrelation function \eqref{eq:autocorr-cond}
of the signal \eqref{eq:dispersion-free-model} when $w(\cdot,t)$
and $w(\cdot,t')$ are jointly Wiener processes for any $t$ and $t'$ is
\begin{align}
& A(t,t') = |S(c)|^2 \Big[ T_R(c) K \rho + \nonumber \\
& \qquad \left( S_R(c) u_0 + j \frac{S_I(c)}{\sqrt{1-\rho^2}} (u_0 - \rho u_0') \right) \nonumber \\
& \qquad \left. \left( S_R(c) u_0' +  j \frac{S_I(c)}{\sqrt{1-\rho^2}} (u_0' - \rho u_0) \right)^* \; \right] \nonumber \\
& \quad \exp\left(j \gamma \, T_R(c) \left[ |u_0|^2 - |u_0'|^2 \right] \right) \nonumber \\
& \quad \exp\left( - \gamma \frac{T_I(c)}{\sqrt{1-\rho^2}}[|u_0|^2+|u'_0|^2-2\rho\Re\{u_0 u_0'^{*}\}] \right).
\label{eq:autocorr-general}
\end{align}
\end{theorem}
\begin{proof}
Expression \eqref{eq:autocorr-general} is derived in Appendix~\ref{app:two-sample-statistics}
by using the development in \cite{Mecozzi-94} that is reviewed in
Appendices~\ref{app:cameron-martin}-\ref{app:one-sample-statistics}.
\end{proof}

\begin{remark}
The first exponential of \eqref{eq:autocorr-general} describes self-phase modulation (SPM)
and the second exponential is due to signal-noise mixing.
The argument of the latter exponential is real, non-positive, and decreases with the
instantaneous powers $|u_0|^2$ and $|u_0'|^2$ of the launch signal unless $t=t'$.
\end{remark}
\begin{example}
Consider $\gamma=0$ so the channel is linear. We have
$c=0$, $S(c)=1$, $T(c)=z$, and therefore
\begin{align}
   A(t,t') & = K z \rho(t-t') + u_0(t) u_0(t')^* .
\end{align}
\end{example}
\begin{example}
Consider $t=t'$ so that $u_0=u_0'$, $\rho=1$, $c=0$, 
$S(c)=1$, $T(c)=z$, and therefore
\begin{align}
   A(t,t) & = \E{\left. |U_z(t)|^2 \right| U_0(t)=u_0(t) } \nonumber \\
   & = Kz + | u_0(t) |^2 .
   \label{eq:t-is-tprime}
\end{align}
The instantaneous power is thus preserved, see \eqref{eq:dispersion-free-model}.
\end{example}
\begin{example}
Consider $u_0(t)=0$ for all $t$. The noise autocorrelation function is
\begin{align}
   & A(t,t') = |S(c)|^2 T_R(c) K \rho
   \label{eq:noise-autocorr}
\end{align}
and for $t \ne t'$ the value $|A(t,t')|$ first increases but
eventually decreases as $z$ grows. This means that the noise PSD
eventually broadens with $z$. However, one usually operates in 
regimes where $|c|$ is small so that $A(t,t') \approx K\rho z$.
\end{example}
\begin{example}
Consider $u_0(t)=\sqrt{P} e^{j \phi(t)}$ for all $t$, i.e., the launch signal
has constant envelope. We compute
\begin{align}
   & A(t,t') \nonumber \\
   & = |S|^2 \left[ T_R K \rho + P e^{j \phi_{\Delta}}
      \left( S_R^2 + \frac{S_I^2}{1-\rho^2} \left( 1-\rho \, e^{- j \phi_{\Delta}} \right)^2 \right) \right] \nonumber \\
   & \quad \exp\left( - \gamma \frac{T_I}{\sqrt{1-\rho^2}} 2P \left[ 1 - \rho \cos(\phi_{\Delta}) \right] \right)
   \label{eq:ce-autocorr}
\end{align}
where $\phi_{\Delta}(t,t')=\phi(t)-\phi(t')$, and where we have suppressed
the dependence on $c$, $t$, and $t'$ for notational convenience.
\end{example}

\begin{comment}
We may further refine $A(t,t')$ and consider
\begin{align*}
& u_0=\sqrt{P_0} \, e^{j\Phi_0}, \quad u_0'=\sqrt{P_0'\,} e^{j\Phi_0'}.
\end{align*}
We compute
\begin{align}
& A(t,t') = |S(c)|^2 \Big( T_R(c) K \rho + \sqrt{P_0 P_0'} \, e^{j(\Phi_0-\Phi_0')} \nonumber \\
& \left. \cdot \left\{ S_R(c)^2 + \frac{S_I(c)^2}{1-\rho^2} \right. \right. \nonumber \\
& \quad \left. \cdot \left(1-\rho \sqrt{\frac{P_0'}{P_0}} e^{j(\Phi_0'-\Phi_0)} \right)
   \left(1-\rho \sqrt{\frac{P_0}{P_0'}} e^{j(\Phi_0'-\Phi_0)} \right) \right\} \nonumber \\
& \quad \left. + j \frac{ S_R(c) S_I(c)}{\sqrt{1-\rho^2}} \rho \left[ P_0 - P_0' \right] \right) \nonumber \\
& \quad \exp\left( j \gamma \, T_R(c) \left[ P_0 - P_0' \right] \right) \nonumber \\
& \quad \exp\left( - \gamma \, \frac{T_I(c)}{\sqrt{1-\rho^2}} \,
[ P_0 + P_0' - 2 \sqrt{P_0 P_0'} \, \rho \cos\left( \Phi_0-\Phi_0' \right) ] \right) .
\end{align}
\end{comment}

%\clearpage
%=============================================================
\subsection{Low Noise-Nonlinearity-Distance Product}
\label{subsec:low-noise}
A commonly studied regime is where the noise and/or the Kerr coefficient are small
with respect to the distance. More precisely, we consider $\sqrt{\gamma K z^2}$ small 
enough so that
\begin{align}
 & S(c) \approx 1 - 2c (z^2/2) = 1 + j \gamma K \sqrt{1-\rho^2} (z^2/2) \label{eq:Sapprox} \\
 & T(c) \approx z -2c(z^3/3) = z + j \gamma K \sqrt{1-\rho^2} (z^3/3) \label{eq:Tapprox}
\end{align}
are accurate approximations for the various constants in \eqref{eq:autocorr-general}.
The autocorrelation function is thus approximately
\begin{align}
& \mathcal{A}(t,t') = \left( K\rho z + u_0 u_0'^{*} + j \gamma K \rho \frac{z^2}{2} \left[ |u_0|^2 - |u_0'|^2 \right] \right) \nonumber \\
& \quad \quad \exp\left(+ j \gamma z \left[ |u_0|^2 - |u_0'|^2 \right] \right) \nonumber \\
& \quad \quad \exp\left(- \frac{\kappa}{2} [|u_0|^2+|u'_0|^2-2\rho\Re\{u_0u_0'^{*}\}] \right)
\label{eq:autocorr-approx}
\end{align}
where
\begin{align}
  \kappa=2 \gamma^2 K z^3/3 % Units of W^{-1}
  \label{eq:kappa}
\end{align}
and where we have kept up to second oder terms in $\sqrt{\gamma K z^2}$.

\subsection{Bounds on the Autocorrelation Amplitude}
\label{subsec:autocorr-bounds}
Consider the argument of the last exponential in \eqref{eq:autocorr-general}
for which we have
\begin{align}
& |u_0|^2+|u'_0|^2-2\rho\Re\{u_0u_0'^{*}\} \nonumber \\
& \quad = |u_0' - \rho u_0|^2 + |u_0|^2 \left( 1-\rho^2 \right).
\label{eq:quad-form-bound}
\end{align}
Now suppose that $|u_0| \ge |u_0'|$ so that
\begin{align}
\begin{array}{l}
|u_0 - \rho u_0'| \le 2 |u_0| \\
|u_0' - \rho u_0| \le 2 |u_0|.
\end{array}
\label{eq:quad-form-bound2}
\end{align}
We may use \eqref{eq:quad-form-bound} and \eqref{eq:quad-form-bound2} with
\eqref{eq:STbounds1} to bound
\begin{align}
|A(t,t')| &\le \left[ K z + |u_0| ^2 \left( 1 + 2 \frac{|S_I(c)|}{\sqrt{1-\rho^2}} \right)^2 \right]
\nonumber \\
& \quad \exp\left( - \gamma \, T_I(c) \,  |u_0|^2 \sqrt{1-\rho^2} \right) \nonumber \\
& \le \left[ K z + |u_0|^2 \left (1 + \gamma K z^2 \right)^2 \right] \nonumber \\
& \quad \exp\left( - \gamma \, T_I(c) \,  |u_0|^2 \sqrt{1-\rho^2} \right).
\label{eq:autocorr-bound}
\end{align}
We further have $\gamma K z^2 \approx 0.017$ for certain parameter ranges
that we are interested in, see Table~\ref{table:parameters}.

\begin{remark}
The amplitude $|A(t,t')|$ captures the influence of signal-noise mixing,
but it removes the SPM exponential in \eqref{eq:autocorr-general}.
The reason for focussing on signal-noise mixing is because this
effect cannot be controlled, other than by reducing power, as opposed
to the deterministic effects of SPM and cross-phase modulation (XPM).
However, in a network environment, the XPM cannot necessarily be
controlled either, and interference can be the main limitation on
capacity~\cite{ekwfg-JLT10}. 
\end{remark}

The bound \eqref{eq:autocorr-bound} is useful when $B$ and $K=B N_A$
are fixed. However, we will also be interested in scaling $B$ with
the launch power. To treat such cases, we keep the $|S(c)|^2$ term from
\eqref{eq:autocorr-general}. We further use \eqref{eq:quad-form-bound} and its
symmetric counterpart to write
\begin{align}
& |u_0|^2+|u'_0|^2-2\rho\Re\{u_0u_0'^{*}\} \nonumber \\
& = \frac{|u_0' - \rho u_0|^2 + |u_0 - \rho u_0'|^2}{2}
+ \frac{|u_0|^2 + |u_0'|^2}{2} \left( 1-\rho^2 \right) .
\label{eq:quad-form-bound3}
\end{align}
We now use \eqref{eq:STbounds1} to bound
\begin{align}
& |A(t,t')| \le |S(c)|^2 \Big[ K z + \nonumber \\
& \quad \left( |u_0| + \frac{|S_I(c)|}{\sqrt{1-\rho^2}} \left| u_0 - \rho u_0' \right|
e^{- \gamma \frac{T_I(c)}{\sqrt{1-\rho^2}} \frac{ |u_0 - \rho u_0'|^2}{2}} \right) 
\nonumber \\
& \quad \left. \left( |u_0'| +  \frac{|S_I(c)|}{\sqrt{1-\rho^2}}  \left| u_0' - \rho u_0 \right|
e^{- \gamma \frac{T_I(c)}{\sqrt{1-\rho^2}} \frac{ |u_0' - \rho u_0|^2}{2} } \right) \; \right]
\nonumber \\
& \quad \exp\left( - \gamma \, T_I(c) \frac{|u_0|^2+|u'_0|^2}{2} \sqrt{1-\rho^2} \right).
\label{eq:autocorr-bound2}
\end{align}
Applying \eqref{eq:xe-bound2} in Appendix~\ref{app:simple-bounds}, we have
\begin{align}
|A(t,t')| & \le |S(c)|^2 \left[ K z + \left( |u_0| + \delta \right) \left( |u_0'| +  \delta \right) \, \right] \nonumber \\
& \quad \exp\left( - \gamma \, T_I(c) \frac{|u_0|^2+|u'_0|^2}{2} \sqrt{1-\rho^2} \right)
\label{eq:autocorr-bound3}
\end{align}
where
\begin{align}
\delta =  \sqrt{\frac{S_I(c)^2}{e \,\gamma \, T_I(c) \sqrt{1-\rho^2}}} .
\label{eq:STratio}
\end{align}
Using \eqref{eq:SbyTbound} with \eqref{eq:x-def}, we have
\begin{align}
   \delta & \le \sqrt{\frac{3Kz}{4e}}
   %, \, \gamma^{3/4} (K/2)^{5/4} z^2 \sqrt{\frac{3}{2e}} \right) \nonumber \\
   \approx1.4 \times 10^{-3}
   \label{eq:SITI}
\end{align}
where the approximation is for the parameters of Table~\ref{table:parameters}.
We here have $|u_0|+\delta \approx |u_0|$ for ``large" signal powers such as
$|u_0|^2 \ge 1\times 10^{-3}$ Watts (or 0 dBm).

%=============================================================
%=============================================================
\section{Rectangular Pulses}
\label{sec:rectangular}
Consider rectangular pulses, for which the SPM term
$$\exp\left( j \gamma \, T_R(c) \left[ |u_0|^2 - |u_0'|^2 \right] \right)$$
in \eqref{eq:autocorr-general} is unity. Rectangular pulses are thus convenient
for studying the spectral broadening characterized by the signal-noise mixing
exponential in \eqref{eq:autocorr-general}.

Consider the fiber parameters shown in Table~\ref{table:parameters}
(see~\cite[Tables I-III]{ekwfg-JLT10}). We have $\sqrt{\gamma K z^2} \approx 0.130$
so that the approximations \eqref{eq:Sapprox}-\eqref{eq:autocorr-approx} are accurate.
We further have $\kappa\approx 28.6$ so that  the second exponential of
\eqref{eq:autocorr-approx} is small (less than $1/e$) for power levels beyond
$|u_0|^2 = |u_0'|^2 = 1/\kappa \approx 0.035$ Watts (or 15.4 dBm) if $\rho=0$.

%%%%%%%%%%%%%%%%%%%%%%%%%%%%%%%%%%%%%%%%
\begin{table}[t]
\begin{center}
  \caption{Fiber Parameters}
  \label{table:parameters}
  \begin{tabular}{ | l | c | }
    \hline %\hline
    %Loss coefficient $\alpha_{dB}$ & 0.2 dB/km \\ \hline
    %Nonlinear refractive index $n_2$ & $2.5 \times 10^{-20}$ m$^2$/W \\ \hline
    %Effective area $A_{\rm eff}$ & 80 $\mu$m$^2$ \\ \hline
    Nonlinear coefficient $\gamma$ & 1.27 (W-km)$^{-1}$ \\ \hline
    %Signal frequency $f_0$ & 193.41 THz ($\lambda_s=1550$ nm) \\ \hline
    %Spontaneous emission factor $n_{\rm sp}$ & 1 \\ \hline
    %Raman pump frequency $f_p$ & 206.75 THz ($\lambda_s=1450$ nm) \\ \hline
    Fiber length & 2000 km \\ \hline
    Temperature $T_e$ & 300 Kelvin \\ \hline
    Receiver noise PSD $N_0$ & $4.142 \times 10^{-21}$ W/Hz \\ \hline
    OA noise PSDD $N_A$ & $6.674 \times 10^{-24}$ W/Hz/m \\ \hline
    OA bandwidth $B$ & 500 GHz \\ %\hline 
    %Symbol rate $1/T_s$ & 100 Gbaud \\
    \hline %\hline
  \end{tabular}
\end{center}
\end{table}
%%%%%%%%%%%%%%%%%%%%%%%%%%%%%%%%%%%%%%%%

%=============================================================
\subsection{Isolated Rectangular Pulse}
\label{subsec:example-rect}
Consider an isolated rectangular pulse
\begin{align}
   u_0(t) & = \left\{ \begin{array}{ll}
   \sqrt{P}, & |t| \le T_s/2 \\
   0, & \text{else}
   \end{array} \right.
   \label{eq:isolated-rectangular}
\end{align}
that has energy $PT_s$ Joules.
%%%%%%%%%%
\begin{comment}
We compute four cases. First, for $|t| \le T_s/2$ and $|t'|\le T_s/2$ we have $u_0=u_0'=\sqrt{P}$ and
\begin{align}
A(t,t') & = |S(c)|^2 \Big[ T_R(c) K \rho \nonumber \\
& \qquad \left. + P \left( S_R(c)^2 + \frac{S_I(c)^2}{1-\rho^2}  (1 - \rho)^2 \right) \right] \nonumber \\
& \quad \exp\left( - \gamma \frac{T_I(c)}{\sqrt{1-\rho^2}} 2 P (1-\rho) \right).
\label{eq:autocorr-rect}
\end{align}
For $|t| \le T_s/2$ and $|t'|> T_s/2$ we have $u_0'=0$ so that
\begin{align}
A(t,t') & = |S(c)|^2 \Big[ T_R(c) K \rho \nonumber \\
& \qquad \left. + P \left( S_R(c) + j \frac{S_I(c)}{\sqrt{1-\rho^2}} \right) \cdot j \frac{S_I(c)}{\sqrt{1-\rho^2}} \rho \right] \nonumber \\
& \quad e^{j \gamma T_R(c) P} \cdot \exp\left( - \gamma \frac{T_I(c)}{\sqrt{1-\rho^2}} P \right).
\label{eq:autocorr-rect2}
\end{align}
\end{comment}
%%%%%%%%%%
The approximation \eqref{eq:autocorr-approx} is
\begin{align}
& \mathcal{A}(t,t') = \nonumber \\
& \left\{ \begin{array}{ll}
   \left( K\rho z + P \right) e^{- \kappa P (1-\rho)}, & |t| \le \frac{T_s}{2}, \; |t'|\le \frac{T_s}{2} \\
   \left( K\rho z + j \gamma K \rho \frac{z^2}{2} P \right) e^{j\gamma z P} e^{- \frac{\kappa}{2} P},
   & |t| \le \frac{T_s}{2}, \; |t'|> \frac{T_s}{2} \\
   \left( K\rho z - j \gamma K \rho \frac{z^2}{2} P \right) e^{-j\gamma z P} e^{- \frac{\kappa}{2} P},
   & |t| > \frac{T_s}{2}, \; |t'|\le \frac{T_s}{2} \\
   K\rho z, & |t|> \frac{T_s}{2}, \; |t'|> \frac{T_s}{2}.
   \end{array} \right.
\label{eq:autocorr-recta}
\end{align}
For example, for a linear channel we have $\gamma=0$ and
\begin{align}
\mathcal{A}(t,t') = 
\left\{ \begin{array}{ll}
   K\rho z + P, & |t| \le \frac{T_s}{2}, \; |t'|\le \frac{T_s}{2} \\
   K\rho z, & \text{else}.
   \end{array} \right.
\label{eq:autocorr-recta-g0}
\end{align}
We choose $T_s=10$ ps to correspond to the symbol rate of $100$ Gbaud.

Fig.~\ref{fig:RectAutoplot} shows $|A(t,t')|$ in dB for $t'=0$ and $t'=5.1$ ps,
and for $P=10, 100, 200, 400$ mW. The plot of the amplitude of the approximation
\eqref{eq:autocorr-recta} is visually indistinguishable from the exact
expression \eqref{eq:autocorr-general}. At low power, $A(t,0)$ is almost
the same as the pulse shape \eqref{eq:isolated-rectangular} and
$A(t,5.1 \text{ ps})$ is close to $K\rho z$.
However, for $P=10$ mW the function $|A(t,0)|$ already has a small
bulge at $t=0$. We have thus entered the nonlinear regime where
signal-noise mixing causes spectral broadening. As the power increases
further, $10\log_{10}|A(t,0)|$ develops a sinc pulse shape in the range
$t=[-5,5]$ ps due to the exponential factor $e^{-\kappa P(1-\rho)}$.
The narrow autocorrelation function for $P=400$ mW implies that the
spectrum has broadened considerably.

%%%%%%%%%%%%%%%%%%%%%%%%%%%%%%%%%%%%%%%%
\begin{figure}[t!]
  \centerline{\includegraphics[scale=0.48]{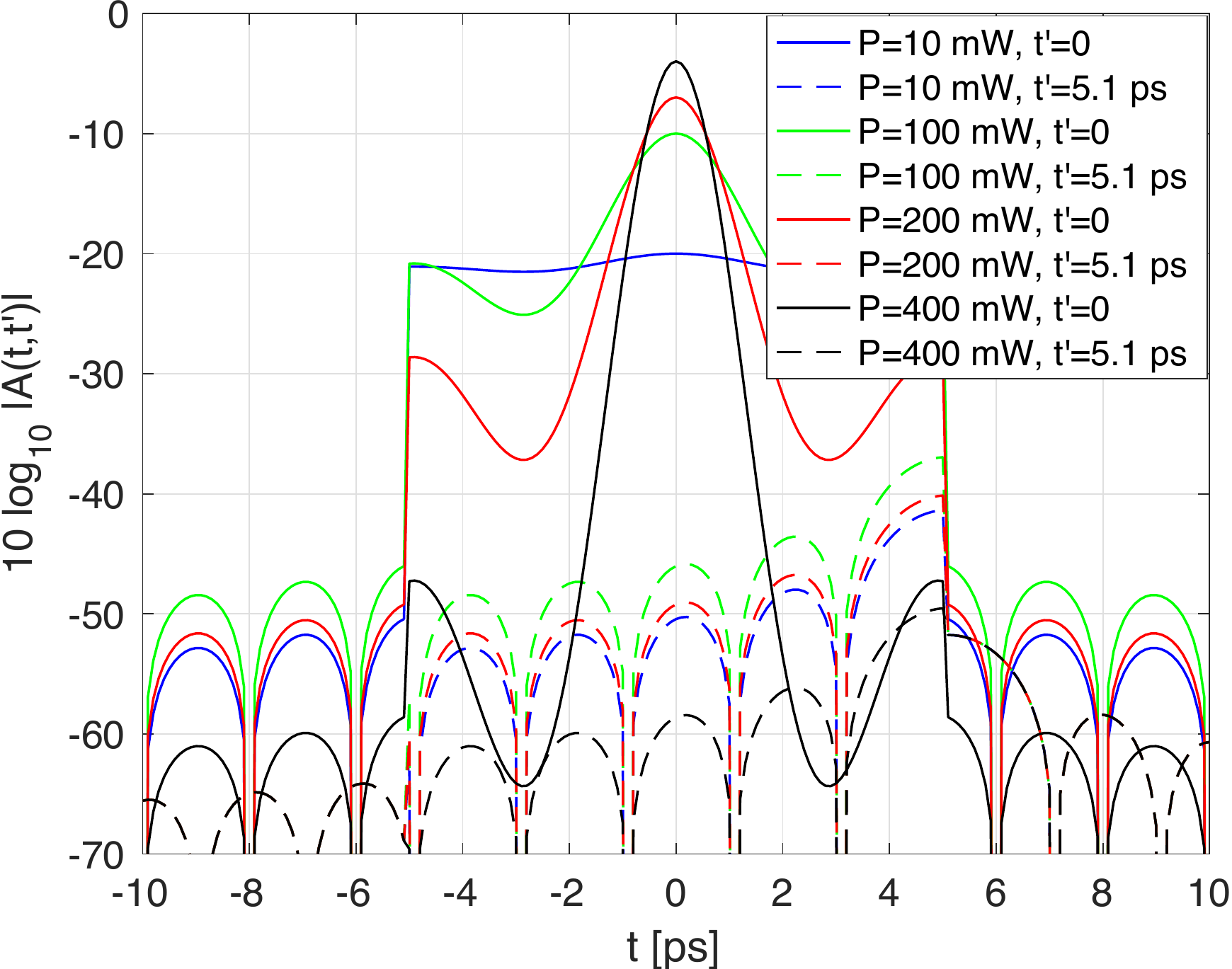}}
  \caption{$\left|A(t,t')\right|$ in dB for the pulse \eqref{eq:isolated-rectangular},
  $t'=0$ and $t'=5.1$ ps, and $P=10,100,200,400$ mW.}
  \label{fig:RectAutoplot}
\end{figure}
\begin{figure}[t!]
  \centerline{\includegraphics[scale=0.48]{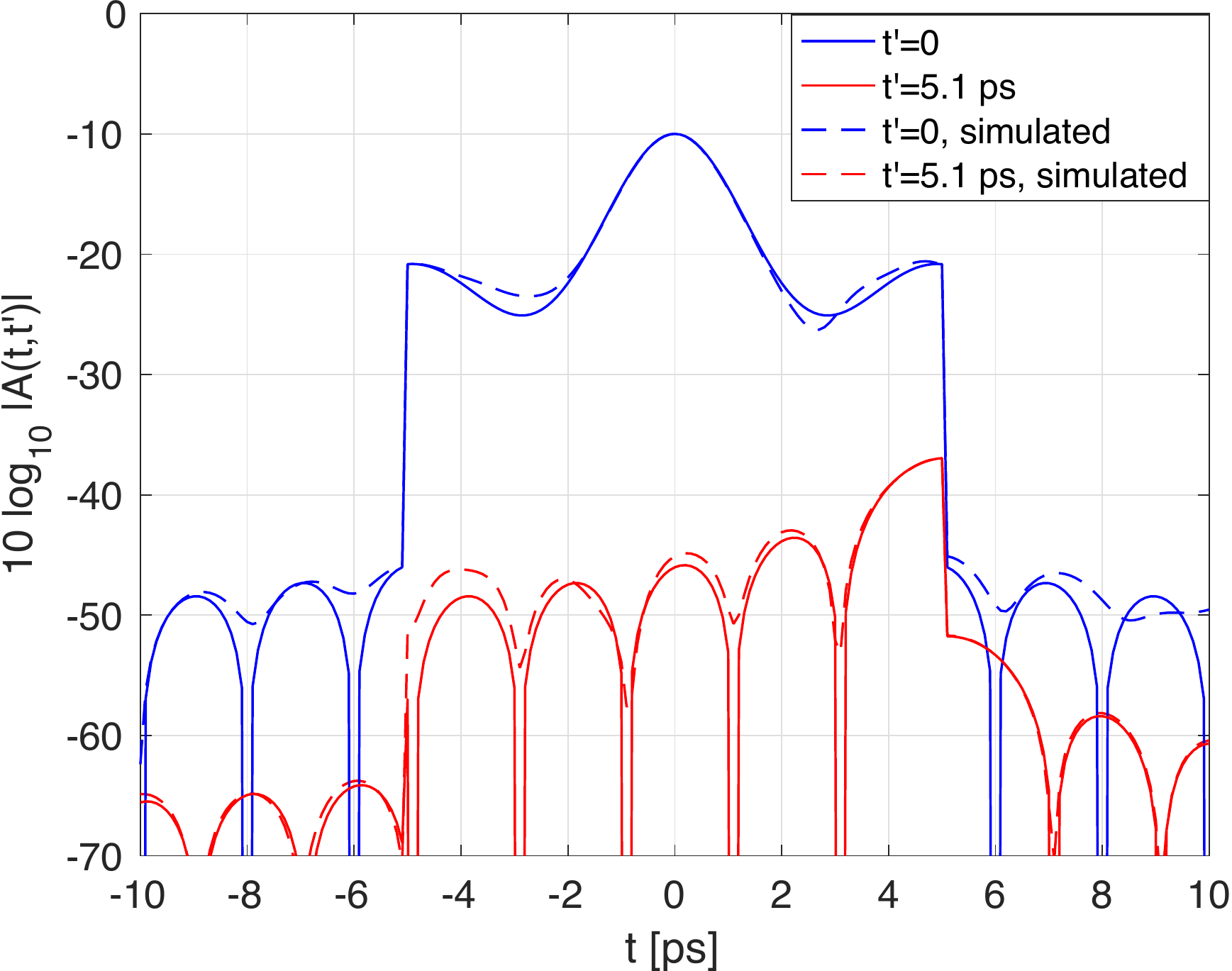}}
  \caption{Computed and simulated $\left|A(t,t')\right|$ in dB for the pulse \eqref{eq:isolated-rectangular},
  $t'=0$ and $t'=5.1$ ps, and $P=100$ mW. }
  \label{fig:RectAuto100plot}
\end{figure}  
%%%%%%%%%%%%%%%%%%%%%%%%%%%%%%%%%%%%%%%%

Fig.~\ref{fig:RectAuto100plot} shows the exact and simulated
autocorrelation functions for $P=100$ mW. The simulated curves
are for the averaged $A(t,t')$ from $10^{4}$ Monte Carlo
simulations of the noise. The curves are in good agreement;
note that the y-axis is logarithmic.

%=============================================================
%=============================================================
\subsection{PAM with Rectangular Pulses and Ring Modulation}
\label{subsec:ring-modulation}
Consider PAM in \eqref{eq:PAM} with rectangular pulses that are time-limited to $[0,T_s)$.
We study a constant amplitude $\sqrt{P}$, and a phase that is
uniformly distributed over $[-\pi,\pi)$, i.e., ring modulation or phase-shift keying (PSK).
Using \eqref{eq:autocorr-approx} with $u_0(t)=\sqrt{P}e^{j\Phi(t)}$
and $u_0(t')=\sqrt{P}e^{j\Phi(t')}$ we have
\begin{align}
\mathcal{A}(t,t')
& = \left( K \rho z + P e^{j \Phi_\Delta} \right)
e^{ - \kappa P \left( 1 - \rho \cos \Phi_\Delta \right)} .
\label{eq:ring-A}
\end{align}
where $\Phi_\Delta=\Phi-\Phi'$.
If $t$ and $t'$ are in the same symbol interval, then we have $\Phi_\Delta=0$;
otherwise $\Phi$ and $\Phi'$ are independent and $\Phi_\Delta$ is uniform.
The average autocorrelation function is therefore
\begin{align}
\bar{\mathcal{A}}(t,t')
& = \left\{ \begin{array}{ll}
\left( K \rho z + P \right) e^{ - \kappa P \left( 1 - \rho \right)}, & \Phi_\Delta=0 \\
\left[ K \rho z \, I_0\left(\kappa P \rho\right) + P \, I_1\left(\kappa P \rho\right) \right]
e^{ - \kappa P}, & \text{else}.
\end{array} \right.
\label{eq:ring-A1}
\end{align}
Note that $\bar{\mathcal{A}}(t,t')$ is real-valued.
We further compute the time-averaged version of \eqref{eq:ring-A1} to be
(see~\eqref{eq:autocorr-stationary})
\begin{align}
\bar{\mathcal{A}}(\tau) & = \left( 1 - \frac{|\tau|}{T_s} \right) \left( K\rho z + P \right)
e^{- \kappa P \left( 1-\rho\right)} \nonumber \\
& \quad + \frac{|\tau|}{T_s} \,
\left[ K \rho z \, I_0\left(\kappa P \rho\right)
+ P \, I_1\left(\kappa P \rho\right) \right] e^{ - \kappa P}
\label{eq:rect-PSD-Atau1}
\end{align}
for $|\tau|<T_s$, and otherwise
\begin{align}
\bar{\mathcal{A}}(\tau) & = \left[ K \rho z \, I_0\left(\kappa P \rho\right)
+ P \, I_1\left(\kappa P \rho\right) \right] e^{ - \kappa P}.
\label{eq:rect-PSD-Atau2}
\end{align}
Note that $\bar{\mathcal{A}}(\tau)$ is a real-valued and even function of $\tau$.

A plot of $|\bar{\mathcal{A}}(\tau)|$ is shown in Fig.~\ref{fig:RectAutotauplot}, along
with the amplitude of the average autocorrelation function from $10^{4}$ Monte Carlo 
simulations of the random signals and noise. The simulations were performed
using \eqref{eq:dispersion-free-model} where $w(\cdot,t)$ and $w(\cdot,t')$ are jointly
Wiener processes for any $t$ and $t'$. 
Observe that $|\bar{\mathcal{A}}(\tau)|$ accurately matches the simulations.
The dash-dotted curve is for $\gamma=0$, and it shows that the nonlinearity has
caused substantial  narrowing of the autocorrelation function, i.e., there is substantial
spectral broadening.

This phenomenon is clearly apparent in Fig.~\ref{fig:RectPAMPSDplot} that plots the PSDs
$\bar{\mathcal{P}}(f)$ for $P=10,50,100,500,1000$ mW,\footnote{The PSDs were
computed with \eqref{eq:rect-PSD-Atau1}-\eqref{eq:rect-PSD-Atau2}.}
as well as the PSD when $P=10$ mW and $\gamma=0$. Recall that the
OA bandwidth is  $B=500$ GHz, and observe that $\bar{\mathcal{P}}(f)$
is large well-beyond this frequency already for $P=100$ mW.
The model is therefore inaccurate at this launch power since
the OA compensates attenuation only for frequencies up to $B$.

%%%%%%%%%%%%%%%%%%%%%%%%%%%%%%%%%%%%%%%%
\begin{figure}[t!]
  \centerline{\includegraphics[scale=0.48]{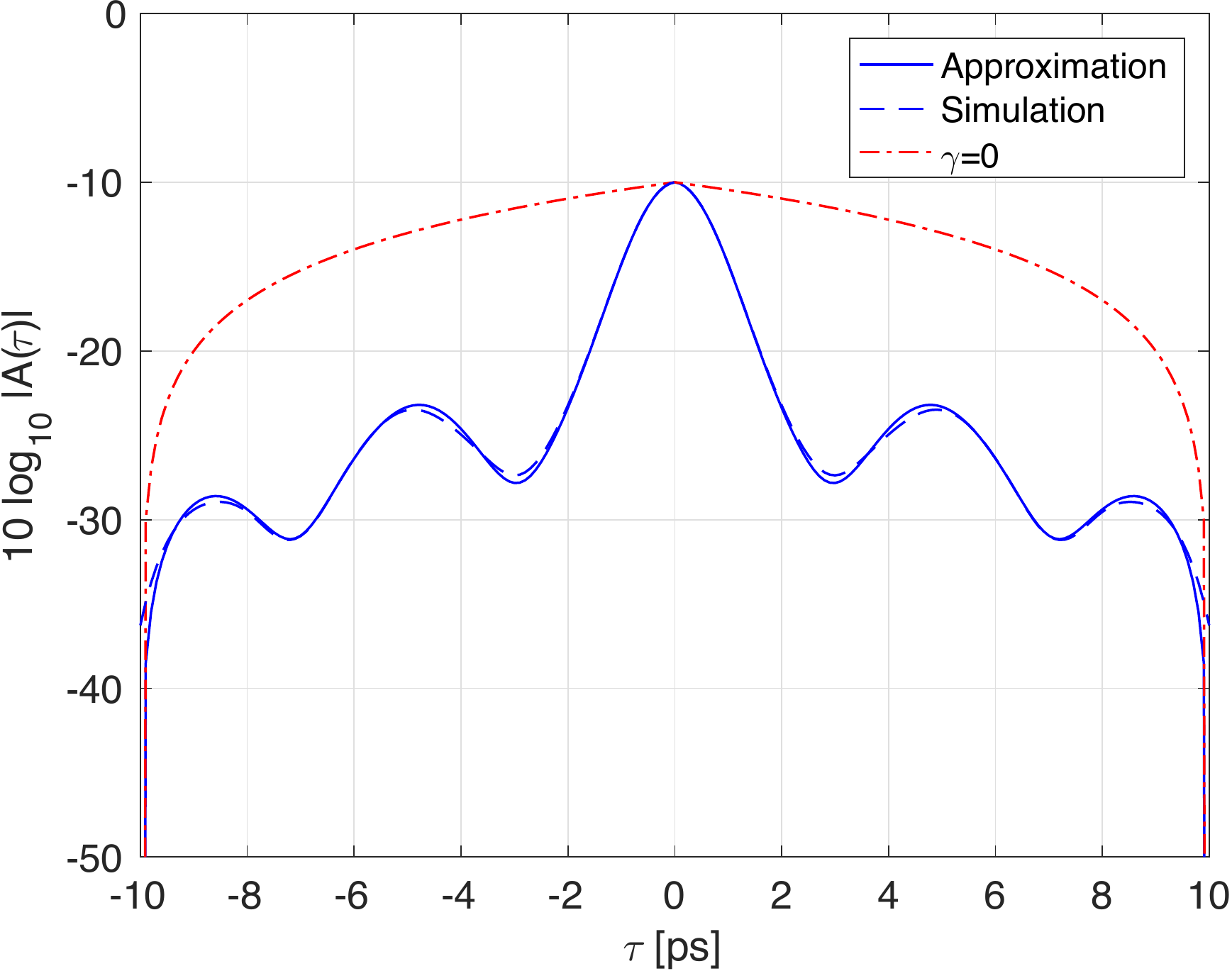}}
  \caption{Approximate and simulated $\left|\bar{A}(\tau)\right|$ in dB for PAM with 
  rectangular pulses, ring modulation, and $P=100$ mW. The dash-dotted curve
  shows $\left|\bar{A}(\tau)\right|$ for the same channel except that $\gamma=0$.}
  \label{fig:RectAutotauplot}
\end{figure}
\begin{figure}[t!]
  \centerline{\includegraphics[scale=0.48]{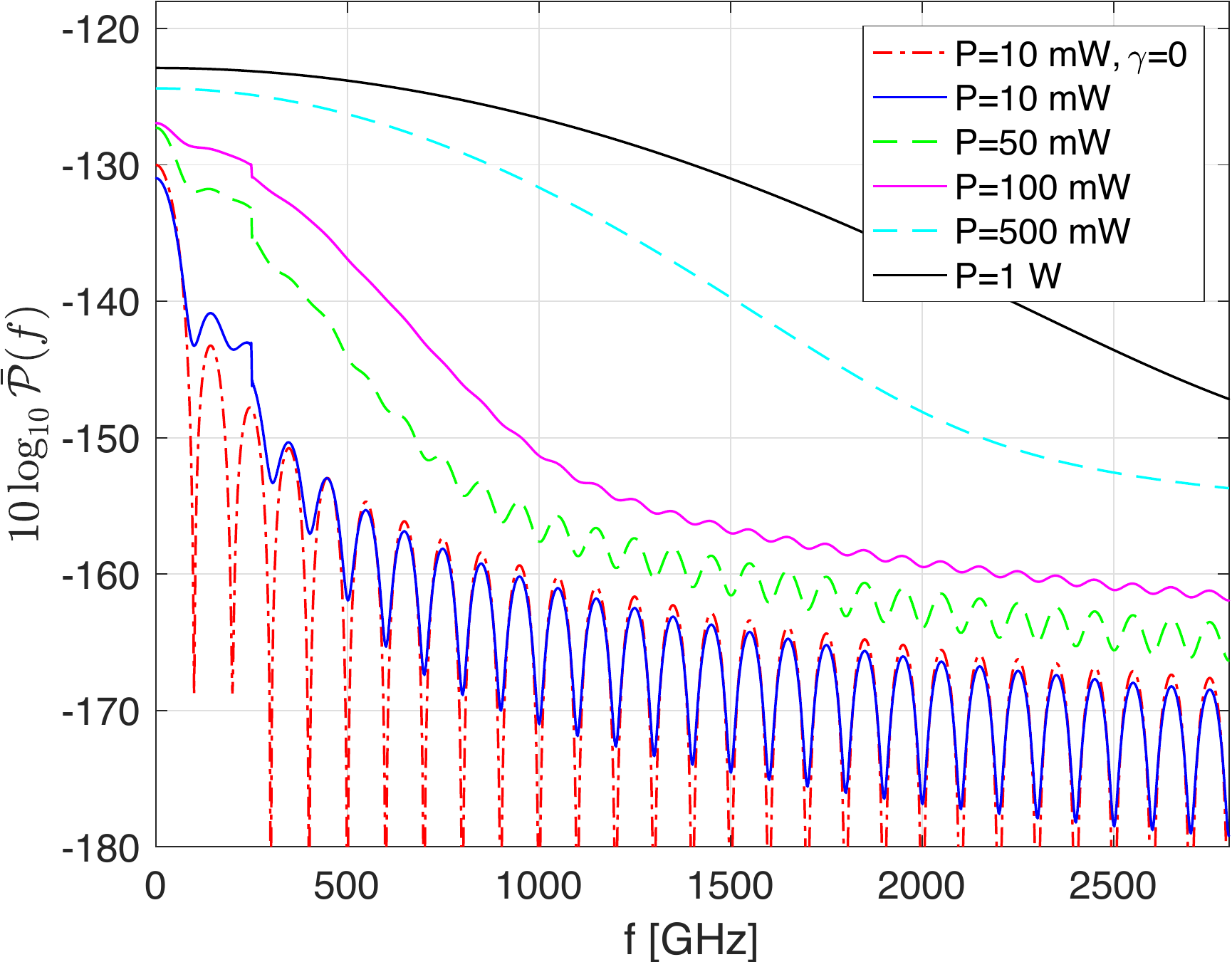}}
  \caption{$\bar{\mathcal{P}}(f)$ in dB for PAM with 
  rectangular pulses, ring modulation, and $P=10,50,100,500,1000$ mW. The dash-dotted curve
  shows $\bar{\mathcal{P}}(f)$ for $P=10$ mW and $\gamma=0$.}
  \label{fig:RectPAMPSDplot}
\end{figure}
%%%%%%%%%%%%%%%%%%%%%%%%%%%%%%%%%%%%%%%%

\begin{remark}
The received signal energy is
$\bar{A}(0) = \bar{\mathcal{A}}(0) = K z + P$, and as $\gamma \rightarrow 0$ we have 
\begin{align}
\bar{A}(\tau) \rightarrow
\left\{ \begin{array}{ll}
K\rho z + \left( 1 - \frac{|\tau|}{T_s} \right) P, & |\tau| \le T_s \\
K\rho z, & \text{else}.
\end{array} \right. \label{eq:PAM-rect-limit}
\end{align}
The same limiting expression  \eqref{eq:PAM-rect-limit} is valid for $\bar{\mathcal{A}}(\tau)$.
\end{remark}

\begin{comment}
\begin{remark}
Some simple bounds are
\begin{align}
   0 \le I_1(x) < I_0(x)\le e^x \label{eq:Ibounds}
\end{align}
and we obtain the following bound that is valid for all $\tau$:
\begin{align}
\bar{\mathcal{A}}(\tau)
& \le \left( Kz + P \right) e^{- \kappa P \left( 1-\rho \right)} .
\label{eq:A-tau-bound}
\end{align}
\end{remark}
\end{comment}

%=============================================================
%=============================================================
\section{Power, Bandwidth, and Energy Bounds}
\label{sec:P-E-scaling}
This section studies the power and energy at the output of bandlimited and
time-resolution limited receivers. The analysis ultimately lets us bound
the propagating signal bandwidth as a function of the launch power.
A simple but useful bound for low launch power is
\begin{align}
   \bar{P}_r(W,T) & \le \bar{P}_r(\infty,T) = Kz + \bar{P}_T \le Kz + P
   \label{eq:PrT-bound1}
\end{align}
where the second inequality is from \eqref{eq:P-constraint}.

%=============================================================
\subsection{Bandlimited Receiver}
\label{subsec:bandlimited-receiver}
Consider \eqref{eq:PrT-bound} and split the double integral into
three parts: one where $|u_0|>|u_0'|$, one where $|u_0|<|u_0'|$,
and one where $|u_0|=|u_0'|$. The first two double integrals are
identical due to the symmetry in the arguments of the integrand,
i.e., for every pair $(t,t')=(t_1,t_2)$ where $|u_0|>|u_0'|$ there is
a pair $(t,t')=(t_2,t_1)$ for which $|u_0|<|u_0'|$ and
$$\left| \bar{A}(t_1,t_2) \, b(t_1-t_2)\right|
= \left| \bar{A}(t_2,t_1) \, b(t_2-t_1) \right|.$$
In other words, using \eqref{eq:PrT-bound} we have
\begin{align}
   \bar{P}_r(W,T) & \le \frac{1}{T} \int_{-T/2}^{T/2} \int_{-T/2}^{T/2}
   |\bar{A}(t,t') \, b(t-t')| \, dt' dt \nonumber \\
   & = \frac{2}{T} \iint_{\mathcal{I}_1}
   |\bar{A}(t,t') \, b(t-t')| \, dt' dt \nonumber \\
   & \quad + \frac{1}{T} \iint_{\mathcal{I}_2}
   |\bar{A}(t,t') \, b(t-t')| \, dt' dt
   \label{eq:PrT-bound2}
\end{align}
where
\begin{align}
   & {\mathcal I}_1 = \{(t,t'): |u_0|>|u_0'|, \, |t| \le T/2, \, |t'| \le T/2 \} \\
   & {\mathcal I}_2 = \{(t,t'): |u_0|=|u_0'|, \, |t| \le T/2, \, |t'| \le T/2 \}.
\end{align}
We further have $|\bar{A}(t,t')| = |\E{A(t,t')}| \le \E{|A(t,t')|}$.
Thus, using \eqref{eq:autocorr-bound3} and $|u_0|\ge |u_0'|$
for all $(t,t') \in {\mathcal I}_1 \cup {\mathcal I}_2$, we have
\begin{align}
   \bar{P}_r(W,T) & \le \frac{1}{T} \int_{-T/2}^{T/2} \E{P_r(W,T,t)}  \, dt
   \label{eq:PrT-bound2a}
\end{align}
where
\begin{align}
   & P_r(W,T,t) = 2 \left[ K z + \left(\sqrt{P_t} + \delta \right)^2 \right] \nonumber \\
   & \int_{-T/2}^{T/2} 
      |S(c)|^2 \exp\left( - \gamma \, T_I(c) \frac{P_t}{2} \sqrt{1-\rho^2} \right) \,
      \left| b(t-t') \right| \, dt'
   \label{eq:PrT-integrand}
\end{align}
and where we have defined $P_t=|u_0(t)|^2$.
%

%=============================================================
\subsection{Bounds on the Instantaneous Received Power}
\label{subsec:bandlimited-receiver-bounds}
Recall that $W$ is the receiver bandwidth and $B$ is the OA bandwidth.
In Appendix~\ref{app:Pt-lemma-proofs},  we prove the following lemmas that bound
the instantaneous received power $P_r(W,T,t)$.
We distinguish cases where $x\ge1$ cannot or can occur.
\begin{itemize}
\item Lemma~\ref{lemma:WlessB-lemma1} applies to the usual case with $W\le B$ and
low noise-nonlinearity-distance product, i.e., $\gamma (K/2) z^2\le 1$.
\item Lemma~\ref{lemma:WlessB-lemma2} is for $W\le B$ and $\gamma (K/2) z^2 \ge 1$.
\item Lemma~\ref{lemma:WgreaterB-lemma} is for $W\ge B$ and $\gamma (K/2) z^2\le 1$.
\end{itemize}
We do not consider the case $W\ge B$ and $\gamma (K/2) z^2\ge 1$ because this case
is slightly more complicated than the others, and because we are mainly interested in $W\le B$.

%------------------
\begin{lemma} \label{lemma:WlessB-lemma1}
If $W\le B$ and $\gamma (K/2) z^2\le 1$, then we have
\begin{align}
   & P_r(W,T,t) \le 4 \left[ Kz + (\sqrt{P_t}+\delta)^2 \right] \nonumber \\
   & \left[ \frac{2\, W/B}{\sqrt{(\kappa/8) P_t}} 
   \frac{\sqrt{\pi}}{2} {\rm erf}\left( \sqrt{(\kappa/8) P_t} \right)
   + 5e^{- \sqrt{\gamma K z^2} - (\kappa/9) P_t } \right] .
   \label{eq:WlessB-lemma1}
\end{align}
\end{lemma}
%------------------
\begin{lemma} \label{lemma:WlessB-lemma2}
If $W\le B$ and $\gamma (K/2) z^2\ge 1$, then we have
\begin{align}
   & P_r(W,T,t) \le 4 \left[ Kz + (\sqrt{P_t}+\delta)^2 \right]  \nonumber \\
   & \quad \left[ \frac{2\, W/B}{\sqrt{(\kappa/8) P_t}}
      \frac{\sqrt{\pi}}{2} {\rm erf}\left( \sqrt{\frac{P_t}{3 Kz}}\right)
      \right. \nonumber \\
      & \qquad \left. + \frac{25 \, W/B}{\gamma K z^2 } e^{- \frac{1}{\sqrt{18} K z} P_t }
       + 5e^{- \sqrt{\gamma K z^2} - \sqrt{\frac{\gamma}{20 K}} P_t} \right]
   \label{eq:WlessB-lemma2}
\end{align}
\end{lemma}
%------------------
\begin{lemma} \label{lemma:WgreaterB-lemma}
If $W\ge B$ and $\gamma (K/2) z^2\le 1$, then we have
\begin{align}
   & P_r(W,T,t) \le 4 \left[ Kz + (\sqrt{P_t}+\delta)^2 \right] \nonumber \\
   & \left[ \frac{2\, W/B}{\sqrt{(\kappa/8) P_t}} 
   \frac{\sqrt{\pi}}{2} {\rm erf}\left( \sqrt{(\kappa/8) P_t} \, \frac{B}{W} \right) \right. \nonumber \\
   & \left. + \frac{1}{4}\left(1-\frac{B}{W}\right) e^{- (\kappa/8) P_t (B/W)^2}
      + 5e^{- \sqrt{\gamma K z^2} - (\kappa/9) P_t} \right] .
     \label{eq:WgreaterB-lemma}
\end{align}
\end{lemma}
%------------------
\begin{remark}
The above bounds are valid for \emph{any} launch signal.
The bounds may be very loose, e.g., for launch signals
with bandwidth larger than $B$.
\end{remark}
\begin{remark}
We have
\begin{align}
\begin{array}{rl}
(\sqrt{\pi}/2) \, {\rm erf}(y) \approx y, & \text{small } |y| \\
{\rm erf}(y) \approx 1, & \text{large } y.
\end{array}
\end{align}
Thus, for large $P_t$, the received power scales at most as $\sqrt{P_t}$
for all the regimes considered in
Lemmas~\ref{lemma:WlessB-lemma1}-\ref{lemma:WgreaterB-lemma}.
The power loss factor of $1/\sqrt{P_t}$ is due to signal-noise mixing,
and the square-root character of the power loss is due to the quadratic
behavior of $\rho(t-t')=\sinc(B(t-t'))$ near $t=t'$. The shape of the OA
noise PSD thus directly affects the power scaling.
\end{remark}
\begin{remark}
For small $P_t$ or small $\gamma$, we know that the instantaneous
received power $P_r(W,T,t)$ can be $Kz+P_t$, as we expect for a memoryless,
noisy, linear channel. For example, for small $P_t$ the RHS
of~\eqref{eq:WlessB-lemma1} approaches
\begin{align}
   4 \left[ Kz + \left( \sqrt{P_t}+\delta \right)^2  \right]
   \left[ 2\,\frac{W}{B} + 5 e^{- \sqrt{\gamma K z^2} } \right] .
   \label{eq:WlessB-lemma1-smallP}
\end{align}
Note that $W/B$ can be small but the term $5e^{- \sqrt{\gamma K z^2}}$ is larger than one
if $\gamma (K/2) z^2\le 1$.  However, for $\gamma (K/2) z^2\ge 1$
and small $P_t$, the RHS of \eqref{eq:WlessB-lemma2} approaches
\begin{align}
   4 \left[ Kz + \left( \sqrt{P_t}+\delta \right)^2 \right]
   \left[ \frac{29}{\gamma K z^2} \, \frac{W}{B} + 5e^{- \sqrt{\gamma K z^2} } \right] .
   \label{eq:WlessB-lemma2-smallP}
\end{align}
Now the receiver may put out less power than $Kz+P_t$ due
to the large noise-nonlinearity-distance product.
\end{remark}
\begin{remark}
If $B$ is very large, then the $W/B$ terms in
\eqref{eq:WlessB-lemma1}-\eqref{eq:WgreaterB-lemma} are small.
We may thus encounter $P_r(W,T,t)$ with an exponential
behavior in $P_t$, see Remark~\ref{rmk:infinite-bandwidth-noise}
and Appendix~\ref{app:infiniteB}.
\end{remark}
\begin{remark}
The case $W > B$ has the receiver measuring signals in bands where
there is no attenuation yet the noise is small, as discussed in the introduction and
Sec.~\ref{subsec:large-capacity}. Moreover, for fixed $P_t$,
fixed $B$, and large $W$ we approach the regime of the per-sample receiver
where \eqref{eq:WgreaterB-lemma} becomes
\begin{align}
   & P_r(\infty,T,t) \nonumber \\
   & \le \left[ Kz + (\sqrt{P_t}+\delta)^2 \right]
      \left[ 9 + 20 e^{- \sqrt{\gamma K z^2} - (\kappa/9) P_t} \right] .
     \label{eq:WgreaterB-infinite}
\end{align}
The correct answer on the RHS of \eqref{eq:WgreaterB-infinite} is
$Kz+P_t$; the extra factors are due to loose bounding steps
that were designed for large $P_t$.
\end{remark}

%=============================================================
\subsection{Bounds on the Average Received Power}
\label{subsec:average-power}
We continue to study the case $W \le B$ as the range of practical and theoretical
interest.  We would next like to develop a bound on the {\em average} received power
$\bar{P}_r(W,T)$ as a function of the maximum average launch power $P$, see~\eqref{eq:P-constraint}.
For this purpose, define the function
\begin{align}
   f(s,P) = \left[ Kz + (\sqrt{P}+\delta)^2 \right]
             \frac{\sqrt{\pi}}{2} \frac{{\rm erf}(\sqrt{sP})}{\sqrt{sP}}
\label{eq:f-function}
\end{align}
and an offset power $P_o=3(Kz + \delta^2)$.
In Appendix~\ref{app:Pt-lemma-proofs},  we prove the following lemmas.
%------------------
\begin{lemma} \label{lemma:WlessB-lemma3}
If $W\le B$ and $\gamma (K/2) z^2\le 1$, then we have
\begin{align}
   & \bar{P}_r(W,T) \le c_1 + \frac{8W}{B} f\left(\frac{\kappa}{8}, P + P_o \right)
   \label{eq:PrT-bound3}
\end{align}
where
\begin{align}
 c_1 & = 20 \left[ Kz + \delta^2 + \sqrt{\frac{18}{\kappa e}} \delta + \frac{9}{\kappa e} \, \right]
              e^{- \sqrt{\gamma K z^2}} .
              \label{eq:c1}
\end{align}
Furthermore, the RHS of \eqref{eq:PrT-bound3} is non-decreasing and concave in $P$.
\end{lemma}
%------------------
\begin{lemma} \label{lemma:WlessB-lemma4}
If $W\le B$ and $\gamma (K/2) z^2\ge 1$, then we have
\begin{align}
   \bar{P}_r(W,T) \le \frac{W}{B} c_2 + c_3 + \frac{16(W/B)}{\gamma K z^2}
                     f\left( \frac{1}{3Kz}, P + P_o \right)
   \label{eq:PrT-bound4}
\end{align}
where
\begin{align}
 c_2 & = \frac{100}{\gamma K z^2 } \left[ Kz + \delta^2 + \sqrt{\frac{6Kz}{e}} \delta
              +  \frac{\sqrt{18}Kz}{e} \, \right] \nonumber \\
 c_3 & = 20 \left[ Kz + \delta^2 + \left(\frac{80 K}{\gamma e^2}\right)^{1/4} \delta
              +  \sqrt{\frac{20K}{\gamma e^2}} \, \right] e^{- \sqrt{\gamma K z^2}} .
\end{align}
Furthermore, the RHS of \eqref{eq:PrT-bound4} is non-decreasing and concave in $P$.
\end{lemma}
%------------------
\begin{remark}
The values $c_1$, $c_2$, and $c_3$ are independent of $P$ and $W$,
but they depend on $B$, $\gamma$, and $z$.
\end{remark}
\begin{remark}
The RHSs of \eqref{eq:PrT-bound3} and \eqref{eq:PrT-bound4}
scale as $\sqrt{P}$ for large $P$.
\end{remark}

%=============================================================
\subsection{Propagating Signal Bandwidth}
\label{subsec:propagating-bandwidth}
We proceed to develop a bound on the propagating signal bandwidth,
which we also write as $W$ (in the previous sections, the parameter
$W$ represented the receiver filter bandwidth). We are particularly interested in
large $P$ where spectral broadening occurs.
We interpret the regime $W\le B$ as being
``practically relevant" and $W>B$ as being ``impractical".\footnote{This
definition does not always make sense, e.g., for very noisy signals
where the useful part of the signal has small bandwidth.}

The average total received power for a linear channel is $Kz + P$.
Suppose we require that 99\% of this power is inside the band 
$f\in [-W/2,W/2]$, i.e., we require
\begin{align}
   & \bar{P}_r(W,T) \ge 0.99 (Kz + P).
   \label{eq:PrT-bound5}
\end{align}
We remark that the value 99\% is not crucial; the results below
remain valid for any other choice near 100\%.

Consider first $W\le B$ and $\gamma (K/2) z^2\le 1$.
Combining \eqref{eq:PrT-bound3} and \eqref{eq:PrT-bound5},
and using ${\rm erf}(y) \le 1$, we have (see \eqref{eq:PrT-bound3-app-1})
\begin{align}
   \frac{W}{B} & \ge \frac{0.99(Kz + P) - c_1}
   {8 f\left( \kappa/8, P + P_o \right)} \label{eq:W-bound1-a} \\
   & \ge \frac{\left\{ 0.99(Kz + P) - c_1 \right\}
    \sqrt{(\kappa/8) \left( P+P_o \right)}}
   {8 \left[Kz + \left( \sqrt{P+ P_o}+\delta \right)^2\right]}.
   \label{eq:W-bound1}
\end{align}
Thus, for fixed $B$ and large $P$, we find that $W$ scales at
least as a constant times $\sqrt{P}$. This means that there is some power
threshold for which $W > B$. We conclude that
\emph{the model loses practical relevance beyond some
launch power threshold.}

An upper bound on the threshold follows by computing the $P$ for which
the RHS of \eqref{eq:W-bound1-a} is one. For example, for the parameters
in Table~\ref{table:parameters}, we compute $P \le 18.6$ Watts.
However, the power 18.6 Watts seems unrealistically large, which suggests that
our bounds are very loose.
Fig.~\ref{fig:RectPAMPSDplot} also suggests that the bound is loose,
since there is substantial spectral broadening already at $P=50$ mW.
However, recall that the bounds \eqref{eq:W-bound1-a}-\eqref{eq:W-bound1}
are valid for any launch signal, and not only PAM with rectangular pulses and ring modulation.

Consider next $W\le B$ and $\gamma (K/2) z^2\ge 1$.
Combining \eqref{eq:PrT-bound4} and \eqref{eq:PrT-bound5},
and using ${\rm erf}(y) \le 1$, we have
\begin{align}
   & \frac{W}{B} \ge \frac{0.99(Kz + P) - c_3}
   {c_2 + \frac{16}{\gamma K z^2} f\left( \frac{1}{3Kz}, P + P_o \right)} \label{eq:W-bound2-a} \\
   & \ge \frac{\left\{ 0.99(Kz + P) - c_3 \right\}
    \sqrt{P+P_o}}{c_2 \sqrt{P+P_o}
   + \sqrt{\frac{512}{\kappa}} \left[Kz + \left( \sqrt{P+ P_o}+\delta \right)^2\right]}.
   \label{eq:W-bound2}
\end{align}
Thus, for fixed $B$ and large $P$, we again find that
$W$ scales at least as a constant times $\sqrt{P}$. We again conclude
that the model loses practical relevance beyond some
launch power threshold.

%=============================================================
%=============================================================
\subsection{Distributed Amplification Bandwidth}
\label{subsec:distributed-bandwidth}
The bounds \eqref{eq:W-bound1-a}-\eqref{eq:W-bound2} let us study
whether we can increase the range of practically relevant $P$ by increasing $B$. We
show that this is not possible in general. In fact, as $B$ increases we
must limit ourselves to progressively smaller $P$, while at the same time
dealing with more noise power $Kz=N_A B z$.

We study the following problem. Suppose the OA bandwidth scales as
$B=P^\beta$ for some non-negative constant $\beta$. For large
$P$, we thus study the case $\gamma (K/2) z^2\ge 1$
where the relevant bounds are \eqref{eq:PrT-bound4} and
 \eqref{eq:W-bound2-a}-\eqref{eq:W-bound2}.
Note that $K$, $\kappa$, and $P_o$ are proportional to $B$, while $\delta$
is proportional to $\sqrt{B}$. Thus, $c_2$ remains a constant and $c_3$
vanishes for large  $B$. Inserting $B=P^\beta$ into
\eqref{eq:PrT-bound4}, the scaling behavior of $\bar{P}_r(W,T)$ for large
$P$ is bounded as
\begin{align}
   \bar{P}_r(W,T) \lesssim \left\{ \begin{array}{ll}
      P^{(1-3\beta)/2}, & 0 \le \beta \le 1 \\ 
      P^{-\beta}, & \beta \ge 1 .
      \end{array} \right.
   \label{eq:PrT-bound4-1}
\end{align}
The average receiver power thus \emph{decreases} with the average
launch power if $\beta>1/3$ and $P$ is sufficiently large.

Next, inserting $B=P^\beta$ into  \eqref{eq:W-bound2-a},
the scaling behavior of $W$ is bounded as
\begin{align}
   W \gtrsim \left\{ \begin{array}{ll}
      P^{(1+3\beta)/2}, & 0 \le \beta \le 1 \\ 
      P^{2\beta}, & \beta \ge 1 .
      \end{array} \right.
   \label{eq:PrT-bound4-2}
\end{align}
The condition $W \le B$ for large $P$ requires $P^{(1+3\beta)/2} \lesssim P^\beta$ for $0 \le \beta \le 1$,
or $P^{2\beta} \lesssim P^\beta$ for $\beta \ge 1$, neither of which is possible.
We conclude that \emph{there is no scaling of $B$ through which we can make the
model practically relevant for large launch power.}

%=============================================================
%=============================================================
\subsection{Time-Resolution Limited Receiver}
\label{subsec:time-resolution-limited}
Bounds for the time-resolution limited receiver can be developed using the
same steps as those for the bandlimited receiver. For instance, using the
same steps as in \eqref{eq:PrT-bound2} but with \eqref{eq:Erbar-def} rather than
\eqref{eq:PrT-bound}, we have the analog of
\eqref{eq:PrT-bound2a}-\eqref{eq:PrT-integrand}, namely
\begin{align}
 \bar{E}_m(T_r) & \le \frac{1}{T_r} \int_{mT_r}^{(m+1)T_r} \E{E_m(T_r,t)} \, dt
  \label{eq:Em-bound3}
\end{align}
where
\begin{align}
  & E_m(T_r,t) = 2 \left[ K z + \left(\sqrt{P_t} + \delta \right)^2 \right] \nonumber \\
  & \quad \int_{mT_r}^{(m+1)T_r} 
      |S(c)|^2 \exp\left( - \gamma \, T_I(c) \frac{P_t}{2} \sqrt{1-\rho^2} \right) \, dt' .
  \label{eq:Em-integrand}
\end{align}
Next, by following similar steps as \eqref{eq:gTI-bound}-\eqref{eq:PrWTt-bound-WlessB}
that were used to derive \eqref{eq:WlessB-lemma1}, for $T_r \ge 1/B$ we have
\begin{align}
   E_m(T_r,t) & \le 4 \left[ Kz + (\sqrt{P_t}+\delta)^2 \right] \nonumber \\
   & \quad \left[ \frac{1/B}{\sqrt{(\kappa/8) P_t}} 
   \frac{\sqrt{\pi}}{2} {\rm erf}\left( \sqrt{(\kappa/8) P_t} \right) \right. \nonumber \\
   & \qquad \left. + 5 \left( T_r - \frac{1}{B} \right) e^{- \sqrt{\gamma K z^2} - (\kappa/9) P_t } \right] .
   \label{eq:TlessB}
\end{align}
Note that there is no extra factor of two in front of the ${\rm erf}(\cdot)$ term,
cf.~\eqref{eq:WlessB-lemma1}, because we do not need to use the filter \eqref{eq:Bfreq}.

For $T_r < 1/B$, we have
\begin{align}
   E_m(T_r,t) & \le 4 \left[ Kz + (\sqrt{P_t}+\delta)^2 \right] \nonumber \\
   & \quad \frac{1/B}{\sqrt{(\kappa/8) P_t}} \frac{\sqrt{\pi}}{2} {\rm erf}\left( \sqrt{(\kappa/8) P_t} B T_r \right)
     \label{eq:TlessB2}
\end{align}
which is simpler than \eqref{eq:WlessB-lemma2} because there is only
one integration  interval, rather than four as in Appendix~\ref{app:Pt-lemma-proofs},
see~\eqref{eq:WlessB-part2}.
As before, for large $P_t$ the energy $E_m(T_r,t)$ scales at most as 
$\sqrt{P_t}$. The same claim is valid for the average energy $\bar{E}_m(T_r)$
by using the concavity steps in Appendix~\ref{app:Pt-lemma-proofs},
see~\eqref{eq:form-app}-\eqref{eq:PrT-bound3-app-a}.

\begin{remark}
The time resolution $T_r$ must scale to zero at least as fast as
$1/\sqrt{P_t}$ to have the RHS of \eqref{eq:TlessB2} scale as $P_t$. 
\end{remark}
\begin{remark}
Consider PAM and fixed $P_t$. As $T_r$ decreases to zero, the RHS of
\eqref{eq:TlessB2} becomes
\begin{align}
 4 \left[ Kz + (\sqrt{P_t}+\delta)^2 \right] T_r
\end{align}
which decreases to zero. However, there are $T_s/T_r$ samples per
transmitted symbol, so the energy collected per symbol is proportional
to $T_s$. 
\end{remark}
\begin{remark}
If $\gamma \rightarrow 0$ then the RHS of \eqref{eq:TlessB2} becomes
$4(Kz + P_t)T_r$. This is loose by a factor of four:
a factor of two is from the step corresponding to \eqref{eq:PrT-bound2},
and another factor of two is from the
step corresponding to \eqref{eq:PrWTt-bound-WlessB}$(b)$
where the interval ${\mathcal I}_1$ was enlarged.
We show in Appendix~\ref{app:PAM-Rect} how
to improve these steps to obtain the expected  $(Kz + P_t)T_r$
for PAM with rectangular pulses, ring modulation, and $T_r=T_s=1/B$.
\end{remark}

%=============================================================
%=============================================================
\section{Capacity Upper Bounds}
\label{sec:capacity}
The capacity result \eqref{eq:general-capacity-bound} implies that
\begin{align}
  C(W)/W \le \log_2\left( 1 + \frac{\bar{P}_r(W)}{W N_0}  \right) \text{ bits/s/Hz}.
  \label{eq:general-capacity-bound1}
\end{align}
We may thus use \eqref{eq:PrT-bound1}, \eqref{eq:PrT-bound3}
and \eqref{eq:PrT-bound4} to upper bound $C(W)$.

Consider the fiber parameters in Table~\ref{table:parameters}
and the receiver bandwidth $W=B=500$ GHz. We study both
the \emph{normalized} capacity \eqref{eq:general-capacity-bound1}
and the spectral efficiency
\begin{align}
\eta = \frac{C(W)}{\max(W,W_{\rm min})}  \text{ bits/s/Hz}
\label{eq:se-end}
\end{align}
where $W_{\rm min}$ is the smallest received signal bandwidth that
satisfies \eqref{eq:W-bound1-a}.

Fig.~\ref{fig:capacity1} shows the resulting bounds as the curves labeled
``Upper bound" and ``$\eta$ bound". We also plot a lower bound
from~\cite[Fig.~36,~curve~(1)]{ekwfg-JLT10}. This bound was computed
for 5 WDM signals, each of bandwidth 100 GHz, but with dispersion and
optical filtering (OF). The upper and lower bounds are thus not directly comparable
at high launch power. However, at low launch power both channels
are basically linear and have the same capacity. We remark that we have
shifted the lower bound by $10\log_{10}(5) \approx 7$ dB to the right,
since the $\rm P_{in}$ in ~\cite[Fig.~36]{ekwfg-JLT10} is the power
per WDM channel.

%%%%%%%%%%%%%%%%%%%%%%%%%%%%%%%%%%%%%%%%
\begin{figure}[t!]
  \centerline{\includegraphics[scale=0.48]{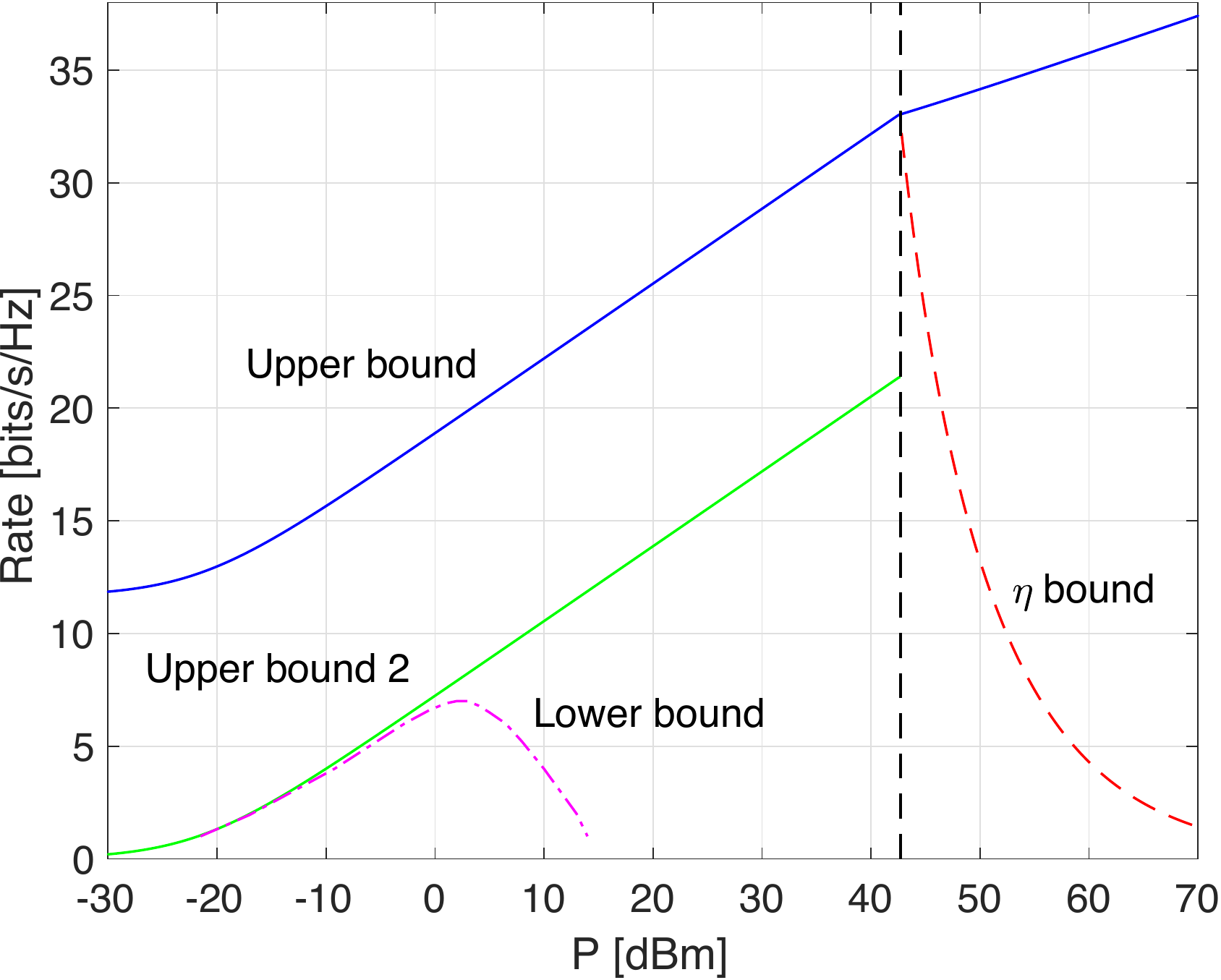}}
  \caption{Normalized capacity bounds for dispersion-free fiber with $B=500$ GHz
  and $W=500$ GHz.
  The curves ``Upper bound" and ``$\eta$ bound" are computed with the RHS of
  \eqref{eq:general-capacity-bound1} and \eqref{eq:se-end}, respectively.
  The curve ``Upper bound 2" is computed with the RHS of
  \eqref{eq:general-capacity-bound2}.
  The lower bound is from~\cite[Fig.~36, curve (1)]{ekwfg-JLT10}.}
  \label{fig:capacity1}
\end{figure}
%%%%%%%%%%%%%%%%%%%%%%%%%%%%%%%%%%%%%%%%

We comment on the behavior of the curves.
\begin{itemize}
\item The upper bound increases with $P$.
\item The model is no longer practically relevant according
to~\eqref{eq:W-bound1-a} for $P > 18.6$ Watts, or $42.7$ dBm.
This bound is shown as the vertical dashed line in 
Fig.~\ref{fig:capacity1}. The real threshold for practical
relevance is much lower.
\item The upper bound has two parts. The part on the left (small to large $P$)
up to the vertical dashed line is based on the known bound \eqref{eq:PrT-bound1}.
The part on the right (very large $P$) is new and is based on \eqref{eq:PrT-bound3}.
\item The bound \eqref{eq:PrT-bound3} seems useful only when the model is
no longer practically relevant. However, this behavior is an artifact of choosing
$W=B$; if $W<B$ then \eqref{eq:PrT-bound3} can be better than
\eqref{eq:PrT-bound1} to the left of the power threshold. Furthermore, it is the bound
\eqref{eq:PrT-bound3} that provides the threshold in the first place.
\item The upper bound is far above the lower bound from~\cite[Fig.~36]{ekwfg-JLT10}.
This suggests that the upper bound is very loose.
\item The upper bound seems extremely loose for small $P$.
To understand why, observe that for small $P$ the
RHS of~\eqref{eq:general-capacity-bound1} is
\begin{align}
   \log_2\left(1 + \frac{Kz}{W N_0}\right) \approx 11.7 \text{ bits/s/Hz}.
   \label{eq:small-P-behavior}
\end{align}
In fact, we expect that $Kz$ should appear in the denominator of the
SNR in \eqref{eq:general-capacity-bound1} and \eqref{eq:small-P-behavior},
and not in the numerator. This issue is discussed in
Sec.~\ref{subsec:capacity-with-OA-noise} below.
\item Beyond the threshold, $\bar{P}_r(W)$ scales
as $\sqrt{P}$. The slope of the bound thus changes from
approximately 3 dB per bit to 6 dB per bit. However, we
expect that the signal phase cannot be used to transmit information
at large $P$, cf.~\cite[Sec.~VI.A]{Yousefi-Kschischang-IT11}.
If this is true, then $C(W)/W$ eventually scales at most as
$\frac{1}{4}\log P$, and the slope of the upper bound becomes 12 dB per bit.
\item The upper bound on \eqref{eq:se-end} 
decreases rapidly beyond the power threshold
because of spectral broadening.
\end{itemize}

%=============================================================
%=============================================================
\subsection{Rates with OA Noise}
\label{subsec:capacity-with-OA-noise}
One might expect that a $Kz$ term should appear in the
denominator of the SNR in \eqref{eq:general-capacity-bound1}.
However, we have so far been unable
to prove this for the model~\eqref{eq:dispersion-free-model}.
The difficulty is related to the signal-noise mixing, the bandlimited
nature of the OA noise, and to the discussion in Sec.~\ref{subsec:large-capacity}.

However, suppose the propagating signal remains inside the
OA band, as required by the inequality $W \le B$. Suppose further
that the propagating signal is accurately characterized
by considering only frequencies within the band $f\in [-W/2,W/2]$ for all $z$.
We can then apply the theory in~\cite{Kramer-etal-15,Yousefi-etal-15} to
improve~\eqref{eq:general-capacity-bound1} to
\begin{align}
  C(W)/W \le \log_2\left(\frac{\bar{P}_r(W)+ W N_0}{Kz (W/B) + W N_0}  \right) \text{ bits/s/Hz}.
  \label{eq:general-capacity-bound2}
\end{align}

Consider again the fiber parameters in Table~\ref{table:parameters}
and $W=B=500$ GHz. Fig.~\ref{fig:capacity1} shows the resulting
bound on $C(W)/W$ as the curve labeled ``Upper bound 2".
We comment on the behavior of the curve.

\begin{itemize}
\item The upper bound now seems reasonable for small $P$.
\item We do not plot the upper bound or spectral efficiency beyond the power
threshold of 42.7 dBm because the signal no longer remains inside the
OA band,  and hence the theory of~\cite{Kramer-etal-15,Yousefi-etal-15}
does not apply. In fact, substantial spectral broadening occurs at much 
smaller launch power, so this theory is more limited than suggested by
Fig.~\ref{fig:capacity1}.
\end{itemize}

%=============================================================
%=============================================================
\subsection{OA Bandwidth Scales with the Launch Power}
\label{subsec:capacity-scaling}
Although the models are impractical for large launch power, we can nevertheless
follow~\cite{Turitsyn-Derevyanko-Yurkevich-Turitsyn-PRL03,Yousefi-Kschischang-IT11}
and study the capacities of the \emph{mathematical} models for large $P$.
For example, suppose $B$ scales as $\sqrt{P}$, which is a lower bound on the
spectral broadening scaling, see \eqref{eq:PrT-bound4-2}. The motivation for 
studying this case is to better understand the limitations of spectral broadening.
We may use the bounds \eqref{eq:PrT-bound3} and \eqref{eq:PrT-bound4}
to upper bound the receiver power, and we can apply the capacity bounds 
\eqref{eq:general-capacity-bound1} and \eqref{eq:general-capacity-bound2}.

Consider again the fiber parameters in Table~\ref{table:parameters}
and the receiver bandwidth $W=500$ GHz. However, based on~\eqref{eq:W-bound2}
and large $P$, we now scale the OA bandwidth as
$B=W \max\left(1, \sqrt{\hat{\kappa} P/512} \right)$ where $\hat{\kappa}=28.6$. 
Fig.~\ref{fig:capacity-scaling} shows the normalized capacity bounds, 
which are similar to Fig.~\ref{fig:capacity1}. The main change is that,
at high power, both $\bar{P}_r(W)$ and $C(W)/W$ scale at most as $P^{-1/4}$,
as predicted by \eqref{eq:PrT-bound4-1} with $\beta=1/2$.\footnote{
In other words, as $P$ increases, $C(W)/W$ first grows as $\log(1+P/(W N_0))$, but
is then upper bounded by $k_1-\frac{1}{4}\log P$ for some constant $k_1$,
and finally is upper bounded by $k_2 P^{-1/4}$ for some constant $k_2$.}

%%%%%%%%%%%%%%%%%%%%%%%%%%%%%%%%%%%%%%%%
\begin{figure}[t!]
  \centerline{\includegraphics[scale=0.48]{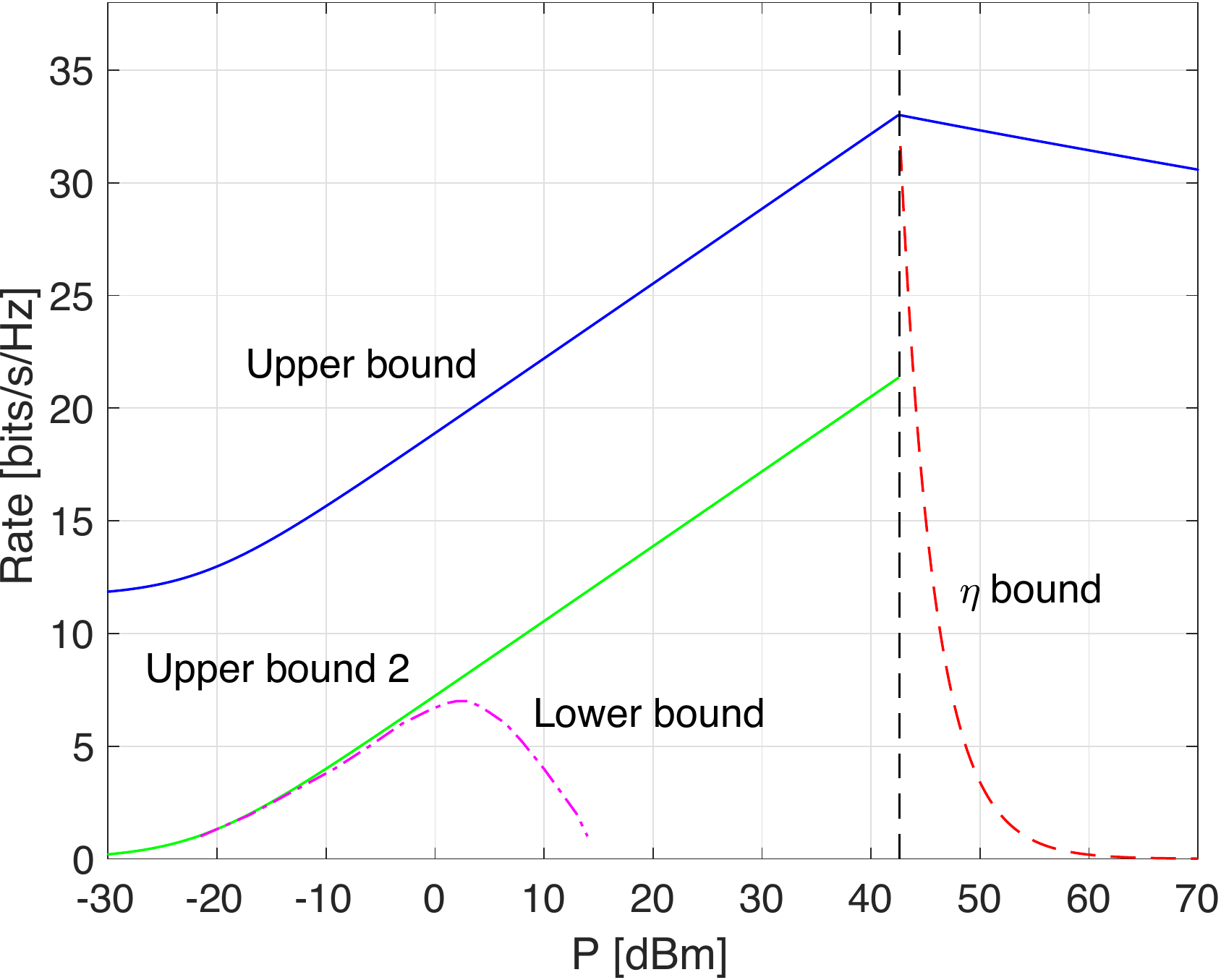}}
  \caption{Normalized capacity bounds for dispersion-free fiber with $W=500$ GHz
  and where $B$ scales as $W\cdot\sqrt{\hat{\kappa} P/512}$ for large $P$.
  The vertical dashed line is the same one as in Fig.~\ref{fig:capacity1}.}
  \label{fig:capacity-scaling}
\end{figure}
%%%%%%%%%%%%%%%%%%%%%%%%%%%%%%%%%%%%%%%%

The reader might expect that the rates in Fig.~\ref{fig:capacity-scaling}
should not decrease with $P$. However, note that the capacities
are normalized, and that the figure is for a system
where the OA bandwidth $B$ changes with $P$. In fact, we expect
that the real (normalized) capacities at large $P$ will be much smaller
than the upper bounds shown in Fig.~\ref{fig:capacity1}
or Fig.~\ref{fig:capacity-scaling}.

%=============================================================
%=============================================================
\section{Conclusions}
\label{sec:conclusions}
We studied a dispersion-free fiber model with distributed OA where
the accumulated spatial noise processes at different time instances
are jointly Wiener.
Our main result is a closed-form expression for the autocorrelation
function of the output signal given the input signal. The expression gives a bound on
the output power of bandlimited and time-resolution limited receivers.
The theory shows that there is a launch power beyond which
the OA bandwidth $B$ can no longer exceed the propagating signal
bandwidth $W$, and the model loses practical relevance.
The growth of $W$ is due to signal-noise mixing that cannot
be controlled by waveform design, other than by reducing power.

The receiver power bounds can be converted to capacity bounds.
However, the latter bounds are far above the true capacity, and
an interesting problem is to improve them.
For example, one can improve the following steps.
\begin{itemize}
\item Treat the noise term in \eqref{eq:autocorr-general} separately.
An upper bound on the received noise power is $Kz \cdot \min(W/B,1)$.
\item Replace \eqref{eq:Bfreq} with a PSD more like the PSD of unit height
in the frequency interval $[-W/2,W/2]$.
% Alternative: use Cauchy-Schwarz inequality ... but this seems to lead to poorer scaling
%
\item For small $\sqrt{\gamma K z^2}$, replace \eqref{eq:Tbound3} with
a bound similar to $(2/3)xz$.
\item For small $\sqrt{\gamma K z^2}$, use \eqref{eq:autocorr-bound}
rather than \eqref{eq:autocorr-bound3}, since \eqref{eq:autocorr-bound}
does not have the factor 1/2 inside the exponential. We
chose~\eqref{eq:autocorr-bound3} in order to treat large $B$.
\item Replace \eqref{eq:sinc-bounds2} with tighter bounds.
\item Use the SPM exponential in \eqref{eq:autocorr-general}.
\end{itemize}
Furthermore, one can improve the bounds for special choices of 
launch signals, e.g., bandlimited signals or PAM with rectangular
pulses, cf.~Sec.~\ref{subsec:ring-modulation}
and Appendix~\ref{app:PAM-Rect}.

Although the bounds are loose, we suspect that they give reasonable
guidance on the capacity behavior of NLSE-based fiber models.
A challenging open problem is to incorporate the spectral broadening
of the noise, see Sec.~\ref{subsec:dispersion-free-model} and
Appendix~\ref{app:raman}. Another challenging problem is to develop
autocorrelation functions for NLSE models with noise, nonlinearity, \emph{and} dispersion.
Finally, one may wish to develop autocorrelation functions and capacity bounds
for OA based on EDFAs.
%There is some hope for progress. For example, if the dispersion is
%due to an all-pass filter, then this filter does not change the PSD.

\clearpage

%%%%%%%%%%%%%%%%%%%%%%%%%%%%%%%%%%%%%%%%%%%%%%%%%%%%%%
% Appendices
%%%%%%%%%%%%%%%%%%%%%%%%%%%%%%%%%%%%%%%%%%%%%%%%%%%%%%
\setcounter{section}{0}
\renewcommand{\thesection}{\Alph{section}}
\renewcommand{\thesubsection}{\arabic{subsection}}
\renewcommand{\appendix}[1]{%
  \refstepcounter{section}%
  \par\begin{center}%
    \begin{sc}%
      Appendix \thesection\par%
      #1%
    \end{sc}%
  \end{center}\nobreak%
}

%%%%%%%%%%%%%%%%%%%%%%%%%%%%%%%%%%%%%%%%%%%%%%%%%%%%%%%%%%%%%%%%%%%%%%%%%
\appendix{Raman Amplification Noise Statistics}
\label{app:raman}
We describe a model for Raman amplification based on the coupled NSLEs in
\cite[p.~305]{Agrawal-03} and \cite{Headly-Agrawal-JQC95}. In particular, the
paper~\cite{Headly-Agrawal-JQC95} defines the following signals and constants:
\begin{itemize}
\item $u_p(z,t)$ and $u_s(z,t)$ are the pump and source signals;
\item $\gamma_p$ and $\gamma_s$ are Kerr coefficients at the pump and source frequencies,
respectively;
\item $\delta_m$ is a fraction related to molecular vibrations;
\item $g_p(t)$ and $g_s(t)$ are filters related to the (third-order) nonlinear susceptibility of the fiber medium;
\item $h_p(t)$ and $h_s(t)$ are filters related to the noise force and the response function that converts this force into a spontaneous polarization \cite[Sec.~II]{Headly-Agrawal-JQC95}.
\end{itemize}
The paper derives a set of coupled equations, see \cite[Eqs.~(12) and~(13)]{Headly-Agrawal-JQC95}.
Setting the dispersion coefficients to zero, we have
\begin{align}
    \frac{\partial u_p}{\partial z} & = j \gamma_p \left(|u_p|^2 + 2(1-\delta_m) |u_s|^2\right) u_p \nonumber \\
   & \quad + \frac{j}{2} u_p u_s \, (g_p(t) * u_s(-t)^*) + \frac{j}{2} u_s \, (h_p * n_w^*) \label{eq:NLSE-raman-p} \\
   \frac{\partial u_s}{\partial z} & = j \gamma_s \left(|u_s|^2 + 2(1-\delta_m) |u_p|^2\right) u_s \nonumber \\
   & \quad + \frac{j}{2} |u_p|^2 (g_s * u_s) + \frac{j}{2} u_p \, (h_s * n_w)
   \label{eq:NLSE-raman-s}
\end{align}
where $n_w(z,t)$ is modeled as additive white Gaussian noise.

We simplify \eqref{eq:NLSE-raman-p}-\eqref{eq:NLSE-raman-s} by neglecting $g_p(t) * u_s(-t)^*$,
$g_s * u_s$, and $h_p*n_w^*$; further discussion on modeling can be found in \cite[p.~305]{Agrawal-03}.
We choose $u_p(0,t)=a_p$ where $a_p$ is a complex constant so that the pump signal is a sinusoid.
Solving equations \eqref{eq:NLSE-raman-p}-\eqref{eq:NLSE-raman-s} yields
\begin{align}
   u_p(z,t) & = a_p \exp\left( j \gamma_p\left[|a_p|^2 z + 2(1-\delta_m) \phi(z,t) \right] \right)
        \label{eq:NLSE-raman-p2} \\
   u_s(z,t) & = \left[ u_s(0,t) + \hat{w}(z,t) \right] \nonumber \\
   & \quad \exp\left( j \gamma_s \left[ \phi(z,t) + 2(1-\delta_m) |a_p|^2 z \right] \right)
   \label{eq:NLSE-raman2}
\end{align}
where 
\begin{align}
   \phi(z,t) & = \int_{0}^{z} |u_s(z',t)|^2 dz' \label{eq:nonlinear-phase} \\
   \hat{w}(z,t) & = \int_{0}^{z} \frac{j}{2} \, a_p \left(h_s*n(z',t)\right) \nonumber \\
   & \qquad \exp\left( j |a_p|^2 \left[ \gamma_p - 2(1-\delta_m)\gamma_s \right] z' \right) \nonumber \\
   & \qquad \exp\left( j \phi(z',t) \left[ - \gamma_s + 2(1-\delta_m)\gamma_p \right]  \right) dz'
   \label{eq:raman-model-noise}
\end{align}
are the accumulated nonlinear phase and noise.

Note that \eqref{eq:dispersion-free-model-noise} multiplies $\phi(z',t)$
by $-\gamma_s$ while \eqref{eq:raman-model-noise} multiplies $\phi(z',t)$ by
$- \gamma_s + 2(1-\delta_m)\gamma_p$.
For example,  if $\gamma_p=\gamma_s$ and $\delta_m=1/2$ then the accumulated
noise at any two times $t$ and $t'$ is jointly Gaussian. However, more realistic
numbers are $\gamma_s=0.95\gamma_p$ and $\delta_m=0.2$
(see the text around~\cite[Eq.~(23)]{Headly-Agrawal-JQC95}) so that $\phi(z',t)$
is multiplied by approximately $0.7\gamma_s$.

%%%%%%%%%%%%%%%%%%%%%%%%%%%%%%%%%%%%%%%%%%%%%%%%%%%%%%%%%%%%%%%%%%%%%%%%%
\appendix{Cameron-Martin Theory}
\label{app:cameron-martin}
The purpose of this appendix, as well as Appendices~\ref{app:mecozzi}
and~\ref{app:one-sample-statistics} is to review relevant results
from~\cite{Cameron-Martin-45} and~\cite{Mecozzi-94}. The space $C[0,1]$ is the set of
real-valued functions $x(t)$ that are continuous on $t\in[0,1]$ and have $x(0)=0$.
%(Note the inclusion of $t=1$ which is important for our Dirac-$\delta$ function $\delta(t-1)$ below.)
Consider the ordered points $t_1,t_2,\ldots,t_n$ with $0<t_1 < t_2 < \ldots < t_n \le1$
and the values $a_i$, $b_i$, $i=1,2,\ldots,n$, for which $a_i < x(t_i) < b_i$.
The Wiener measure is defined as (see~\cite[p.~73]{Cameron-Martin-45})
\begin{align}
  & \frac{1}{\left( \pi^n t_1 (t_2-t_1) \cdots (t_n-t_{n-1}) \right)^{1/2}}
     \int_{a_1}^{b_1} \cdots \int_{a_n}^{b_n} \nonumber \\
  & \exp\left( - \frac{s_1^2}{t_1} 
      - \frac{(s_2-s_{1})^2}{t_2-t_1} - \cdots - \frac{(s_n-s_{n-1})^2}{t_n-t_{n-1}} \right)
      ds_1 \cdots ds_n .
\label{eq:CM00}
\end{align}
The Wiener integral over the space $C[0,1]$ of a functional $F[x]$  is defined
using this measure, and the integral is written as
\begin{align}
   \int_{C[0,1]}^{\mathcal{W}} F[x] \, d_{\mathcal{W}} x.
   \label{eq:CM01}
\end{align}
Note that \eqref{eq:CM01} is the same as $\E{F[X]}$ where $\sqrt{2} \, X(\cdot)$
is a Wiener process on $t\in[0,1]$.

Let $p(t)$ be real-valued, continuous, and positive on $0\le t\le1$ and consider
a complex number $\lambda$. Let $\lambda_0$ be the least characteristic
value of the differential equation
\begin{align}
   f''(t) + \lambda \, p(t) f(t) =0 \label{eq:de}
\end{align}
subject to the boundary conditions $f(0)=f'(1)=0$. Let $f_\lambda(t)$ be any
non-trivial solution of \eqref{eq:de} satisfying $f_{\lambda}'(1)=0$.
Let $g(\cdot)$ be a complex-valued and $L^2$ function on $0\le t\le 1$.
We have the following lemma by using Theorem 2 in \cite{Cameron-Martin-45}
for real $\lambda$ and real-valued $g(\cdot)$ (see especially (3.2) and (3.3) of \cite{Cameron-Martin-45}).
An extension to complex $\lambda$ and $g(\cdot)$ follows
by using the same arguments as in~\cite[pp.~218-219]{Cameron-Martin-45a}.

\begin{lemma}[See \cite{Cameron-Martin-45} and \cite{Cameron-Martin-45a}] \label{lemma:cm}
If $\Re(\lambda)< \lambda_0$ then we have
\begin{align}
  & \int_{C[0,1]}^{\mathcal{W}}
  \exp\left( \lambda \int_0^1 \left[p(t) x^2(t) + 2 g(t) x(t) \right] dt \right) d_{\mathcal{W}}x \nonumber \\
  & = \left(\frac{f_{\lambda}(1)}{f_{\lambda}(0)}\right)^{1/2} \exp\left( \lambda^2 \beta^2 \right)
\label{eq:CM32}
\end{align}
where
\begin{align}
  \beta^2 = \int_0^1 \left[\frac{1}{f_{\lambda}(t)} \int_1^t g(s) f_{\lambda}(s) ds \right]^2 dt.
\label{eq:CM33}
\end{align}
\end{lemma}

\begin{example}[Example 1 in \cite{Cameron-Martin-45}]
If $p(t)=1$ for all $t$, then we have $\lambda_0=\pi^2/4$.
For $\Re(\lambda)<\pi^2/4$, we thus have
\begin{align}
   f_{\lambda}(t)=\cos\left( \lambda^{1/2} (t-1) \right). \label{eq:CM43}
\end{align}   
\end{example}

%\bigskip
%%%%%%%%%%%%%%%%%%%%%%%%%%%%%%%%%%%%%%%%%%%%%%%%%%%%%%%%%%%%%%%%%%%%%%%%%
\appendix{Mecozzi's Identity}
\label{app:mecozzi}
We prove a (slightly corrected) result from \cite[eq. (18)]{Mecozzi-94}.

\begin{lemma}[See {\cite[eq. (18)]{Mecozzi-94}}] \label{lemma:Mecozzi}
Consider a standard real Wiener process $W(\cdot)$. For complex numbers $a$ and $b$,
and imaginary-valued $c$, we have
\begin{align}
  & \E{\exp\left( a W(z) + \int_0^z b W(z') dz' - c \int_0^z W(z')^2 dz' \right)} \nonumber \\
  & = \sqrt{S(c)} \exp\left(\lambda^2 \beta^2\right)
\label{eq:Mecozzi18}
\end{align}
where
\begin{align}
  \lambda^2 \beta^2 
  & =  \left( \frac{a^2}{2} - \frac{b^2}{4c} \right) T(c)
  + \frac{ab}{2c} \left(1 - S(c) \right) + \frac{b^2}{4c} z.
  \label{eq:Mecozzi18-beta}
\end{align}
\end{lemma}
%=============================================================
\begin{proof}
Consider the change of variables $y=z'/z$, and write \eqref{eq:Mecozzi18} as
\begin{align}
  & \E{\exp\left( a W(z) + b z \int_0^1 W(yz) dy - c z \int_0^1 W(yz)^2 dy \right)} \nonumber \\
  & = {\rm E} \left[ \exp\left( a \sqrt{2} z^{1/2} \frac{W(1)}{\sqrt{2}}
         + b \sqrt{2} z^{3/2} \int_0^1 \frac{W(y)}{\sqrt{2}} dy \right. \right. \nonumber \\
  & \qquad \qquad \left. \left. - 2 c z^2 \int_0^1 \frac{W(y)^2}{2} dy \right) \right].
 \label{eq:EW1}
\end{align}
Now consider the function
\begin{align}
  g(s) = \frac{1}{2\lambda}\left( a \sqrt{2} z^{1/2} \frac{1(1-\epsilon\le s \le 1)}{\epsilon} + b \sqrt{2} z^{3/2} \right)
  \label{eq:g-function}
\end{align}
%and we note that the integral in \eqref{eq:CM33} includes the point $t=1$.
where $\epsilon$ is a small positive number. The idea of including the function $1(\cdot)$ is
to avoid a Dirac-$\delta$ function, and so that we can write
\begin{align}
  W(1) = \lim_{\epsilon\rightarrow0} \int_{0}^{1} \frac{1(1-\epsilon\le s \le 1)}{\epsilon} \, W(s) \, ds.
\end{align}
We apply \eqref{eq:CM32}-\eqref{eq:CM43} with $\lambda=-2cz^2$ and compute
\begin{align}
  & \lambda^2 \beta^2 %& = \int_0^1 \left[ \frac{\lambda}{f_{\lambda}(t)} \int_1^t g(s) f_{\lambda}(s) ds \right]^2 dt \nonumber \\
     = \int_0^1 \left[ - \frac{az^{1/2}}{\sqrt{2}} \, \frac{1}{\cos\left( \lambda^{1/2} (t-1) \right)}
         \, \frac{\sin\left( \lambda^{1/2} \epsilon \right)}{\lambda^{1/2} \epsilon} \right. \nonumber \\
  & \qquad \qquad \left. + \frac{bz^{3/2}}{\sqrt{2} \lambda^{1/2}}
      \tan\left( \lambda^{1/2} (t-1) \right) \right]^2 dt.
\end{align}
For vanishing $\epsilon$, the sine ratio becomes 1, and we have
\begin{align}
  \lambda^2 \beta^2
  & =  \frac{a^2z}{2 \lambda^{1/2}} \tan\left( \lambda^{1/2} \right)
          - \frac{abz^2}{\lambda} \left[ 1 - \frac{1}{\cos\left( \lambda^{1/2} \right)} \right] \nonumber \\
  & \quad   + \frac{b^2z^{3}}{2 \lambda^{3/2}} \left[ \tan\left( \lambda^{1/2} \right) - \lambda^{1/2} \right] .
  \label{eq:lambda-beta-2}
\end{align}
We obtain \eqref{eq:Mecozzi18} and \eqref{eq:Mecozzi18-beta} by using
$$\tan(jx) = j\tanh(x), \quad \cos(jx)=\cosh(x).$$
\end{proof}

\begin{example}
If $c\rightarrow0$ then we have
\begin{align}
  \lambda^2 \beta^2 \rightarrow \frac{z}{2} \left( a^2 + ab \, z + \frac{b^2}{3} z^2 \right)
  \label{eq:czero}
\end{align}
which follows by using the Taylor series expansions
$$\tanh x \approx x - x^3/3, \quad \sech x \approx 1 - x^2/2.$$
Alternatively, one can prove \eqref{eq:czero} without using \eqref{eq:CM32} 
by observing that, for $c=0$, the term inside the exponential in \eqref{eq:Mecozzi18}
is a zero-mean Gaussian random variable.
\end{example}

%\clearpage
%=============================================================
%=============================================================
\appendix{One-Sample Statistics}
\label{app:one-sample-statistics}
We review the sample moments computed in \cite[eq. (17)]{Mecozzi-94}.
As we consider only one time instant $t$, we drop the time variables for convenience of notation.

Recall that $u_0=u_{0R} + j u_{0I}$. Consider the conditional moments
\begin{align}
   \mu_{m,n} = \E{\left. U^m (U^*)^n \right| U_0=u_0} 
\end{align}
and the moment generating function
\begin{align}
  & M_{m,n}(s_1,s_2)
     = {\rm E}\left[\exp\left( s_1 \left[ u_0 + \sqrt{K} W(z) \right] \right. \right. \nonumber \\
  & \quad + s_2 \left[ u_0 + \sqrt{K} W(z) \right]^* \nonumber \\
  & \quad \left. \left. + j \gamma (m-n) \int_0^z |u_0 + \sqrt{K} W(z') |^2 dz' \right) \right]
\end{align}  
so that
\begin{align}
  \mu_{m,n} = \left. \frac{\partial^{m+n}}{\partial^m s_1 \, \partial^n s_2} M_{m,n}(s_1,s_2) \right|_{s_1=s_2=0} .
  \label{eq:mu-eqn}
\end{align}
We compute
\begin{align}
  M_{m,n}(s_1,s_2) & = e^d \cdot \E{e^Z} 
\end{align}
where
\begin{align}
& d = s_1 u_0 + s_2 u_0^* + j \gamma (m-n) |u_0|^2 z \\
& Z = \sqrt{K} \left[ s_1 W(z) + s_2 W(z)^* \right] \nonumber \\ 
& \qquad + j \gamma (m-n) \sqrt{K} \int_0^z 2 \, \Re\left\{ u_0 \, W(z')^* \right\} dz' \nonumber \\
& \qquad + j \gamma (m-n) K \int_0^z |W(z')|^2 dz' .
\end{align}

Recall that $W(z)=(W_R(z) + j W_I(z))/\sqrt{2}$, where $W_R(\cdot)$ and $W_I(\cdot)$ 
are independent, standard, real, Wiener processes of unit variance. We may thus simplify
$\E{e^Z}=\E{e^{A+B}}=\E{e^A} \E{e^B}$ where
\begin{align}
A & = \sqrt{\frac{K}{2}} (s_1 + s_2) W_R(z) \nonumber \\
& \qquad + j \gamma(m-n) \sqrt{2K} u_{0R} \int_0^z  W_R(z') dz' \nonumber \\ 
& \qquad + j \gamma(m-n) \frac{K}{2} \int_0^z  W_R(z')^2 dz' \\
B & = j \sqrt{\frac{K}{2}} (s_1 - s_2) W_I(z) \nonumber \\
& \qquad + j \gamma (m-n) \sqrt{2K} u_{0I} \int_0^z  W_I(z') dz' \nonumber \\ 
& \qquad + j \gamma (m-n) \frac{K}{2} \int_0^z  W_I(z')^2 dz' .
\end{align}
Now define the values
\begin{align}
\begin{array}{ll}
   a_1 = \sqrt{K/2}(s_1+s_2), & a_2 = j \sqrt{K/2}(s_1-s_2) \\
   b_1 = j \gamma (m-n) \sqrt{2K} u_{0R}, &  b_2 = j \gamma (m-n) \sqrt{2K} u_{0I} \\
   c_1 = c_2 = c = - j \gamma (m-n) K/2  &
\end{array} \label{eq:abc}
\end{align}
so that
\begin{align*}
 & \frac{a_1^2+a_2^2}{2} = K s_1 s_2, \qquad \frac{b_1^2+b_2^2}{4c} = j \gamma (n-m) |u_0|^2 \\
 & \frac{a_1 b_1 + a_2 b_2}{2c} = - \left(s_1 u_0 + s_2 u_0^*\right).
\end{align*}
We have the following expression using \eqref{eq:Mecozzi18}:
\begin{align}
  & \E{e^A}\E{e^B} =  S(c) \exp\Big[ \nonumber \\
  & \qquad \left( K s_1 s_2 + j \gamma (m-n) |u_0|^2 \right) T(c) \nonumber \\
  & \qquad \left. - \left(s_1 u_0 + s_2 u_0^*\right) \left(1 - S(c) \right) 
      \right. \nonumber \\
  & \qquad - j \gamma (m-n) |u_0|^2 z \; \Big].
\end{align}
This gives a result corresponding to \cite[eq. (19)]{Mecozzi-94}:
\begin{align}
  M_{m,n}(s_1,s_2)
  & = S(c) \exp\Big[ \left( s_1 u_0 + s_2 u_0^* \right) S(c) \nonumber \\
  & \quad + \left( K s_1 s_2 + j \gamma (m-n) |u_0|^2 \right) T(c) \Big]
\end{align}
where $c$ is given in \eqref{eq:abc}. For example, for $m=n$ we
have $c=0$ from \eqref{eq:abc}, and therefore $S(c)=1$ and $T(c)=z$.
If we further have $m=n=1$, then \eqref{eq:mu-eqn} gives
\begin{align}
  \mu_{1,1}
  & = \left. \frac{\partial^{2}}{\partial s_1 \, \partial s_2} M_{1,1}(s_1,s_2) \right|_{s_1=s_2=0} \nonumber \\
  & = Kz + |u_0|^2
\end{align}
as expected from \eqref{eq:t-is-tprime}.

%=============================================================
\subsection{First Moment}
\label{subsec:moment1}
The $m$th moment is
\begin{align}
   \E{U^m} = \mu_{m,0} = u_0^m \, E_m(c) S(c)^{m+1}.
   \label{eq:EUm}
\end{align}
where $c = - j \gamma m K/2$ and
\begin{align}
  E_m(c) = \exp\left( j \gamma m |u_0|^2 T(c) \right).
  \label{eq:Emc}
\end{align}
For example, the first moment has $c = - j \gamma K/2$ and
\begin{align}
   \E{U} = u_0 \, E_1(c) S(c)^2. %= u_0 \cdot \frac{\exp\left( - j \gamma |u_0|^2 T(c) \right)}{C_1(z)^2}.
   \label{eq:Emc1}
\end{align}
Now consider small $\sqrt{\gamma K z^2}$ for which (see \eqref{eq:Sapprox}-\eqref{eq:Tapprox})
\begin{align}
  & T(c) \approx z + j \gamma K z^3/3 \label{eq:tanh-approx} \\
  & S(c)^2 \approx 1.
\end{align}
We thus have
\begin{align}
   \E{U} &\approx u_0 \exp\left( |u_0|^2 \left[ j \gamma z - \gamma^2 K z^3/3 \right] \right) \nonumber \\
   & = u_0 \exp\left( |u_0|^2 \left[ j \gamma z - \kappa/2 \right] \right)
   \label{eq:Eu-approx}
\end{align}
where we have used $\kappa= 2 \gamma^2 K z^3/3$ as in \eqref{eq:kappa}.
The first moment thus experiences a power reduction of
\begin{align}
  f(z) & = \left| E_1(c) S(c)^2 \right|^2  \approx \exp \left( - \kappa |u_0|^2  \right).
  \label{eq:Mecozzi-f}
\end{align}
This matches Mecozzi's equations (29) and (30) from~\cite{Mecozzi-94}.

\begin{remark}
The value of the first moment may seem curious from the following perspective.
The first moment is
\begin{align}
  \E{[u_0 + \sqrt{K} W(z)] \exp\left( j \gamma \int_0^z |u_0 + \sqrt{K} W(z') |^2 dz' \right)}
  \label{eq:first-moment}
\end{align}
where $n=1$. A casual guess is that \eqref{eq:first-moment} should simplify to
\begin{align}
  u_0 \, M_{1,0}(0,0) & = u_0 \, \E{\exp\left( j \gamma \int_0^z |u_0 + \sqrt{K} W(z') |^2 dz' \right)}
\end{align}
since the term with $\sqrt{K} W(z)$ seems to evaluate to zero. However, the
first moment would then be
$$u_0 \, M_{1,0}(0,0) = u_0 \, E_1(c) S(c).$$
Note that the $S(c)$ term is not squared. In fact, we have
\begin{align}
  & \E{ \sqrt{K} W(z) \, \exp\left( j \gamma \int_0^z |u_0 + \sqrt{K} W(z') |^2 dz' \right)} \nonumber \\
  & = u_0 \left[ S(c) -1\right] E_1(c) S(c)
\end{align}
which gives the desired result.
\end{remark}

%\clearpage
%=============================================================
%=============================================================
\appendix{Two-Sample Statistics}
\label{app:two-sample-statistics}
We write $u_0=u_0(t)$, $u_0'=u_0(t')$, and similarly for $u(t)=u_z(t)$.
Consider the conditional moments
\begin{align}
   & \mu_{mnk\ell} = {\rm E} \left[ \left. U^m (U^*)^n (U')^k (U'^*)^{\ell} \right| U_0(\cdot)=u_0(\cdot) \right]
\end{align}
and the moment generating function
\begin{align}
  & M_{mnk\ell}({\bf s})
     = {\rm E}\left[\exp\left(  \right. \right. \nonumber \\
  & \quad \;  s_1 [ u_0 + \sqrt{K} W(z,t)] + s_2 [u_0 + \sqrt{K} W(z,t)]^* \nonumber \\
  & + s_3 [u_0' + \sqrt{K} W(z,t')] + s_4 [u_0' + \sqrt{K} W(z,t')]^* \nonumber \\
  & + j \gamma \int_0^z (m-n) \; |u_0 + \sqrt{K} W(z',t) |^2 \nonumber \\
  & \quad \quad \left. \left. + (k-\ell) \; |u_0' + \sqrt{K} W(z',t') |^2 dz' \right) \right]
\end{align}
where ${\bf s}=[s_1,s_2,s_3,s_4]$ so that
\begin{align}
  \mu_{mnk\ell} = \left. \frac{\partial^{m+n+k+\ell}}{\partial^m s_1 \, \partial^n s_2 \, \partial^k s_3 \, \partial^{\ell} s_4}
  M_{mnk\ell}({\bf s}) \right|_{{\bf s}={\bf 0}} .
\end{align}
%

%=============================================================
\subsection{Autocorrelation Function}
\label{subsec:autocorrelation-function}
For the autocorrelation function, we choose $mnk\ell=1001$.
For simplicity of notation, we replace $s_4$ with $s_2$ and write
\begin{align}
  &M_{1001}(s_1,s_2) = e^d \cdot \E{e^{Z}}
\end{align}
where
\begin{align}
& d = s_1 u_0 + s_2 u_0'^* + j \gamma z \left[ |u_0|^2 - |u_0'|^2 \right] \\
& Z  = \sqrt{K} \left[ s_1 W(z,t)  + s_2 W(z,t')^* \right] \nonumber \\ %\label{eq:autocorrZ1} \\
& \quad + j \gamma \sqrt{K} \int_0^z 2 \Re\left\{ u_{0} W(z',t)^* - u_{0}' W(z',t')^* \right\} dz' \nonumber \\%\label{eq:autocorrZ2} \\ 
& \quad + j \gamma K \int_0^z  |W(z',t)|^2 - |W(z',t')|^2 dz' . \label{eq:autocorrZ}
\end{align}

%=============================================================
\subsection{Noise}
\label{subsec:noise}
Observe that $W(\cdot,t)$ and $W(\cdot,t')$ are {\em correlated} complex Wiener processes.
Since $W(\cdot,t)$ and $W(\cdot,t')$  are circularly symmetric, we may write
\begin{align}
  W(z,t') = \rho^* \, W(z,t) + \sqrt{1-|\rho|^2} \, \tilde{W}(z,t')
  \label{eq:correlatedW}
\end{align}
where the correlation coefficient is
\begin{align}
  \rho = \E{W(z,t) W(z,t')^* }
\end{align}
and $\tilde{W}(z,t')=(\tilde{W}_R(z,t') + j \tilde{W}_I(z,t'))/\sqrt{2}$
where $\tilde{W}_R(\cdot,t')$ and $\tilde{W}_I(\cdot,t')$  are independent,
standard, real, Wiener processes that are jointly independent of $W(\cdot,t)$.
Using $\rho=\rho_R + j \rho_I$, we have $W(\cdot,t')=(W_R(\cdot,t')+j W_I(\cdot,t'))/\sqrt{2}$ where
\begin{align}
  & W_R(z,t') = \rho_R \, W_R(z,t) + \rho_I W_I(z,t) + \sqrt{1-|\rho|^2} \, \tilde{W}_R(z,t') \nonumber \\
  & W_I(z,t') = \rho_R \, W_I(z,t) - \rho_I W_R(z,t) + \sqrt{1-|\rho|^2} \, \tilde{W}_I(z,t').
\end{align}
We thus have $\E{W_R(z,t') W_I(z,t')} = 0$ as required.

We are particularly interested in real $\rho$, see \eqref{eq:Wcorrelation}.
In this case, we have
\begin{align}
  \begin{array}{l}
  W_R(z,t') = \rho \, W_R(z,t) + \sqrt{1-\rho^2} \, \tilde{W}_R(z,t') \\
  W_I(z,t') = \rho \, W_I(z,t) + \sqrt{1-\rho^2} \, \tilde{W}_I(z,t').
  \end{array}
  \label{eq:correlatedW-real-rho}
\end{align}
Thus, the real and imaginary processes are independent, i.e.,
$W_R(z,\cdot)$ is independent of $W_I(z,\cdot)$.

%=============================================================
\subsection{Analysis of $Z$}
\label{subsec:analysis-Z}
Inserting \eqref{eq:correlatedW-real-rho} into \eqref{eq:autocorrZ}, we have
\begin{align}
& Z = \sqrt{K} \left[ s_1 W(z,t) + s_2 \rho W(z,t)^* + s_2 \sqrt{1-\rho^2} \, \tilde{W}(z,t')^* \right] \nonumber \\
& \quad + j \gamma \sqrt{K} \int_0^z 2 \Re\left\{ [u_{0}(t) - u_0(t') \rho] W(z',t)^* \right. \nonumber \\
& \qquad \quad \left. - u_{0}(t') \sqrt{1-\rho^2} \, \tilde{W}(z',t')^* \right\} dz' \nonumber \\ 
& \quad + j \gamma K \int_0^z  (1-\rho^2) \left( |W(z',t)|^2 - |\tilde{W}(z',t')|^2 \right) \nonumber \\
& \qquad \quad - 2 \Re\left\{ \rho \sqrt{1-\rho^2} W(z',t) \tilde{W}(z',t')^* \right\}dz' .
\end{align}
The quadratic form in the last integral is ${\bf W}^\dag {\bf Q} {\bf W}$, where
\begin{align}
  & {\bf W} = \begin{bmatrix} W(z',t) & \tilde{W}(z',t') \end{bmatrix}^T \\
  & {\bf Q} = \begin{bmatrix} 1-\rho^2 & -\rho\sqrt{1-\rho^2} \\
          -\rho\sqrt{1-\rho^2} & -(1-\rho^2) \end{bmatrix}.
\label{eq:quad-form}
\end{align}
The eigenvalue decomposition is $\bf Q={\bf S} {\bf \Lambda} {\bf S}^T$, where
\begin{align*}
& {\bf \Lambda} = \begin{bmatrix} \lambda_1 & 0 \\ 0 & \lambda_2 \end{bmatrix}
= \begin{bmatrix} \sqrt{1-\rho^2} & 0 \\ 0 & -\sqrt{1-\rho^2} \end{bmatrix} \\
& {\bf S} = \begin{bmatrix} {\bf e}_1 & {\bf e}_2 \end{bmatrix}
= \begin{bmatrix} a  & b \\ - b & a \end{bmatrix}
\end{align*}
with
\begin{align}
a &= \frac{1}{\sqrt{2}} \sqrt{1+\sqrt{1-\rho^2}}\\
b &= \frac{1}{\sqrt{2}} \sqrt{1-\sqrt{1-\rho^2}} \cdot {\rm sgn}(\rho) .
\end{align}
Note that ${\bf S}^T {\bf S}={\bf I}$, $a^2+b^2=1$, $a^2-b^2=\sqrt{1-\rho^2}$, and $ab=\rho/2$.
The quadratic form of interest is ${\bf W}^\dag \left( {\bf S} {\bf \Lambda} {\bf S}^T\right) {\bf W}$.
We thus define
\begin{align}
& {\bf V} = {\bf S}^T {\bf W}
\end{align}
where ${\bf V} = [ V_1(z') \: V_2(z')]^T$ with
\begin{align}
  & V_1(z') = (V_{1R}(z') + j V_{1I}(z'))/\sqrt{2} \\
  & V_2(z') = (V_{2R}(z') + j V_{2I}(z'))/\sqrt{2}.
\end{align}
Since the columns of ${\bf S}$ are orthonormal, the random processes
$V_{1R}(\cdot)$, $V_{1I}(\cdot)$, $V_{2R}(\cdot)$, and $V_{2I}(\cdot)$
are jointly independent standard Wiener. We further have
\begin{align}
 &{\bf W} = {\bf S} {\bf V}
    = \begin{bmatrix} a  & b \\ - b & a \end{bmatrix} {\bf V} \\
& {\bf W}^\dag {\bf Q} {\bf W} = {\bf V}^\dag
    \begin{bmatrix} \lambda_1 & 0 \\
    0 & \lambda_2 \end{bmatrix} {\bf V}.
\end{align}

We expand $Z$ as
\begin{align}
& Z = \left(a_{1R} V_{1R} + a_{1I} V_{1I}\right)
   + \left(a_{2R} V_{2R}  + a_{2I} V_{2I}\right) \nonumber \\
& + \int_{0}^{z} \left(b_{1R} V_{1R} + b_{1I} V_{1I}\right)
   + \left(b_{2R} V_{2R} + b_{2I} V_{2I}\right) dz' \nonumber \\
& - \int_{0}^{z} \left(c_{1R} V_{1R}^2 + c_{1I} V_{1I}^2\right)
   + \left(c_{2R} V_{2R}^2 + c_{2I} V_{2I}^2\right) dz'
\label{eq:Zdetail}
\end{align}
where
\begin{align*}
& a_{1R} = \sqrt{\frac{K}{2}}\left( a(s_1+s_2\rho) -bs_2 \sqrt{1-\rho^2} \right)\\
& a_{1I} = j\sqrt{\frac{K}{2}}\left( a(s_1-s_2\rho) +bs_2 \sqrt{1-\rho^2} \right)\\
& a_{2R} = \sqrt{\frac{K}{2}}\left( b(s_1+s_2\rho) +as_2 \sqrt{1-\rho^2} \right)\\
& a_{2I} = j\sqrt{\frac{K}{2}}\left( b(s_1-s_2\rho) -as_2 \sqrt{1-\rho^2} \right)\\
& b_{1R} = j \gamma\sqrt{2K}\left( a(u_{0R}-\rho u_{0R}') +bu_{0R}'\sqrt{1-\rho^2} \right)\\
& b_{1I} = j \gamma\sqrt{2K}\left( a(u_{0I}-\rho u_{0I}') +bu_{0I}'\sqrt{1-\rho^2} \right)\\
& b_{2R} = j \gamma\sqrt{2K}\left( b(u_{0R}-\rho u_{0R}') -au_{0R}'\sqrt{1-\rho^2} \right)\\
& b_{2I} = j \gamma\sqrt{2K}\left( b(u_{0I}-\rho u_{0I}') -au_{0I}'\sqrt{1-\rho^2} \right)\\
& c = c_{1R} = c_{1I} =-c_{2R} = -c_{2I} = - j \gamma \frac{K}{2} \sqrt{1-\rho^2}
\end{align*}
and where $u_{0R}=u_{0R}(t)$, $u_{0I}=u_{0I}(t)$, $u_{0R}'=u_{0R}(t')$, $u_{0I}'=u_{0I}(t')$.

%=============================================================
\subsection{Moment Generating Function}
\label{subsec:MGF}
We use \eqref{eq:S-function}, \eqref{eq:T-function}, the identities
\begin{align*}
& \sech\left(\sqrt{-jx}\right) = \sech\left(\sqrt{jx}\right)^* \\
& \tanh\left(\sqrt{-jx}\right)/\sqrt{-jx} = \left(\tanh\left(\sqrt{jx}\right)\big/\sqrt{jx}\right)^*
\end{align*}
for real $x$, and Lemma \ref{lemma:Mecozzi} to calculate
\begin{align}
& M_{1001}(s_1,s_2) = e^d \, |S(c)|^2 \nonumber \\
&\cdot \exp\left[\left(\frac{a_{1R}^2+a_{1I}^2}{2} - \frac{b_{1R}^2+b_{1I}^2}{4c}\right) T(c) \right.+ \nonumber \\
&\quad \left.\frac{a_{1R}b_{1R}+a_{1I}b_{1I}}{2c}\left(1-S(c) )\right)
+\frac{b_{1R}^2+b_{1I}^2}{4c}z\right] \nonumber \\
&\cdot \exp\left[\left(\frac{a_{2R}^2+a_{2I}^2}{2} + \frac{b_{2R}^2+b_{2I}^2}{4c}\right)T(c)^*\right.- \nonumber \\
&\quad \left.\frac{a_{2R}b_{2R}+a_{2I}b_{2I}}{2c}\left(1-S(c)^*\right)
-\frac{b_{2R}^2+b_{2I}^2}{4c}z\right] .
\label{eq:M0s1s2-1}
\end{align}
We rewrite \eqref{eq:M0s1s2-1} as
\begin{align}
& M_{1001}(s_1,s_2) = e^d \, |S(c)|^2 \exp\Big[ T_R(c) \nonumber \\
&\cdot \left(\frac{a_{1R}^2+a_{1I}^2+a_{2R}^2+a_{2I}^2}{2} - \frac{b_{1R}^2+b_{1I}^2-b_{2R}^2-b_{2I}^2}{4c}\right) \nonumber \\
& + j T_I(c) \nonumber \\
& \cdot \left(\frac{a_{1R}^2+a_{1I}^2-a_{2R}^2-a_{2I}^2}{2} - \frac{b_{1R}^2+b_{1I}^2+b_{2R}^2+b_{2I}^2}{4c}\right) \nonumber \\
&+\left(1-S_R(c)\right) \frac{a_{1R}b_{1R}+a_{1I}b_{1I}-a_{2R}b_{2R}-a_{2I}b_{2I}}{2c} \nonumber \\
&-j S_I(c) \frac{a_{1R}b_{1R}+a_{1I}b_{1I}+a_{2R}b_{2R}+a_{2I}b_{2I}}{2c} \nonumber \\
&\left.+\frac{b_{1R}^2+b_{1I}^2-b_{2R}^2-b_{2I}^2}{4c}z\right].
\label{eq:M0s1s2-1a}
\end{align}
The sums in this expression are
\begin{align*}
& \frac{a_{1R}^2+a_{1I}^2+a_{2R}^2+a_{2I}^2}{2} = Ks_1s_2\rho\\
& \frac{a_{1R}^2+a_{1I}^2-a_{2R}^2-a_{2I}^2}{2} = 0\\
& \frac{b_{1R}^2+b_{1I}^2+b_{2R}^2+b_{2I}^2}{4c} \nonumber \\
& \quad = \frac{- j \gamma}{\sqrt{1-\rho^2}}[|u_0|^2+|u'_0|^2-2\rho\Re\{u_0 u_0'^* \}]\\
& \frac{b_{1R}^2+b_{1I}^2-b_{2R}^2-b_{2I}^2}{4c} = - j \gamma[|u_0|^2-|u'_0|^2]\\
& \frac{a_{1R}b_{1R}+a_{1I}b_{1I}+a_{2R}b_{2R}+a_{2I}b_{2I}}{2c} \nonumber \\
& \quad = \frac{1}{\sqrt{1-\rho^2}}[ s_1(\rho u_0' - u_0) + s_2(u_0'-\rho u_0)^* \big]\\
& \frac{a_{1R}b_{1R}+a_{1I}b_{1I}-a_{2R}b_{2R}-a_{2I}b_{2I}}{2c} \nonumber \\
& \quad = - s_1 u_0 - s_2 u_0'^*  .
\end{align*}
We thus have
\begin{align}
& M_{1001}(s_1,s_2) = |S(c)|^2 \nonumber \\
& \cdot \exp\Big[ T_R(c) \cdot \left( Ks_1s_2\rho + j \gamma[|u_0|^2-|u'_0|^2] \right) \nonumber \\
& - \gamma \frac{T_I(c)}{\sqrt{1-\rho^2}}[|u_0|^2+|u'_0|^2-2\rho\Re\{u_0 u_0'^{*}\}] \nonumber \\
& + S_R(c) \left[ s_1 u_0 + s_2 u_0'^* \right] \nonumber \\
&\left. -j \frac{S_I(c)}{\sqrt{1-\rho^2}} \left[ s_1(\rho u_0' - u_0) + s_2(u_0' - \rho u_0)^* \right] \right].
\end{align}
Taking derivatives and setting $s_1 = s_2 = 0$ we obtain the autocorrelation function \eqref{eq:autocorr-general}.

%\clearpage
%=============================================================
%=============================================================
\appendix{Infinite Bandwidth and Finite Power Noise}
\label{app:infiniteB}
Consider large $B$ but fixed $K$ as in Remark~\ref{rmk:infinite-bandwidth-noise},
i.e., we have a noise PDD of $K$ W/m
that is independent of $B$. Such infinite bandwidth noise is relatively easy to
treat because, conditioned on the input, any two samples $U(t)$ and $U(t')$
with $t'\ne t$ are statistically independent.

%-------
\subsection{Autocorrelation and PSD}
We have $\rho=1(t=t')$ and use \eqref{eq:autocorr-general}  to compute 
\begin{align}
   A(t,t') & = \left\{ 
   \begin{array}{ll}
   Kz + |u_0(t)|^2, & t=t' \\
   v(t) \, v(t')^*, & t \ne t'
   \end{array} \right.
   \label{eq:autocorr-large-B}
\end{align}
where (cf.~\eqref{eq:Emc} and \eqref{eq:Emc1})
\begin{align}
   & v(t) = u_0(t) E_1(c,t) S(c)^2 \\
   & E_1(c,t) = \exp\left( j \gamma |u_0(t)|^2 T(c) \right) \label{eq:E1ct}
\end{align}
and $c = - j \gamma K/2$. The PSD is therefore
\begin{align}
\bar{\mathcal{P}}(f,T)
& = \E{\left| \int_{-\infty}^{\infty} \frac{1}{\sqrt{T}} V(t) e^{-j 2 \pi f t} \, dt \right|^2}
\label{eq:PSD-infiniteB}
\end{align}
where $V(\cdot)$ is the random signal with realization $v(\cdot)$, and where
the expectation is over the random launch signal $U_0(\cdot)$.

%-------
\subsection{A Discrete-Time Model}
We develop a discrete-time model.
Let $\left\{\phi_{m}(\cdot)\right\}_{m=1}^{\infty}$ be a complete orthonormal basis for $L^2[0,T]$.
Consider the projection output
\begin{align}
   Y_m = \int_0^T U(t) \phi_{m}(t)^* dt
   \label{eq:projection-def}
\end{align}
and collect these values in the sequence ${\bf Y}=\{Y_m\}_{m=1}^{\infty}$ with energy
$\|{\bf Y}\|^2=\sum_{m=1}^{\infty} |Y_m|^2$. We can use \eqref{eq:EUm} and the same steps as in
\cite{Barletta-Kramer-CROWNCOM14,Barletta-Kramer-ISIT14} (see also~\cite[Sec.~IV.C-D]{Goebel-etal-IT11})
to show that if $U_0(\cdot)=u_0(\cdot)$, then we have
\begin{align}
   Y_m & \overset{a.s.}{=} \int_0^T \E{\left. U(t) \right| U_0(\cdot)=u_0(\cdot)} \phi_{m}(t)^* dt \nonumber \\
   & = \int_0^T E_1(c,t) S(c)^2 u_0(t) \phi_{m}(t)^* dt
   \label{eq:projection1}
\end{align}
where $E_1(c,t)$ is given by \eqref{eq:E1ct}.
In other words, the channel effectively modulates $u_0(t)$ by the factor $E_1(c,t) S(c)^2$. 
This means there is {\em no phase noise} since $E_1(c,t)$ is a function of $u_0(t)$.
Instead, the signal loses {\em energy} since $| E_1(c,t) S(c)^2 |<1$
if $|u_0(t)|>0$, $\gamma>0$, and $K>0$.

The result \eqref{eq:projection1} suggests that we study the complex-alphabet
and continuous-time model
\begin{align}
   Y(t) = E_1(c,t) S(c)^2 u_0(t) + N_r(t)
   \label{eq:simple-model}
\end{align}
where $N_r(\cdot)$ is an AWGN process with a one-sided PSD of $N_0$ W/Hz.
The model \eqref{eq:simple-model} has several interesting features. First, an optimal
receiver\footnote{An optimal receiver puts out sufficient statistics for estimating which
signal of a set $\{u_{0,s}(t)\}_{s=1}^{S}$ was transmitted. Note that an optimal receiver
for the model \eqref{eq:simple-model} may not be an optimal receiver for the original
model \eqref{eq:dispersion-free-model}.}  may use matched filtering for signals of the
form $u_0(t) E_1(c,t)$.
Moreover, suppose $\sqrt{\gamma K z^2}$ is small so that we have
(see \eqref{eq:E1ct} and \eqref{eq:tanh-approx})
\begin{align}
  & E_1(c,t) S(c)^2 \approx \exp\left( \left[ j \gamma z - \kappa/2 \right] |u_0(t)|^2 \right)
\end{align}
where $\kappa= 2 \gamma^2 K z^3/3$ as in \eqref{eq:kappa}.
We see that the receiver modulates the phase and amplitude of
the received signal as a function of $|u_0(t)|^2$. In particular,
if we use PAM with rectangular pulses that are time-limited to $[0,T_s)$,
then the standard matched filter is optimal but the receiver a-posteriori probability calculation
should account for the channel's symbol-dependent attenuation and phase shift.

%-------
\subsection{Capacity Bounds}
Consider the channel \eqref{eq:simple-model} with an 
amplitude constraint $|u_0(t)|\le A_{\rm max}=\sqrt{P T_s}$.
For rectangular pulses, we have
\begin{align}
  & C = \max_{X: |X|^2 \le P} \frac{1}{T_s} \left[ h(Y) - \log_2(\pi e N_0) \right] \nonumber \\
  & \le \frac{1}{T_s} \log_2\left( 1 + \frac{\max_{X: |X|^2\le P} \E{T_s |X|^2 e^{-\kappa |X|^2}} }{N_0} \right) \nonumber \\
  & = \left\{ \begin{array}{ll}
     \frac{1}{T_s}  \log_2\left( 1 + \frac{T_s P e^{-\kappa P}}{N_0} \right), & \text{if $P < 1/\kappa$} \\
     \frac{1}{T_s}  \log_2\left( 1 + \frac{T_s}{\kappa e N_0} \right), & \text{else}.
     \end{array} \right.
%  & \le \log\left( 1 + \frac{(z_0/z)^3T_s}{N_0} \right).
\label{eq:C-bound}
\end{align}
The smallest $P$ that achieves the maximal upper bound is $P=1/\kappa$.
In fact, the bound on the RHS of \eqref{eq:C-bound} can be approached
if $A_{\rm max}\rightarrow0$, cf.~\cite{Thangaraj-Kramer-Boecherer-IT17}. 
Furthermore, for fixed $P$, we can maximize the RHS of \eqref{eq:C-bound}
over $T_s$ to obtain $T_s \rightarrow 0$ and therefore
\begin{align}
  \lim_{T_s \rightarrow 0} C \le \frac{1}{\kappa eN_0} \log_2(e) \quad \text{bits/s}.
\end{align}
The optimal signaling thus uses very fast pulses, and the capacity $C$ decreases
inversely proportional to $\gamma^2$, $K$, and $z^3$.

%=============================================================
%=============================================================
\appendix{Nonlinearity Can Increase Capacity}
\label{app:nonlinear-capacity}
We show that nonlinearity can increase capacity even with receiver noise.
In the absence of OA noise, the model \eqref{eq:electronic-noise-model} with
$u(z,t)$ defined by \eqref{eq:dispersion-free-model} is
\begin{align}
   u_r(t) = u_0(t) e^{j \gamma |u_0(t)|^2} + n_r(t).
   \label{eq:nonlinear-AWGN-model}
\end{align}
Suppose the transmitter uses PAM with square-root pulses
\begin{align}
   g(t) = \left\{ \begin{array}{ll}
   \sqrt{\frac{t - T_s + 1}{T_s(1-T_s/2)}}  , & t \in [0,T_s) \\
   0 & \text{else}
   \end{array} \right.
   \label{eq:root-pulse}
\end{align}
where $T_s\le 1$, see~\eqref{eq:PAM}.
For $\gamma=0$ the capacity is
\begin{align}
   C = \frac{1}{T_s} \log_2\left( 1 + \frac{P T_s}{N_0} \right) \text{ bits/s}.
   \label{eq:CW-gamma0}
\end{align}
This capacity can be achieved by scaling and shaping quadrature 
amplitude modulation (QAM) symbols $x_k = x_{R,k} + j x_{I,k}$
where the $x_{R,k}$ and $x_{I,k}$ take on values in
$\left\{ \pm 1, \pm 3, \ldots \right\}$. Note that the capacity scales
as $\log_2 P$ for large $P$.

Suppose now that $\gamma > 0$. The main observation is
that the nonlinearity in \eqref{eq:nonlinear-AWGN-model} converts
the pulse \eqref{eq:root-pulse} to a tone whose frequency and power is proportional
to $|x_k|^2$. More precisely, the noise-free output signals have the form
\begin{align}
   u(t) =  x_k \sqrt{\frac{t - T_s + 1}{T_s(1-T_s/2)}}  \exp\left( j 2 \pi h \, |x_k|^2 \, (t - T_s +1) \right)
   \label{eq:output-signals}
\end{align}
for $t \in [ kT_s,(k+1)T_s )$, where
$$h = \frac{\gamma}{2 \pi T_s(1-T_s/2)}$$
is a modulation index~\cite[p.~118]{Proakis-Salehi-5}.
Suppose we use intensity modulation where we choose the $M$ symbols
\begin{align}
x_k \in \{ (2i-1) \Delta: \; i=M+1,M+2,\ldots,2M \}
\end{align}
each with probability $1/M$.
The average energy is then $E=(28M^2-1) \Delta^2 /3$. We further choose $\gamma$
so that  $hT_s$ is a positive integer, e.g., $\gamma=2\pi (1-T_s/2)$ so that $hT_s=1$.
The pulses \eqref{eq:output-signals} are then mutually orthogonal:
for $x_\ell \ne x_m$ we have
\begin{align}
  & \int_0^{T_s} \frac{t-T_s+1}{T_s(1-T_s/2)} e^{j2\pi h (|x_\ell|^2-|x_m|^2) (t-T_s+1) } \, dt
  = 0.
\end{align}
The channel has thus converted the ASK signals to orthogonal FSK signals
for which the frequency grows with the power.

Next, a standard upper bound on the error probability of signal sets is the
union bound~\cite[p.~185]{Proakis-Salehi-5}
\begin{align}
  P_e \le (M-1)\, Q\left( \frac{d_{\rm min}}{\sqrt{2N_0}} \right)
\end{align}
where $d_{\rm min}$ is the minimum Euclidean distance between different pulses.
Since our FSK signals are mutually orthogonal, the minimum distance corresponds to
the signals with $i=M+1$ and $i=M+2$, i.e., we have
$$d_{\rm min}= \sqrt{(2M+1)^2 + (2M+3)^2} \cdot \Delta \ge \sqrt{8} \, M \Delta$$
and therefore
\begin{align}
  P_e \le (M-1)\, Q\left( \frac{2 \, M \Delta}{ \sqrt{N_0}} \right).
\end{align}
We use $R=\log_2 M$ bits/symbol and $Q(x) \le e^{-x^2/2}$ for positive $x$ to write
\begin{align}
  P_e < \exp\left( R \ln 2 - \frac{6}{28} \frac{E}{N_0}\right).
\end{align}
This bound shows that, for any choice of target error probability $\tilde{P}_e$, we can
choose the rate as
\begin{align}
 \frac{R}{T_s} = \frac{6}{28} \frac{P}{N_0} \log_2 e + \frac{1}{T_s} \ln \tilde{P}_e \quad \text{bits/s}.
\end{align}
The capacity thus scales linearly with $P$ rather than logarithmically as for $\gamma=0$.
\begin{remark}
The reason for the capacity gain is because the channel has spread the spectrum
of the PAM signal. The gain is thus at the expense of using more
frequency resources.
\end{remark}
\begin{remark}
The above example shows that intensity modulation can achieve a capacity
that grows {\em linearly} with $P$  for large $P$. The per-sample rate
$\frac{1}{2} \log P$ from~\cite{Turitsyn-Derevyanko-Yurkevich-Turitsyn-PRL03,Yousefi-Kschischang-IT11} 
thus underestimates capacity even with AWGN at the receiver.
\end{remark}
\begin{remark}
The channel is artificial because we have assumed the channel is lossless
without amplification.
\end{remark}

%=============================================================
\appendix{Proofs of Lemmas~\ref{lemma:WlessB-lemma1}-\ref{lemma:WlessB-lemma4}}
\label{app:Pt-lemma-proofs}

We repeat \eqref{eq:PrT-integrand} here for convenience:
\begin{align}
   & P_r(W,T,t) = 2 \left[ K z + \left(\sqrt{P_t} + \delta \right)^2 \right] \nonumber \\
   & \int_{-T/2}^{T/2} 
      |S(c)|^2 \exp\left( - \gamma \, T_I(c) \frac{P_t}{2} \sqrt{1-\rho^2} \right) \,
      \left| b(t-t') \right| \, dt' .
   \label{eq:PrT-integrand-app}
\end{align}

%=============================================================
\subsection*{Proof of Lemma~\ref{lemma:WlessB-lemma1}}
\label{subsec:bandlimited-receiver-app}
Consider $W \le B$ and $\gamma (K/2) z^2 \le 1$. We have
$x\le 1$ (see~\eqref{eq:x-def}) and the bound \eqref{eq:Tbound3} gives
\begin{align}
   \gamma\, T_I(c) \ge \gamma\, \frac{z\,x}{3} = (\kappa/4) \sqrt{1-\rho^2}.
   \label{eq:gTI-bound}
\end{align}
We further use the crude bounds \eqref{eq:sinc-bounds2} in Appendix~\ref{app:simple-bounds}
to upper bound the exponential of \eqref{eq:PrT-integrand-app} as
\begin{align}
   & \exp\left( - \gamma \, T_I(c) \frac{P_t}{2} \sqrt{1-\rho^2} \right) \nonumber \\
   & \le \left\{ \begin{array}{ll} 
   \exp\left( - (\kappa/8) P_t \, B^2 (t-t')^2 \right), & B|t-t'| \le 1 \\
   \exp\left( - (\kappa/9) P_t \right), & B|t-t'| > 1.
   \end{array} \right. \label{eq:region-bounds}
\end{align}
We now define
\begin{align}
  \sigma = \sqrt{(\kappa/8) P_t} \, B , \quad y = \sigma (t-t')
  \label{eq:sigma-y}
\end{align}
and use the time intervals
\begin{align}
   & {\mathcal I}_1 = \{t': |t-t'| \le 1/B, \, |t'| \le T/2 \} \\
   & {\mathcal I}_2 = \{t': |t-t'| > 1/B, \, |t'| \le T/2 \}
\end{align}
to bound \eqref{eq:PrT-integrand-app} as
\begin{align}
   & P_r(W,T,t) \nonumber \\
   & \overset{(a)}{\le} 2 \left[ Kz + (\sqrt{P_t}+\delta)^2 \right]
      \left[ \int_{\mathcal{I}_1} 2 W  e^{- \sigma^2 (t-t')^2} \, dt' \right. \nonumber \\
   & \qquad \left. + \int_{\mathcal{I}_2} |S(c)|^2  e^{- (\kappa/9) P_t} \, |b(t-t')| \, dt' \right] \nonumber \\
   & \overset{(b)}{\le} 4 \left[ Kz + (\sqrt{P_t}+\delta)^2 \right]
      \left[ \int_{0}^{\sigma/B} \frac{2W}{\sigma}  e^{-y^2} \, dy \right. \nonumber \\
   & \qquad \left. + \int_{1/B}^{\infty} 5 e^{- \sqrt{\gamma K z^2}} \, e^{- (\kappa/9) P_t} \, |b(\tau)| \, d\tau \right].
   \label{eq:PrWTt-bound-WlessB}
\end{align}
Step $(a)$ in \eqref{eq:PrWTt-bound-WlessB} follows by using \eqref{eq:region-bounds},
$|b(t-t')| \le 2W$, and $|S(c)|^2\le 1$; step $(b)$ follows by using $\tau=t-t'$,
inserting \eqref{eq:Smagbound}, and applying the second inequality in \eqref{eq:sinc-bounds2}
to bound $x\le\gamma (K/2) z^2 (19/20)$.
Evaluating the integrals and using \eqref{eq:Btime-bound} gives
\begin{align}
   & P_r(W,T,t) \le 4 \left[ Kz + (\sqrt{P_t}+\delta)^2 \right] \nonumber \\
   & \left[ \frac{2 \, W/B}{\sqrt{(\kappa/8) P_t}} 
   \frac{\sqrt{\pi}}{2} {\rm erf}\left( \sqrt{(\kappa/8) P_t} \right)
   + 5 e^{- \sqrt{\gamma K z^2} - (\kappa/9) P_t} \right] .
   \label{eq:WlessB-lemma1-app}
\end{align}

%=============================================================
\subsection*{Proof of Lemma~\ref{lemma:WlessB-lemma2}}
\label{subsec:bandlimited-receiver2-b}
Consider $W \le B$ and $\gamma (K/2) z^2\ge 1$. We now have the
situation that $x \ge1$ can occur, so that
we need both bounds in \eqref{eq:Tbound3} depending on the value of $\tau=t-t'$.
We further need both bounds of \eqref{eq:sinc-bounds2}
in Appendix~\ref{app:simple-bounds}, depending on whether $|\tau|$ is smaller
or larger than $1/B$. This leads to four integration regions in general, as
described below.

We begin with \eqref{eq:sinc-bounds2}-\eqref{eq:sinc-bounds2a} to write
\begin{align}
   B |\tau| \le \sqrt{1 - \sinc^2 (B \tau)} \le 2 B |\tau|
\end{align}
for $0 \le B |\tau| \le 1$. Using \eqref{eq:x-def}, we thus have
\begin{align}
  \gamma (K/2) z^2 B |\tau| \le x \le \gamma K z^2 B |\tau| .
\label{eq:tau-bound}
\end{align}
Defining
\begin{align}
   \tau^*=1/(\gamma K z^2 B)
\end{align}
we have $2\tau^* \le 1/B$ by hypothesis, and \eqref{eq:tau-bound} gives
\begin{align}
\begin{array}{l}
|\tau| \le \tau^* \Rightarrow x \le 1 \\
|\tau| \ge 2 \tau^* \Rightarrow x \ge 1.
\end{array} \label{eq:taustar}
\end{align}

We proceed to upper bound $P_r(W,T,t)$ by splitting the integral
\eqref{eq:PrT-integrand-app} into four parts with $|\tau| \le \tau^*$,
$\tau^* \le |\tau| \le 2\tau^*$, $2\tau^* \le |\tau| \le 1/B$,
and $|\tau| > 1/B$. Using \eqref{eq:Tbound3} and \eqref{eq:sinc-bounds2},
we bound the exponential of \eqref{eq:PrT-integrand-app} as
\begin{align}
   & \exp\left( - \gamma \, T_I(c) \frac{P_t}{2} \sqrt{1-\rho^2} \right) \nonumber \\
   & \le \left\{ \begin{array}{ll}
      \exp\left( - (\kappa/8) P_t \, B^2 \tau^2 \right), & |\tau| \le \tau^* \\
      \exp\left( - \sqrt{\frac{\gamma B}{18 K}} P_t \sqrt{\tau} \right), & 2\tau^* \le |\tau| \le 1/B   \\
      \exp\left( - \sqrt{\frac{\gamma}{20 K}} P_t \right), & |\tau| > 1/B .
      \end{array} \right.
   \label{eq:WlessB-part2}
\end{align}
For the regime $\tau^* \le |\tau| \le 2\tau^*$, we use as upper bound the sum of the
first and second terms on the RHS of \eqref{eq:WlessB-part2}.

Inserting \eqref{eq:WlessB-part2} into \eqref{eq:PrT-integrand-app}, and
following similar steps as in \eqref{eq:sigma-y}-\eqref{eq:PrWTt-bound-WlessB},
we have
\begin{align}
   & P_r(W,T,t) \nonumber \\
   & \overset{(a)}{\le} 4 \left[ Kz + (\sqrt{P_t}+\delta)^2 \right]
      \left[ \int_{0}^{2 \sigma \tau^*} \frac{2W}{\sigma}  e^{-y^2} \, dy \right. \nonumber \\
   &  \left. + \int_{\tau^*}^{1/B} 10W  e^{- a \sqrt{\tau}} \, d\tau
      + \int_{1/B}^{\infty} 5 e^{- \sqrt{\gamma K z^2} - \sqrt{\frac{\gamma}{20 K}} P_t} \, |b(\tau)| \, d\tau \right] \nonumber \\
   & \overset{(b)}{\le} 4 \left[ Kz + (\sqrt{P_t}+\delta)^2 \right]  \nonumber \\
   & \left[ \frac{2 \, W/B}{\sqrt{(\kappa/8) P_t}}
      \frac{\sqrt{\pi}}{2} {\rm erf}\left( \sqrt{(\kappa/8) P_t} \, 2 B \tau^* \right)
      \right. \nonumber \\
   & \left. + \frac{20 W}{a} \left( \sqrt{\tau^*} + \frac{1}{a} \right) e^{- a \sqrt{\tau^*}}
       + 5e^{- \sqrt{\gamma K z^2} - \sqrt{\frac{\gamma}{20 K}} P_t} \right]
   \label{eq:WlessB-part2-1}
\end{align}
where for step $(b)$ we have defined
\begin{align}
   a & = \sqrt{\frac{\gamma B}{18 \, K}} P_t + \sqrt{2 \gamma K z^2 B}.
   \label{eq:WlessB-part2-2a}
\end{align}
Step $(a)$ in \eqref{eq:WlessB-part2-1} used the first inequality in
\eqref{eq:sinc-bounds2} to bound $|S(c)|^2$ for the second integral.
For step $(b)$, we applied
\begin{align}
   \int_{\tau^*}^{1/B} e^{- a \sqrt{\tau}} \, d\tau
   & = \int_{\sqrt{\tau^*}}^{1/\sqrt{B}} e^{- a t} \, 2t \, dt \nonumber \\
   & \le \frac{2}{a} \left( \sqrt{\tau^*} + \frac{1}{a} \right) e^{- a \sqrt{\tau^*}}.
   \label{eq:WlessB-part2-2}
\end{align}
Finally, we use
\begin{align}
   & a \ge \max\left( \sqrt{\frac{\gamma B}{18 \, K}} P_t, \sqrt{2 \gamma K z^2 B} \right) \\
   & \sqrt{(\kappa/8) P_t} \, 2 B \tau^* = \sqrt{\frac{P_t}{3 K z}}
\end{align}
to simplify \eqref{eq:WlessB-part2-1} and obtain
\begin{align}
   & P_r(W,T,t) \le 4 \left[ Kz + (\sqrt{P_t}+\delta)^2 \right]  \nonumber \\
   & \quad \left[ \frac{2\, W/B}{\sqrt{(\kappa/8) P_t}}
      \frac{\sqrt{\pi}}{2} {\rm erf}\left( \sqrt{\frac{P_t}{3 Kz}}\right)
      \right. \nonumber \\
   & \qquad \left. + \frac{25 \, W/B}{\gamma K z^2 } e^{- \frac{1}{\sqrt{18} K z} P_t }
       + 5 e^{- \sqrt{\gamma K z^2} - \sqrt{\frac{\gamma}{20 K}} P_t} \right] .
   \label{eq:WlessB-lemma2-app}
\end{align}

%=============================================================
\subsection*{Proof of Lemma~\ref{lemma:WgreaterB-lemma}}
\label{subsec:bandlimited-receiver2-c}
Consider the case $W\ge B$ and $\gamma (K/2) z^2 \le 1$.
We again use the bound \eqref{eq:gTI-bound} and the
parameters \eqref{eq:sigma-y} to write
\begin{align}
   & P_r(W,T,t) \nonumber \\
   & \le 4 \left[ Kz + (\sqrt{P_t}+\delta)^2 \right]  
       \left[ \int_{0}^{\sigma/W} \frac{2W}{\sigma}  e^{-y^2} \, dy \right. \nonumber \\
   &  \left.  + \int_{1/W}^{1/B} e^{- \sigma^2 \tau^2} \, |b(\tau)| \, d\tau
       + 5 e^{- \sqrt{\gamma K z^2} - (\kappa/9) P_t} \right].
   \label{eq:WgreaterB-part1}
\end{align}
We now use $|b(t)|\le 2 (W (\pi t)^2)^{-1}$ to
upper bound the second integral in \eqref{eq:WgreaterB-part1} as follows:
\begin{align}
   \int_{1/W}^{1/B} e^{- \sigma^2 \tau^2}  \frac{2}{\pi^2 W \tau^2} \, d\tau  \le e^{- (\kappa/8) P_t (B/W)^2} \frac{W-B}{4W} .
   \label{eq:WgreaterB-part2}
\end{align}
Inserting into \eqref{eq:WgreaterB-part1}, we have
\begin{align}
   & P_r(W,T,t) \le 4 \left[ Kz + (\sqrt{P_t}+\delta)^2 \right] \nonumber \\
   & \left[ \frac{2\, W/B}{\sqrt{(\kappa/8) P_t}} 
   \frac{\sqrt{\pi}}{2} {\rm erf}\left( \sqrt{(\kappa/8) P_t} \, \frac{B}{W} \right) \right. \nonumber \\
   & \left. + \frac{1}{4}\left(1-\frac{B}{W}\right) e^{- (\kappa/8) P_t (B/W)^2}
      + 5 e^{- \sqrt{\gamma K z^2} - (\kappa/9) P_t} \right] .
     \label{eq:WgreaterB-lemma-app}
\end{align}

%=============================================================
\subsection*{Proof of Lemma~\ref{lemma:WlessB-lemma3}}
\label{subsec:average-power-app1}
Observe that the RHS of \eqref{eq:WlessB-lemma1-app} includes the form
\begin{align}
   f(P_t) = (a + b \sqrt{P_t} + c P_t) \frac{\sqrt{\pi}}{2} \frac{{\rm erf}(\sqrt{s P_t})}{\sqrt{ sP_t}}
   \label{eq:form-app}
\end{align}
where $a=Kz + \delta^2$, $b=2 \delta$, $c=1$, and $s=\kappa/8$.
The results \eqref{eq:concave-deriv1}-\eqref{eq:concave-deriv2} derived in
Appendix~\ref{app:simple-bounds} state that \eqref{eq:form-app} is non-decreasing and
concave in $P_t$ if $P_t \ge 3(Kz + \delta^2)$.
Thus, if we replace $P_t$ with $P_t + P_o$, where the offset power is $P_o=3(Kz + \delta^2)$,
then \eqref{eq:form-app} is non-decreasing and concave for $P_t \ge 0$.

Next, by using \eqref{eq:xe-bound1}-\eqref{eq:xe-bound2} in Appendix~\ref{app:simple-bounds}, we have
\begin{align}
 20 \left[ Kz + (\sqrt{P_t}+\delta)^2 \right] e^{- \sqrt{\gamma K z^2} - (\kappa/9) P_t} \le c_1
\end{align}
where
\begin{align}
 c_1 & = 20 \left[ Kz + \delta^2 + \sqrt{\frac{18}{\kappa e}} \delta + \frac{9}{\kappa e} \, \right]
              e^{- \sqrt{\gamma K z^2}} .
              \label{eq:c1-app}
\end{align}
Observe that $c_1$ is independent of $P_t$ and $W$.

We now loosen \eqref{eq:WlessB-lemma1-app} to
\begin{align}
   & P_r(W,T,t) \le c_1 + 8 (W/B) f\left( P_t + P_o \right)
   \label{eq:WlessB-lemma1-app-bound}
\end{align}
where the RHS is non-decreasing and concave in $P_t$.
Jensen's inequality applied to the RHS of \eqref{eq:PrT-bound2a} thus gives the bound
\begin{align}
   & \bar{P}_r(W,T) \le c_1 + 8 (W/B) f\left( \bar{P}_T + P_o \right)
   \label{eq:PrT-bound3-app-a}
\end{align}
where we have replaced $P_t$ with $\bar{P}_T$, see \eqref{eq:P-constraint}.
Furthermore, we require $\bar{P}_T \le P$, so we have
\begin{align}
   & \bar{P}_r(W,T) \le c_1 + 8 (W/B) f\left( P + P_o \right).
   \label{eq:PrT-bound3-app}
\end{align}

We may simplify the bound further without changing the scaling behavior
that we are interested in. We use ${\rm erf}(y) \le 1$ and loosen
\eqref{eq:PrT-bound3-app} to
\begin{align}
   & \bar{P}_r(W,T) \le c_1 + \frac{8 W}{B}
   \frac{Kz + \left( \sqrt{P+ P_o}+\delta \right)^2}
   {\sqrt{(\kappa/8) \left( P+P_o \right)}}.
   \label{eq:PrT-bound3-app-1}
\end{align}
For example, the RHS of \eqref{eq:PrT-bound3-app-1} scales as
$\sqrt{P}$ for large $P$.

%=============================================================
\subsection*{Proof of Lemma~\ref{lemma:WlessB-lemma4}}
\label{subsec:average-power-app2}
We repeat the above steps \eqref{eq:form-app}-\eqref{eq:PrT-bound3-app-1}
for \eqref{eq:WlessB-lemma2-app}. We now have $s=1/(3Kz)$ in
\eqref{eq:form-app}, and we use
\eqref{eq:xe-bound1}-\eqref{eq:xe-bound2} to compute
\begin{align}
 \frac{100}{\gamma K z^2 } \left[ Kz + (\sqrt{P_t}+\delta)^2 \right]
    e^{- \frac{1}{\sqrt{18} K z} P_t } & \le c_2 \nonumber \\
 20 \left[ Kz + (\sqrt{P_t}+\delta)^2 \right]
    e^{- \sqrt{\gamma K z^2} - \sqrt{\frac{\gamma}{20 K}} P_t} & \le c_3
\end{align}
where
\begin{align}
 c_2 & = \frac{100}{\gamma K z^2 } \left[ Kz + \delta^2 + \sqrt{\frac{6Kz}{e}} \delta
              +  \frac{\sqrt{18}Kz}{e} \, \right] \nonumber \\
 c_3 & = 20 \left[ Kz + \delta^2 + \left(\frac{80 K}{\gamma e^2}\right)^{1/4} \delta
              +  \sqrt{\frac{20K}{\gamma e^2}} \, \right] e^{- \sqrt{\gamma K z^2}} .
\end{align}
Observe that $c_2$ and $c_3$ are independent of $P_t$ and $W$.
We loosen \eqref{eq:WlessB-lemma2-app}
%\begin{align}
%   \bar{P}_r(W,T,t) & \le (W/B) c_2 + c_3 \nonumber \\
%   & \quad + \frac{8(W/B)}{\gamma K z^2} f\left( P_t + 3(Kz + \delta^2) \right)
%   \label{eq:WlessB-lemma2-app-bound}
%\end{align}
and use Jensen's inequality to write
\begin{align}
   \bar{P}_r(W,T) & \le \frac{W}{B} c_2 + c_3 + \frac{16(W/B)}{\gamma K z^2} f\left( \bar{P}_T + P_o \right)
   \label{eq:PrT-bound4-app-a}
\end{align}
which is the analog of \eqref{eq:PrT-bound3-app-a}. The RHS of
\eqref{eq:PrT-bound4-app-a} is increasing in $\bar{P}_T$, so we have
\begin{align}
   \bar{P}_r(W,T) & \le \frac{W}{B} c_2 + c_3 + \frac{16(W/B)}{\gamma K z^2} f\left( P + P_o \right) .
   \label{eq:PrT-bound4-app}
\end{align}

To study large $P$, we may simplify \eqref{eq:PrT-bound4-app}
by again using ${\rm erf}(y) \le 1$. We arrive at a similar bound as
\eqref{eq:PrT-bound3-app-1}, and $\bar{P}_r(W,T)$ again
scales at most as $\sqrt{P}$ for large $P$.

%=============================================================
\appendix{Energy Bound for PAM with Rectangular Pulses and Ring Modulation}
\label{app:PAM-Rect}
PAM with rectangular pulses and ring modulation has a constant
envelope, so we can apply~\eqref{eq:ce-autocorr}.
Suppose $T_r=T_s$, so that for $mT_r \le t,t'< (m+1)T_r$,
we have  $\phi_\Delta(t,t')=0$ and
\begin{align}
   & \bar{A}(t,t') = A(t,t') \nonumber \\
   & = |S|^2 \left[ T_R K \rho + P
      \left( S_R^2 + \frac{S_I^2}{1-\rho^2} \left( 1-\rho \right)^2 \right) \right] \nonumber \\
   & \qquad \exp\left( - \gamma \frac{T_I}{\sqrt{1-\rho^2}} 2P (1 - \rho) \right).
   \label{eq:ce-autocorr-app}
\end{align}
Note that \eqref{eq:ce-autocorr-app} is real-valued.
Using \eqref{eq:Erbar-def} rather than \eqref{eq:PrT-bound}, we thus have
(cf.~\eqref{eq:autocorr-bound})
\begin{align}
   & \bar{E}_m(T_r) \le   \frac{1}{T_r} 
   \left[ K z + P \left (1 + \gamma^2 K^2 z^4 \right) \right] \nonumber \\
   & \quad \int_{0}^{T_r} \int_{0}^{T_r} \exp\left( - \gamma \frac{T_I}{\sqrt{1-\rho^2}} 2P (1 - \rho) \right)
   \, dt' \, dt .
   \label{eq:Tr-bound2}
\end{align}

Consider $T_r = 1/B$ and $\gamma (K/2) z^2 \le 1$. We
use \eqref{eq:gTI-bound} and \eqref{eq:sinc-bounds2}
to upper bound the exponential of \eqref{eq:Tr-bound2} with
\begin{align}
   \exp\left( - (\kappa/2) \, P \, B^2 (t-t')^2 \right) \label{eq:region-bounds2}
\end{align}
for the range of interest with $B|t-t'| \le 1$. We thus have
\begin{align}
   & \bar{E}_m(T_r) \le \frac{1}{T_r} \left[ K z + P \left (1 + \gamma^2 K^2 z^4 \right) \right] \nonumber \\
   & \int_{0}^{T_r} \frac{1}{\sigma} \frac{\sqrt{\pi}}{2} \left[  {\rm erf}\left( \sigma t \right)
       + {\rm erf}\left( \sigma (T_r-t) \right) \right] dt
   \label{eq:Tr-bound3}
\end{align}
where $\sigma=\sqrt{(\kappa/2) P} B$. Evaluating the integral gives
\begin{align}
   & \bar{E}_m(T_r)
   \le \left[ K z + P \left (1 + \gamma^2 K^2 z^4 \right) \right] 
   \frac{2 T_r}{\sqrt{(\kappa/2) P}} \frac{\sqrt{\pi}}{2} \nonumber \\
   & \left[ {\rm erf}\left( \sqrt{(\kappa/2) P} \right)
   - \frac{1}{\sqrt{(\kappa/2) P}} \frac{1}{\sqrt{\pi}} \left(1 - e^{-(\kappa/2) P} \right) \right].
   \label{eq:Tr-bound-app}
\end{align}
As $\gamma \rightarrow 0$ we have $\kappa \rightarrow 0$, and we find that
the RHS of \eqref{eq:Tr-bound-app} becomes $(Kz + P) T_r$, as expected for
a linear channel.

%=============================================================
\appendix{Various Bounds}
\label{app:simple-bounds}

\begin{comment}
%-------
\subsection*{Hyperbolic Functions}
%
The product form for $\cosh(\cdot)$ is \cite[Sec.~1.43]{Gradshteyn-Ryzhik-91}
\begin{align}
  \cosh(y) = \prod_{k=0}^{\infty} \left( 1 + \frac{4y^2}{(2k+1)^2 \pi^2} \right).
\end{align}
Inserting $y=\sqrt{2c}z=\sqrt{-2jx}$, we thus have
\begin{align}
  \frac{1}{S(c)} = \prod_{k=0}^{\infty} \left( 1 - \frac{8jx}{(2k+1)^2 \pi^2} \right)
\end{align}
which gives $|S(c)| \le 1$, and therefore $S_R(c)\le 1$.
\end{comment}

%-------
\vspace{-7mm}
\subsection*{Sinc Function}
We have the following bounds, see Fig.~\ref{fig:SimpleBounds}:
\begin{align}
|\sinc(y)| = \left| \frac{\sin(\pi y)}{\pi y} \right| & \le \left\{ \begin{array}{ll}
1 - y^2, & |y|\le 1 \\
1/4, & |y|>1
\end{array} \right. \label{eq:sinc-bounds1}  \\
\sinc(y)^2 & \le \left\{ \begin{array}{ll}
1 - y^2, & |y|\le 1 \\
1/20, & |y|>1
\end{array} \right. \label{eq:sinc-bounds2} \\
\sinc(y)^2 & \ge 1 - 4 y^2 . \label{eq:sinc-bounds2a}
\end{align}

%-------
\subsection*{Exponential Function}
We use two bounds on the exponential function with $a>0$:
\begin{align}
& y e^{-a y} \le \frac{1}{ae} \; \text{ with equality if } y=\frac{1}{a}
\label{eq:xe-bound1} \\
& y e^{-a y^2} \le \frac{1}{\sqrt{2ae}} \; \text{ with equality if } y=\frac{1}{\sqrt{2a}}.
\label{eq:xe-bound2} 
\end{align}

%%%%%%%%%%%%%%%%%%%%%%%%%%%%%%%%%%%%%%%%
\begin{figure}[t!]
  \centerline{\includegraphics[scale=0.48]{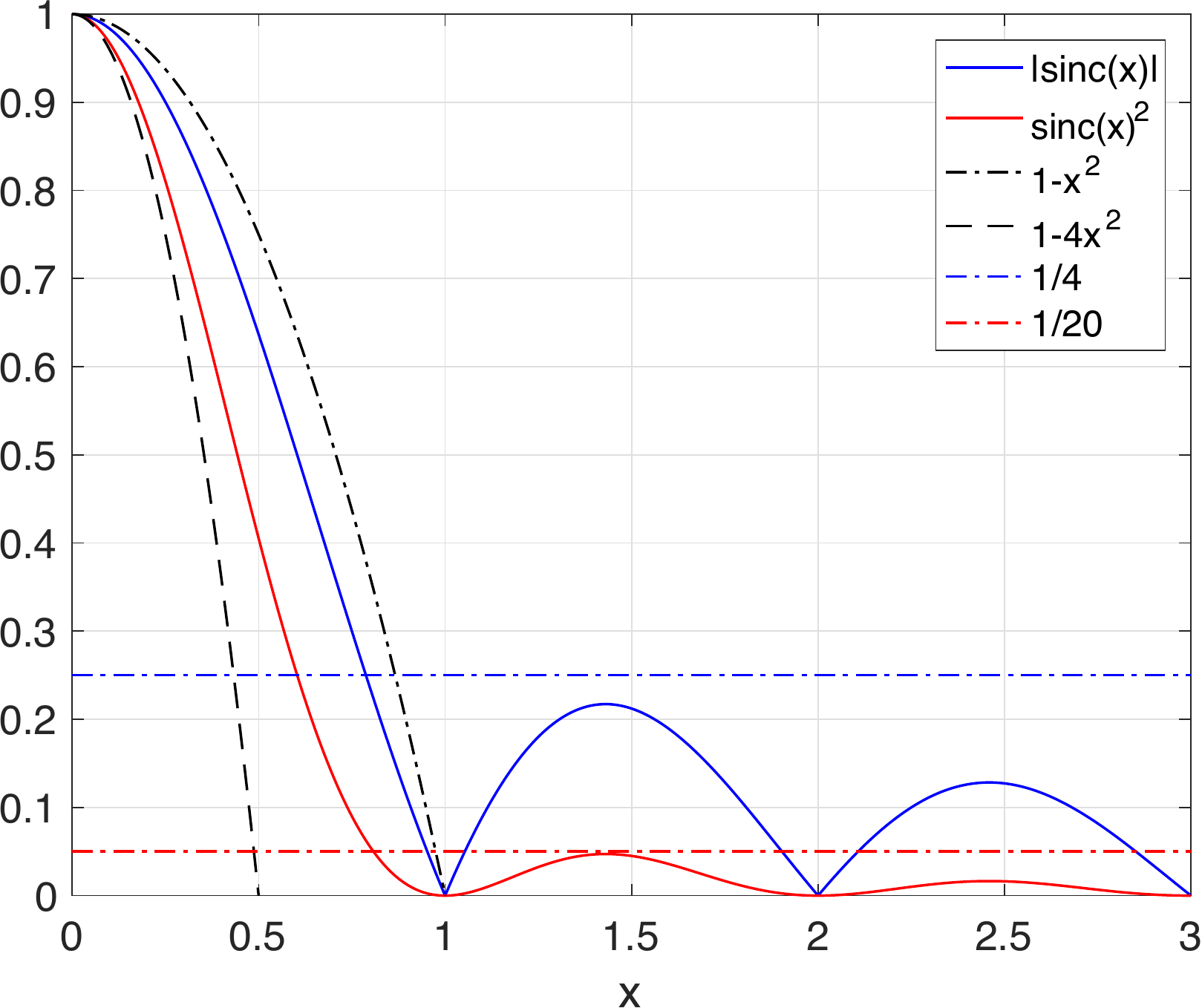}}
  \caption{Simple bounds on $|\sinc(x)|$ and $\sinc(x)^2$.}
  \label{fig:SimpleBounds}
\end{figure}
%%%%%%%%%%%%%%%%%%%%%%%%%%%%%%%%%%%%%%%%

%-------
\subsection*{Concavity of a Special Function}
Consider the function
\begin{align}
   f(P) = \left( a+b\sqrt{P}+cP \right) \frac{\sqrt{\pi}}{2} \, \frac{{\rm erf}\left(\sqrt{s P} \right)}{\sqrt{s P}}
\label{eq:concave-function}
\end{align}
where $a$, $b$, and $c$ are non-negative constants and $s$ is a positive constant. We compute
\begin{align}
  \frac{df}{dP} =
  & \frac{e^{-sP}}{2 s P^{3/2}} \left[ (cP+a) s \sqrt{P} + b s P \right. \nonumber \\
  & \qquad \left. + \sqrt{s} \, (cP-a) \, e^{sP} \frac{\sqrt{\pi}}{2} {\rm erf}\left(\sqrt{sP} \right) \right]
  \label{eq:concave-deriv1} \\
  \frac{d^2f}{dP^2}  =
  & - \frac{e^{-sP}}{4 P^2 \sqrt{sP}} \left[ \sqrt{s}P \left\{ 2s \sqrt{P} (cP+a) + 2bsP + b \right\} \right. \nonumber \\
  & \left. + (cP-3a) \left\{ e^{sP} \frac{\sqrt{\pi}}{2} {\rm erf}\left(\sqrt{sP}\right) - \sqrt{sP} \right\} \right].
  \label{eq:concave-deriv2}
\end{align}
From \eqref{eq:concave-deriv1}, we see that $f(P)$ is non-decreasing if $cP\ge a$.
Similarly, for \eqref{eq:concave-deriv2} we use~\cite[8.253.1]{Gradshteyn-Ryzhik-91} to bound
\begin{align}
  e^{y^2} \cdot \frac{\sqrt{\pi}}{2} {\rm erf}(y) \ge y
  \label{eq:concave-lower}
\end{align}
for $y\ge 0$ and find that $f(P)$ is concave if $cP\ge 3a$. Thus,
$f(P+3a/c)$ is both non-decreasing and concave if $P\ge0$ and $c>0$.

%=============================================================
%=============================================================
%%%%%%%%%%%%%%%%%%%%%%%%%%%%%%%%%%%%%%%%
\renewcommand{\arraystretch}{1.05}
\begin{table*}[th]
\begin{center}
  \caption{List of Acronyms, Symbols, and Notation}
  \label{table:notation}
  \begin{tabular}{| l | l | l | }
    \hline %\hline
    \hline
    \multicolumn{2}{|l|}{\bf Acronyms} & {\bf Defined in:} \\
    \hline
    ASK, FSK, PSK & Amplitude-, frequency-, phase-shift keying & Sec.~\ref{subsec:receiver-noise} and~\ref{subsec:ring-modulation} \\ \hline
    AWGN & Additive white Gaussian noise & Sec.~\ref{sec:intro} \\ \hline
    EDFA  & Erbium-doped fiber amplifier & Sec.~\ref{subsec:OA-noise-model} \\ \hline
    %GVD & Group velocity dispersion & \\ \hline
    LHS, RHS  & Left-hand side, right-hand side & Sec.~\ref{subsec:hyperbolic-functions} and~\ref{subsec:timelimited-receiver} \\ \hline
    NLSE & Nonlinear Schr\"odinger equation & Sec.~\ref{subsec:per-sample-vs-multi-sample} and~\ref{subsec:NLSE} \\ \hline
    OA & Optical amplification / optical amplifier & Sec.~\ref{sec:intro} \\ \hline
    PAM & Pulse amplitude modulation & Sec.~\ref{subsec:cyclostationary} \\ \hline
    PDD, PSD, PSDD & Power distance density, power spectral density, power spectral-distance density & Sec.~\ref{sec:intro} and~\ref{subsec:OA-noise-model} \\ \hline
    %QAM & Quadrature amplitude modulation & App.~\ref{app:nonlinear-capacity} \\ \hline
    SNR & Signal-to-noise ratio & Sec.~\ref{subsec:AWGN} \\ \hline
    SPM, XPM & Self-phase modulation, cross-phase modulation & Sec.~\ref{sec:auto} and~\ref{subsec:autocorr-bounds} \\ \hline
    WDM & Wavelength division multiplexing & Sec.~\ref{sec:intro} \\ \hline
    \hline
    \multicolumn{3}{|l|}{\bf Fiber and Noise Parameters} \\ \hline
    $\alpha$, $\beta_1$, $\beta_2$ & Loss coefficient, group velocity, group velocity dispersion & Sec.~\ref{subsec:NLSE} \\ \hline
    $\gamma$ & Nonlinear Kerr coefficient & Sec.~\ref{subsec:NLSE} \\ \hline
    %$\gamma$, $\gamma_p$, $\gamma_s$ & Nonlinear Kerr coefficients & Sec.~\ref{subsec:NLSE}, App.~\ref{app:raman} \\ \hline
    %$\delta_m$ & Fraction related to molecular vibrations & App.~\ref{app:raman} \\ \hline
    %$g_p(t)$, $g_s(t)$, $h_p(t)$, $h_s(t)$ & Filters related to the third-order nonlinear susceptibility and noise forces & App.~\ref{app:raman} \\ \hline
    $k_B$, $T_e$ & Boltzmann's constant, temperature & Sec.~\ref{subsec:AWGN} \\ \hline
    %Spontaneous emission factor $n_{\rm sp}$ & 1 \\ \hline
    %Raman pump frequency $f_p$ & 206.75 THz ($\lambda_s=1450$ nm) \\ \hline
    $N_0=k_B T_e$ & receiver noise PSD & Sec.~\ref{subsec:AWGN} \\ \hline
    $N_A$ & OA noise PSDD & Sec.~\ref{subsec:OA-noise-model} \\ \hline
    $K=N_A B$ & OA noise PDD & Sec.~\ref{subsec:OA-noise-model} \\ \hline
    $\rho(t-t')$ & OA noise correlation coefficient & Sec.~\ref{subsec:OA-noise-model} \\ \hline
    \hline
    \multicolumn{3}{|l|}{\bf Signal Parameters and Variables} \\
    \hline
    $B$ & OA bandwidth & Sec.~\ref{subsec:per-sample-vs-multi-sample} and~\ref{subsec:OA-noise-model} \\ \hline
    $W$ & Signal or receiver bandwidth & Sec.~\ref{subsec:per-sample-vs-multi-sample},~\ref{subsec:bandlimited-receiver-defn} and~\ref{subsec:propagating-bandwidth} \\ \hline
    $f_0$  & Carrier frequency & Sec.~\ref{subsec:NLSE} \\ \hline
    $T_s$ & Time period of PAM & Sec.~\ref{subsec:cyclostationary} \\ \hline
    $P$, $P_t$, $P_o$ & Average launch power, instantaneous launch power, offset power & Sec.~\ref{sec:intro},~\ref{subsec:bandlimited-receiver} and~\ref{subsec:average-power} \\ \hline
    \rule{0pt}{3mm}$c$  & $-j \gamma (K/2) \sqrt{1-\rho^2}$ & Sec.~\ref{subsec:hyperbolic-functions} and \eqref{eq:c-def} \\ \hline
    \rule{0pt}{3mm}$\kappa$ & $2 \gamma^2 K z^3/3$ & Sec.~\ref{subsec:low-noise} and \eqref{eq:kappa} \\ \hline
    $\delta$ & Variable used for bounding & Sec.~\ref{subsec:autocorr-bounds} and \eqref{eq:STratio} \\ \hline
    $S(c)$, $S_R(c)$, $S_I(c)$ &  $\sech\left( \sqrt{2c} \, z\right)$, real part $\Re(S(c))$, imaginary part $\Im(S(c))$ & Sec.~\ref{subsec:hyperbolic-functions} \\ \hline
    $T(c)$, $T_R(c)$, $T_I(c)$ &  $\left. \tanh\left( \sqrt{2c} \, z\right)\right/\sqrt{2c}$, real part $\Re(T(c))$, imaginary part $\Im(T(c))$ & Sec.~\ref{subsec:hyperbolic-functions} \\ \hline
    \hline
    \multicolumn{3}{|l|}{\bf Signals} \\
    \hline
    $u(z,t)=u_z(t)=u(t)=u$ & Signal at distance $z$ and time $t$ & Sec.~\ref{subsec:notation} \\ \hline
    \rule{0pt}{3mm}$u(z,t')=u_z(t')=u(t')=u'$ & Signal at distance $z$ and time $t'$ & Sec.~\ref{subsec:notation} \\ \hline
    $U(z,t)=U_z(t)=U(t)=U$ & Random signal at distance $z$ and time $t$ & Sec.~\ref{subsec:notation} \\ \hline
    \rule{0pt}{3mm}$U(z,t')=U_z(t')=U(t')=U'$ & Random signal at distance $z$ and time $t'$ & Sec.~\ref{subsec:notation} \\ \hline
    \rule{0pt}{3mm}$\tilde{u}(z,f)$ &  Fourier transform of $u(z,t)$ at distance $z$ & Sec.~\ref{subsec:linear-model} \\ \hline
    %$u_p(z,t)$, $u_s(z,t)$ & Pump and source signals & App.~\ref{app:raman} \\ \hline
    $u_r(t)$ & Receiver signal with AWGN & Sec.~\ref{subsec:AWGN} \\ \hline
    $N_r(\cdot)$ and $n_r(\cdot)$ & Receiver AWGN process and its realization & Sec.~\ref{subsec:AWGN} \\ \hline
    $W(z,t)$, $W_R(z,t)$, $W_I(z,t)$ & Spatial Wiener process, real part $\Re(W(z,t))$, imaginary part $\Im(W(z,t))$ & Sec.~\ref{subsec:OA-noise-model} \\ \hline
    $\hat{w}(z,t)$ & Accumulated noise for the model~\eqref{eq:NLSE-dispersion-free} & Sec.~\ref{subsec:dispersion-free-model} and~\eqref{eq:dispersion-free-model-noise} \\ \hline
    \rule{0pt}{3mm}$\left\{\phi_{m}(\cdot)\right\}_{m=1}^{\infty}$ & Complete orthonormal basis & Sec.~\ref{subsec:receiver-noise} \\[1mm] \hline
    \rule{0pt}{3mm}$b(t)$ and $\tilde{b}(f)$ & $2W \, \sinc\left( W t \right)^2$ and its Fourier transform & Sec.~\ref{subsec:bandlimited-receiver-defn} \\ \hline
    \hline
    \multicolumn{3}{|l|}{\bf Autocorrelation Functions} \\
    \hline
   % $\mu_{m,n}$, $M_{m,n}(s_1,s_2)$ & moments and moment-generating functions for one sample & App.~\ref{app:one-sample-statistics} \\ \hline
    %$\mu_{mnk\ell}$, $M_{mnk\ell}({\bf s})$ & moments and moment-generating functions for two samples & App.~\ref{app:two-sample-statistics} \\ \hline
    \rule{0pt}{3mm}$A_z(t,t')=A(t,t')$ & Autocorrelation function at distance $z$ conditioned on $U_0=u_0$ & Sec.~\ref{subsec:autocorr-psd} \\ \hline
    \rule{0pt}{3mm}$\bar{A}_z(t,t')=\bar{A}(t,t')$ & Average autocorrelation function at distance $z$ & Sec.~\ref{subsec:autocorr-psd} \\ \hline
    \rule{0pt}{3mm}$\bar{A}(\tau)$ & Time-averaged autocorrelation function at distance $z$ &Sec.~\ref{subsec:cyclostationary} \\ \hline
    \rule{0pt}{3mm}$\mathcal{A}(t,t')$ & Approximate autocorrelation function at distance $z$ conditioned on $U_0=u_0$ & Sec.~\ref{subsec:low-noise} \\ \hline
    \rule{0pt}{3mm}$\bar{\mathcal{A}}(t,t')$, $\bar{\mathcal{A}}(\tau)$ & Average approximate autocorrelation functions at distance $z$ & Sec.~\ref{subsec:ring-modulation} \\ \hline
    \hline
    \multicolumn{3}{|l|}{\bf PSD, Receiver Power, Receiver Energy} \\
    \hline
    \rule{0pt}{3mm}$\bar{\mathcal{P}}(f)$, $\bar{\mathcal{P}}(f,T)$ & PSDs at distance $z$ & Sec.~\ref{subsec:autocorr-psd} \\ \hline
    \rule{0pt}{3mm}$\bar{P}_r(W)$, $\bar{P}_r(W,T)$ & Average receiver powers in a band of bandwidth $W$ & Sec.~\ref{subsec:AWGN} and~\ref{subsec:bandlimited-receiver-defn} \\ \hline
    $P_r(W,T,t)$ & Upper bound on instantaneous receiver power conditioned on $U_0=u_0$ & Sec.~\ref{subsec:bandlimited-receiver} \\ \hline
    $E_m(T_r)$ & Energy at receiver in time interval $[mT_r,(m+1)T_r)$ conditioned on $U_0=u_0$ & Sec.~\ref{subsec:timelimited-receiver} \\ \hline
    \rule{0pt}{3mm}$\bar{E}_m(T_r)$ & Average energy at receiver in time interval $[mT_r,(m+1)T_r)$ & Sec.~\ref{subsec:timelimited-receiver} \\ \hline
    $E_m(T_r,t)$ & Upper bound on instantaneous receiver energy conditioned on $U_0=u_0$ & Sec.~\ref{subsec:time-resolution-limited} \\ \hline
    \hline
    \multicolumn{3}{|l|}{\bf Capacity} \\
    \hline
    $C(W)$ & Capacity with bandwidth $W$ & Sec.~\ref{subsec:AWGN} \\ \hline
    $\eta(W)$ & Spectral efficiency with bandwidth $W$ & Sec.~\ref{subsec:AWGN} \\ \hline
    \hline %\hline
  \end{tabular}
\end{center}
\end{table*}
%%%%%%%%%%%%%%%%%%%%%%%%%%%%%%%%%%%%%%%%

%=============================================================
%=============================================================
\section*{Acknowledgments}
\label{sec:acks}
The author wishes to thank R.-J. Essiambre, J. Garc\'ia, and P. Schulte for comments,
corrections, and suggestions. The author also wishes to thank the Associate Editor and
the reviewers for suggestions that improved the presentation.

%=============================================================
%\bibliographystyle{IEEE} %{plain} {usrt}
%\bibliography{Mecozzi-94} %{ITgeneral,mu}

\begin{thebibliography}{10}

\bibitem{ekwfg-JLT10}
R.~J. Essiambre, G.~Kramer, P.~J. Winzer, G.~J. Foschini, and B.~Goebel,
\newblock ``Capacity limits of optical fiber networks,''
\newblock {\em J.\ Lightwave Techn.}, vol. 28, no. 4, pp. 662--701, Feb. 2010.

\bibitem{Mecozzi-94}
A.~Mecozzi,
\newblock ``Limits to long-haul coherent transmission set by the {K}err
  nonlinearity and noise of the in-line amplifiers,''
\newblock {\em J.\ Lightwave Techn.}, vol. 12, no. 11, pp. 1993--2000, Nov.
  1994.

\bibitem{Turitsyn-Derevyanko-Yurkevich-Turitsyn-PRL03}
K.~S. Turitsyn, S.~A. Derevyanko, I.~V. Yurkevich, and S.~K. Turitsyn,
\newblock ``Information capacity of optical fiber channels with zero average
  dispersion,''
\newblock {\em Phys.\ Rev.\ Lett.}, vol. 91, no. 20, pp. 203901(1--4), Nov.
  2003.

\bibitem{Yousefi-Kschischang-IT11}
M.~I. Yousefi and F.~R. Kschischang,
\newblock ``On the per-sample capacity of nondispersive optical fibers,''
\newblock {\em IEEE Trans.\ Inf.\ Theory}, vol. 57, no. 11, pp. 7522--7541,
  Nov. 2011.

\bibitem{Mecozzi-04}
A.~Mecozzi,
\newblock ``Probability density functions of the nonlinear phase noise,''
\newblock {\em Opt.\ Lett.}, vol. 29, no. 7, pp. 673--675, Apr. 2004.

\bibitem{Terekhov-etal-A16}
I.~S. Terekhov, A.~V. Reznichenko, Y.~A. Kharkov, and S.~K. Turitsyn,
\newblock ``Optimal input signal distribution and per-sample mutual information
  for nondispersive nonlinear optical fiber channel at large {SNR},''
\newblock {\em CoRR}, vol. abs/1508.05774, 2016.

\bibitem{Fahs-etal-A17}
J.~Fahs, A.~Tchamkerten, and M.~I. Yousefi,
\newblock ``Capacity-achieving input distributions in nondispersive optical
  fibers,''
\newblock {\em CoRR}, vol. abs/1704.04904, 2017.

\bibitem{Tang-JLT01a}
J.~Tang,
\newblock ``The {S}hannon channel capacity of dispersion-free nonlinear optical
  fiber transmission,''
\newblock {\em J.\ Lightwave Techn.}, vol. 19, no. 8, pp. 1104--1109, Aug.
  2001.

\bibitem{Tang-JLT01b}
J.~Tang,
\newblock ``The multispan effects of {K}err nonlinearity and amplifier noises
  on {S}hannon channel capacity of a dispersion-free nonlinear optical fiber,''
\newblock {\em J.\ Lightwave Techn.}, vol. 19, no. 8, pp. 1110--1115, Aug 2001.

\bibitem{Pinsker-64}
M.~S. Pinsker,
\newblock {\em Information and Information Stability of Random Variables and
  Processes},
\newblock Holden Day, 1964.

\bibitem{Wei-Plant-A06}
H.~Wei and D.~V. Plant,
\newblock ``Comment on "{I}nformation capacity of optical fiber channels with
  zero average dispersion",''
\newblock {\em CoRR}, vol. abs/physics/0611190, 2006.

\bibitem{Wyner-BSTJ66}
A.~D. Wyner,
\newblock ``The capacity of the band-limited {G}aussian channel,''
\newblock {\em Bell Sys.\ Tech.\ J.}, vol. 45, no. 3, pp. 359--395, Mar. 1966.

\bibitem{Proakis-Salehi-5}
J.~G. Proakis and M.~Salehi,
\newblock {\em Digital Communications},
\newblock McGraw Hill, 5th edition, 2008.

\bibitem{Shannon48}
C.~E. Shannon,
\newblock ``A mathematical theory of communication,''
\newblock {\em Bell Sys.\ Tech.\ J.}, vol. 27, pp. {379--423 and 623--656},
  July and October 1948,
\newblock {Reprinted in {\it Claude Elwood Shannon: Collected Papers}, pp.
  5-83, (N.J.A. Sloane and A.D. Wyner, eds.) Piscataway: IEEE Press, 1993.}

\bibitem{Gastpar-IT07}
M.~Gastpar,
\newblock ``On capacity under receive and spatial spectrum-sharing
  constraints,''
\newblock {\em IEEE Trans.\ Inf.\ Theory}, vol. 53, no. 2, pp. 471--487, Feb.
  2007.

\bibitem{Agrawal-03}
G.~P. Agrawal,
\newblock {\em Nonlinear Fiber Optics},
\newblock Academic Press, 3rd edition, 2001.

\bibitem{Gallager68}
R.~G. Gallager,
\newblock {\em Information Theory and Reliable Communication},
\newblock Wiley, New York, 1968.

\bibitem{Slepian76}
D.~Slepian,
\newblock ``On bandwidth,''
\newblock {\em Proc. IEEE}, vol. 64, pp. 292--300, Mar. 1976.

\bibitem{Goebel-etal-IT11}
B.~Goebel, R.~J. Essiambre, G.~Kramer, P.~J. Winzer, and N.~Hanik,
\newblock ``Calculation of mutual information for partially coherent {G}aussian
  channels with applications to fiber optics,''
\newblock {\em IEEE Trans.\ Inf.\ Theory}, vol. 57, no. 9, pp. 5720--5736,
  Sept. 2011.

\bibitem{Barletta-Kramer-CROWNCOM14}
L.~Barletta and G.~Kramer,
\newblock ``Signal-to-noise ratio penalties for continuous-time phase noise
  channels,''
\newblock in {\em Int. Conf. Cognitive Radio Oriented Wireless Networks
  Commun.}, June 2014, pp. 232--235.

\bibitem{Barletta-Kramer-ISIT14}
L.~Barletta and G.~Kramer,
\newblock ``On continuous-time white phase noise channels,''
\newblock in {\em IEEE Int. Symp. Inf. Theory}, {Honolulu, HI}, June 2014, pp.
  2426--2429.

\bibitem{Kramer-etal-15}
G.~Kramer, M.~I. Yousefi, and F.~R. Kschischang,
\newblock ``Upper bound on the capacity of a cascade of nonlinear and noisy
  channels,''
\newblock in {\em IEEE Inf. Theory Workshop}, Apr. 2015, pp. 1--4.

\bibitem{Yousefi-etal-15}
M.~I. Yousefi, G.~Kramer, and F.~R. Kschischang,
\newblock ``Upper bound on the capacity of the nonlinear {S}chr{\"o}dinger
  channel,''
\newblock in {\em IEEE Can. Workshop Inf. Theory}, July 2015, pp. 22--26.

\bibitem{Headly-Agrawal-JQC95}
C.~Headly III and G.~P. Agrawal,
\newblock ``Noise characteristics and statistics of picosecond {S}tokes pulses
  generated in optical fibers through stimulated {R}aman scattering,''
\newblock {\em J. Quant. Electron.}, vol. 31, no. 11, pp. 2058--2067, Nov.
  1995.

\bibitem{Cameron-Martin-45}
R.~H. Cameron and W.~T. Martin,
\newblock ``Evaluation of various {W}iener integrals by use of certain
  {S}turm-{L}iouville differential equations,''
\newblock {\em Bull.\ Amer.\ Math.\ Soc.}, vol. 51, pp. 73--90, Feb. 1945.

\bibitem{Cameron-Martin-45a}
R.~H. Cameron and W.~T. Martin,
\newblock ``Transformations of {W}iener integrals under a general class of
  linear transformations,''
\newblock {\em Trans. Amer. Math. Soc.}, vol. 58, pp. 184--219, Sept. 1945.

\bibitem{Thangaraj-Kramer-Boecherer-IT17}
A.~Thangaraj, G.~Kramer, and G.~B\"ocherer,
\newblock ``Capacity bounds for discrete-time, amplitude-constrained, additive
  white {G}aussian noise channels,''
\newblock {\em IEEE Trans.\ Inf.\ Theory}, vol. 63, 2017.

\bibitem{Gradshteyn-Ryzhik-91}
I.~S. Gradshteyn and I.~M. Ryzhik,
\newblock {\em Table of Integrals, Series, and Products},
\newblock Academic Press, 7th edition, 2007.

\end{thebibliography}
%\end{document}
%========================

%=============================================================
%=============================================================
\begin{IEEEbiographynophoto}{\bf Gerhard Kramer} (S'91-M'94-SM'08-F'10) received the Dr. sc. techn. degree from ETH Zurich in 1998. From 1998 to 2000, he was with Endora Tech AG in Basel, Switzerland, and from 2000 to 2008 he was with the Math Center at Bell Labs in Murray Hill, NJ, USA. He joined the University of Southern California, Los Angeles, CA, USA, as a Professor of Electrical Engineering in 2009. He joined the Technical University of Munich (TUM) in 2010, where he is currently Alexander von Humboldt Professor and Chair of Communications Engineering. His research interests include information theory and communications theory, with applications to wireless, copper, and optical fiber networks. Dr. Kramer served as the 2013 President of the IEEE Information Theory Society.
\end{IEEEbiographynophoto}

\end{document}